\documentclass[11pt]{article}
\usepackage[margin=1in]{geometry}
\usepackage{amsmath,amssymb,amsthm}
\usepackage{algorithm}
\usepackage{algpseudocode}
\usepackage{enumitem}
\usepackage{hyperref}
\usepackage{xcolor}
\usepackage{booktabs}
\usepackage{graphicx}
\usepackage{tikz}
\usetikzlibrary{decorations.pathreplacing}
\usepackage{caption}
\usepackage{float}

\newtheorem{theorem}{Theorem}[section]
\newtheorem{lemma}[theorem]{Lemma}
\newtheorem{proposition}[theorem]{Proposition}
\newtheorem{corollary}[theorem]{Corollary}
\newtheorem{definition}[theorem]{Definition}
\newtheorem{remark}[theorem]{Remark}
\newtheorem{example}[theorem]{Example}

\newtheorem{conjecture}[theorem]{Conjecture}

\title{\textbf{Combinatorial Privacy: Private Multi-Party Bitstream\\Grand Sum by Hiding in Birkhoff Polytopes}}
\author{
Praneeth Vepakomma\\[4pt]
Mohamed bin Zayed University of Artificial Intelligence (MBZUAI)\\
Massachusetts Institute of Technology (MIT)\\
\texttt{vepakom@mit.edu}
}
\date{}

\begin{document}
\maketitle

\begin{abstract}
We introduce PolyVeil, a protocol for private aggregation across $k$ clients that encodes private bits as permutation matrices in the Birkhoff polytope. A two-layer architecture gives the server perfect simulation-based security (statistical distance zero) while a separate aggregator faces \#P-hard likelihood inference via the permanent and mixed discriminant.

We develop DP analyses under multiple frameworks (Berry--Esseen, R\'{e}nyi, f-DP). In the full variant, where the aggregator sees a doubly stochastic matrix per client, the DP guarantee is non-vacuous only when the signal is undetectable. In the compressed variant, where the aggregator sees a scalar, f-DP gives $\varepsilon \approx 7.8$ per client. Shuffle-model amplification then yields $\varepsilon \approx 0.37$ for $k = 1{,}000$ clients with no accuracy loss, since the aggregator needs only the sum of the shuffled scalars. This exposes a tension between \#P-hardness (requiring the matrix view) and strong DP (requiring the shuffled scalar view).

For the Boolean sum alone, additive secret sharing dominates. The Birkhoff encoding's advantage is multi-statistic extraction from a single matrix, enabling per-bit marginals and weighted sums without further client interaction. The protocol needs no PKI and outputs exact aggregates.
\end{abstract}

\tableofcontents
\newpage

\section{Introduction}
\label{sec:intro}

Computing aggregate statistics over private data held by many parties is a foundational problem in privacy-preserving computation. A concrete and widely applicable instance is the Boolean sum problem, in which $k$ clients each hold a private binary string of length $n$ and a server wishes to learn the total number of ones across all $kn$ bits without learning any individual client's data. Applications range from epidemiological surveillance, where a health authority counts positive test results without accessing individual diagnoses, to federated analytics, where a service provider tallies binary feature occurrences without centralizing user data.

Existing approaches to this problem broadly fall into three categories. Generic secure multi-party computation (MPC) protocols, built on garbled circuits or secret sharing, provide strong composable security with guarantees based on specific computational hardness assumptions (such as the difficulty of integer factorization or the Learning with Errors problem over lattices), but impose communication and computation costs that scale with circuit complexity and become prohibitive when $k$ is large or network conditions are constrained. Homomorphic encryption allows computation on ciphertexts, with security based on assumptions such as the composite residuosity problem (Paillier~\cite{paillier1999}) or the Learning with Errors problem (lattice-based FHE~\cite{gentry2009}), but carries substantial per-operation cost and requires careful key management. Differential privacy provides a framework for releasing aggregate statistics to untrusted parties by adding calibrated noise, offering formal $(\varepsilon,\delta)$-guarantees that hold regardless of an adversary's computational power, but inherently sacrificing accuracy for privacy.

PolyVeil occupies a distinct point in this design space, one that we argue represents a new paradigm we call \emph{Combinatorial Privacy}. In PolyVeil, the core idea is to encode each client's private bitstream as a permutation matrix, embed it inside a doubly stochastic matrix by mixing it with random decoy permutations, and use secure aggregation to recover the aggregate bit sum exactly. The Birkhoff--von Neumann theorem guarantees that every doubly stochastic matrix admits many decompositions into convex combinations of permutation matrices, and this non-uniqueness forms one of the two security layers of the protocol.

We present a rigorous security analysis that identifies a fatal vulnerability in naive implementations of this approach (the de-shuffling attack, which recovers all individual data with probability 1) and then develops a corrected \emph{two-layer protocol} that achieves provable security. In the corrected protocol, the main server receives only aggregate scalars and is information-theoretically secure (its view is identically distributed for any two inputs with the same aggregate). A separate aggregator entity receives Birkhoff-encoded matrices but not the noise values, and faces a computational barrier: recovering the private permutation matrix from its encoding requires evaluating the density of a random doubly stochastic matrix at a given point, which we prove is \#P-hard via a reduction to the permanent. The two layers compose so that no single entity can learn individual data: the server lacks the information, and the aggregator lacks the computational power.

This two-layer architecture distinguishes Combinatorial Privacy from existing paradigms. Unlike MPC and HE, which derive security from number-theoretic hardness (factoring, LWE), PolyVeil's computational layer derives from the \#P-hardness of evaluating Birkhoff polytope decomposition likelihoods (proved for likelihood-based attacks; conjectured for all attacks). Unlike DP, PolyVeil produces exact answers. The protocol requires no public-key infrastructure and has $O(k)$ communication in the compressed variant.

\subsection*{High-Level Framework}

Figure~\ref{fig:pipeline} illustrates the two-layer protocol at a high level. Each of the $k$ clients holds a private binary string $\mathbf{b}_t \in \{0,1\}^n$. The protocol proceeds in three stages.

\textit{Encoding.}
Each client encodes its bit vector $\mathbf{b}_t$ as a block-diagonal permutation matrix $M_t \in \{0,1\}^{2n \times 2n}$ and masks it by forming the doubly stochastic matrix $D_t = \alpha^* M_t + (1-\alpha^*) R_t$, where $R_t$ is a random convex combination of decoy permutation matrices. The client also computes the scalar $f_t = \alpha^* s_t + \eta_t$ (where $s_t = \sum_j b_{t,j}$ is the bit count and $\eta_t$ is the noise from the decoys) and the noise value $\eta_t$ separately.

\textit{Separation.}
The client sends $D_t$ (or $f_t$ in the compressed variant) to the aggregator, and $\eta_t$ to the noise aggregator. These two entities do not communicate with each other. The aggregator computes $F = \sum_t \mathbf{w}^T D_t \mathbf{y}$ (the aggregate of the bilinear extractions). The noise aggregator computes $H = \sum_t \eta_t$ (the aggregate noise). Both scalars are sent to the server.

\textit{Recovery.}
The server computes $S = (F - H)/\alpha^*$, recovering the exact Boolean sum. The noise cancels algebraically: $F - H = \sum_t (\alpha^* s_t + \eta_t) - \sum_t \eta_t = \alpha^* S$.

The key property is that no single entity sees enough to learn individual data. The server sees only $(F, H)$, which depends only on the aggregate $S$ (information-theoretic security). The aggregator sees $D_t$ but not $\eta_t$, so it cannot undo the noise cancellation; extracting $M_t$ from $D_t$ requires solving \#P-hard problems (computational security). The noise aggregator sees $\eta_t$ but not $D_t$ or $f_t$, so it learns nothing about $\mathbf{b}_t$.

A single matrix $D_t$ encodes the entire bit vector $\mathbf{b}_t$, not merely its sum. This means the aggregator can extract multiple statistics from the same data (per-bit marginals and weighted sums) without further client interaction, a capability that additive secret sharing does not provide for the same communication cost (Section~\ref{sec:multistat}).

\begin{figure}[!htbp]
\centering
\begin{tikzpicture}[
    box/.style={draw, rounded corners, minimum width=2.2cm, minimum height=0.9cm, align=center, font=\small},
    arrow/.style={->, >=stealth, thick},
    label/.style={font=\footnotesize, midway}
]

\node[box, fill=blue!10] (c1) at (0, 0) {Client 1\\$\mathbf{b}_1$};
\node[box, fill=blue!10] (c2) at (0, -1.5) {Client 2\\$\mathbf{b}_2$};
\node[font=\large] at (0, -2.7) {$\vdots$};
\node[box, fill=blue!10] (ck) at (0, -3.9) {Client $k$\\$\mathbf{b}_k$};

\node[box, fill=orange!15] (enc) at (4.2, -0.9) {Birkhoff\\Encoding\\$D_t = \alpha^* M_t$\\$+ (1{-}\alpha^*) R_t$};

\node[box, fill=green!12] (agg) at (8.5, 0) {Aggregator\\$F = \sum_t \mathbf{w}^T\! D_t \mathbf{y}$};
\node[box, fill=green!12] (nagg) at (8.5, -3.9) {Noise Agg.\\$H = \sum_t \eta_t$};

\node[box, fill=red!10] (srv) at (12.5, -1.95) {Server\\$S = \frac{F-H}{\alpha^*}$};

\draw[arrow] (c1) -- (enc);
\draw[arrow] (c2) -- (enc);
\draw[arrow] (ck) -- (enc);

\draw[arrow] (enc) -- (agg) node[label, above] {\footnotesize $D_t$ or $f_t$};
\draw[arrow] (enc) -- (nagg) node[label, below, sloped] {\footnotesize $\eta_t$};

\draw[arrow] (agg) -- (srv) node[label, above, sloped] {\footnotesize $F$};
\draw[arrow] (nagg) -- (srv) node[label, below, sloped] {\footnotesize $H$};

\draw[<->, >=stealth, thick, red!80] (agg.south) -- (nagg.north);
\draw[line width=2pt, red!80] ([xshift=-3pt, yshift=3pt] 8.5, -1.6) -- ([xshift=3pt, yshift=-3pt] 8.5, -2.3);
\draw[line width=2pt, red!80] ([xshift=3pt, yshift=3pt] 8.5, -1.6) -- ([xshift=-3pt, yshift=-3pt] 8.5, -2.3);
\node[font=\scriptsize\bfseries, red!80, fill=white, inner sep=2pt] at (9.8, -1.95) {non-colluding};

\end{tikzpicture}
\caption{The two-layer PolyVeil protocol. Each client encodes its private bit vector as a masked doubly stochastic matrix $D_t$ and sends it (or the scalar $f_t$) to the aggregator, and the noise $\eta_t$ to a separate noise aggregator. The two aggregators do not communicate (dashed line). The server receives only the aggregate scalars $F$ and $H$, from which it recovers $S$ exactly. The server has information-theoretic security (its view depends only on $S$). The aggregator faces \#P-hard inference (it sees $D_t$ but cannot efficiently extract $M_t$).}
\label{fig:pipeline}
\end{figure}
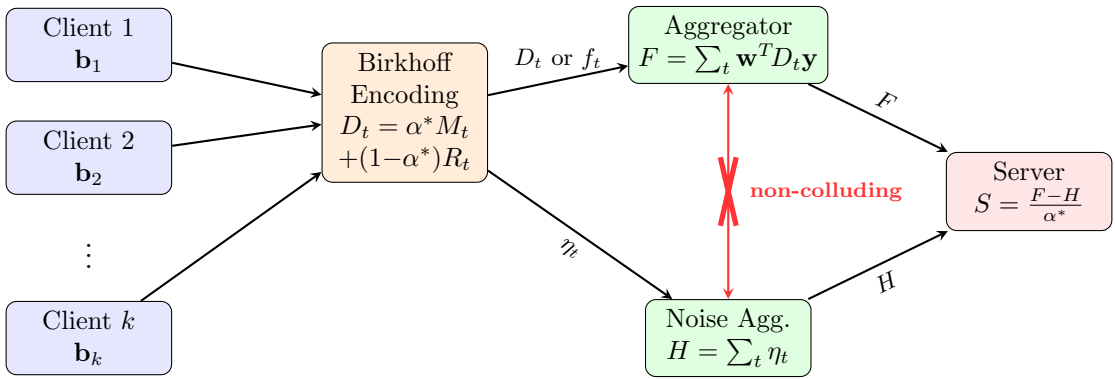

The remainder of this paper is organized as follows. Section~\ref{sec:prelim} establishes mathematical preliminaries. Section~\ref{sec:protocol} presents the weak protocol variants and the de-shuffling attack. Section~\ref{sec:security} provides security analysis. Section~\ref{sec:it_security} develops the secure two-layer protocol with proofs for both layers. Section~\ref{sec:multistat} derives multi-statistic extraction from the Birkhoff encoding and compares with additive secret sharing. Section~\ref{sec:resolve} proves a finite-sample $(\varepsilon, \delta)$-DP guarantee for the aggregator, analyzes the SNR regime where it is meaningful, and derives a non-vacuous $\varepsilon$ for the compressed two-layer protocol. Section~\ref{sec:conclusion} concludes. Appendix~\ref{app:sim} provides background on simulation-based proofs.

\section{Related Work}
\label{sec:related}

\paragraph{Secure multi-party computation.}
The problem of computing functions over distributed private inputs has been studied since the foundational work of Yao~\cite{yao1986} on garbled circuits and Goldreich, Micali, and Wigderson~\cite{gmw1987} on the GMW protocol. These generic constructions can compute any function securely, including Boolean sums, but their communication and computation costs scale with the circuit complexity of the target function. More recent frameworks such as SPDZ~\cite{damgard2012} reduce the online cost through preprocessing, but the per-gate overhead remains significant compared to the simple arithmetic in PolyVeil. Specialized secure aggregation protocols such as that of Bonawitz et al.~\cite{bonawitz2017} reduce communication through pairwise secret sharing and handle client dropout, achieving $O(k)$ per-client communication with $O(k^2)$ setup. PolyVeil achieves $O(1)$ per-client communication (two scalars) in its compressed variant without pairwise key agreement.

\paragraph{Homomorphic encryption.}
Additively homomorphic schemes such as Paillier~\cite{paillier1999} support additive aggregation natively. Each client encrypts their bit count under a common public key, the server multiplies ciphertexts, and a designated party decrypts the sum. This achieves exact results with IND-CPA security but requires public-key infrastructure. A Paillier ciphertext is typically 4096 bits at 128-bit security, whereas PolyVeil transmits a single scalar per client. Fully homomorphic encryption~\cite{gentry2009} generalizes to arbitrary computations but with substantially greater overhead.

\paragraph{Differential privacy.}
Differential privacy~\cite{dwork2006} provides formal privacy guarantees through calibrated noise. In the local model, each client randomizes their data before sending it to the server, achieving privacy without trust but with error $\Theta(\sqrt{kn}/\varepsilon)$. The central model achieves error $\Theta(n/\varepsilon)$ but requires a trusted curator to see raw data. The shuffle model~\cite{cheu2019,balle2019} interpolates by interposing an anonymous shuffler, achieving central-model accuracy with local-model trust. PolyVeil produces exact results and uses the same shuffling infrastructure, but derives privacy from algebraic masking rather than statistical noise.

\paragraph{Instance mixing and data obfuscation.}
InstaHide~\cite{huang2020} mixes private data records with public datasets and random sign patterns for privacy-preserving machine learning. While both InstaHide and PolyVeil involve mixing private data with random elements, PolyVeil focuses on aggregation rather than prediction, achieves exact results through algebraic noise cancellation, and provides security guarantees rooted in the combinatorial structure of the Birkhoff polytope.

\paragraph{Secret sharing.}
Secret sharing schemes~\cite{shamir1979} distribute a secret among multiple parties so that only authorized subsets can reconstruct it. PolyVeil does not use secret sharing directly but instead exploits the structure of doubly stochastic matrices so that the private data is one of many valid decompositions of a publicly shared matrix, creating a computational barrier for any entity that observes the matrix but not the decomposition coefficients.

\section{Preliminaries}
\label{sec:prelim}

We collect notation used throughout the paper. All symbols are defined in context at first use; this table serves as a reference.

\begin{center}
\renewcommand{\arraystretch}{1.25}
\begin{tabular}{cl}
\toprule
\textbf{Symbol} & \textbf{Meaning} \\
\midrule
$n$ & Number of bits per client \\
$k$ & Number of clients \\
$K_t$ & Number of decoy permutations for client $t$ \\
$\mathbf{b}_t \in \{0,1\}^n$ & Client $t$'s private bit vector \\
$s_t = \sum_j b_{t,j}$ & Bit count (Hamming weight) of $\mathbf{b}_t$ \\
$S = \sum_t s_t$ & Total bit count (the target aggregate) \\
$\Pi(b)$ & $2 \times 2$ permutation matrix encoding bit $b$ \\
$M_t = M(\mathbf{b}_t)$ & $2n \times 2n$ block-diagonal permutation matrix encoding $\mathbf{b}_t$ \\
$P_{t,i}$ & $i$-th decoy permutation matrix for client $t$, uniform over $S_{2n}$ \\
$\alpha^*$ & Public weight on the true encoding in $D_t$ \\
$\alpha_{t,i}$ & Weight on the $i$-th decoy ($\sum_i \alpha_{t,i} = 1 - \alpha^*$) \\
$D_t$ & Masked doubly stochastic matrix: $\alpha^* M_t + \sum_i \alpha_{t,i} P_{t,i}$ \\
$R_t$ & Normalized decoy component: $(D_t - \alpha^* M_t)/(1-\alpha^*)$ \\
$\mathbf{w}, \mathbf{y}$ & Extraction vectors ($\mathbf{w}_{2j-1}=1$, $\mathbf{w}_{2j}=0$; $\mathbf{y}_{2j-1}=0$, $\mathbf{y}_{2j}=1$) \\
$f_t = \mathbf{w}^T D_t \mathbf{y}$ & Extracted scalar: $\alpha^* s_t + \eta_t$ \\
$\eta_t$ & Noise in extracted scalar: $(1-\alpha^*)\mathbf{w}^T R_t \mathbf{y}$ \\
$\xi_{t,j}$ & Per-bit noise: $(1-\alpha^*)(R_t \mathbf{y})_{2j-1}$ \\
$F, H$ & Aggregated signal and noise: $F = \sum_t f_t$, $H = \sum_t \eta_t$ \\
$\mathcal{B}_m$ & Birkhoff polytope (set of $m \times m$ doubly stochastic matrices) \\
$S_m$ & Symmetric group (set of all $m!$ permutation matrices of size $m$) \\
$\mathrm{perm}(A)$ & Permanent of matrix $A$ \\
$A(R')$ & Support matrix: $A_{ab} = \mathbf{1}[R'_{ab} > 0]$ \\
$\mathrm{Supp}(R')$ & Support set: $\{Q \in S_{2n} : Q_{ab}=1 \Rightarrow R'_{ab} > 0\}$ \\
$\nu(R')$ & Density of $R_t$ on $\mathcal{B}_{2n}$ evaluated at $R'$ \\
$\mathcal{P}(\sigma_1,\ldots,\sigma_K; R')$ & Coefficient polytope for a permutation tuple and target $R'$ \\
$\varepsilon, \delta$ & Differential privacy parameters \\
$\mu$ & Gaussian DP parameter: $\Delta/\sigma_\eta$ \\
$\rho$ & zCDP parameter: $\Delta^2/(2\sigma^2)$ \\
\bottomrule
\end{tabular}
\end{center}

\subsection{Doubly Stochastic Matrices and the Birkhoff Polytope}

\begin{definition}[Doubly stochastic matrix]
\label{def:ds}
A square matrix $A = (a_{ij}) \in \mathbb{R}^{m \times m}$ with non-negative entries is \emph{doubly stochastic} if every row and every column sums to one, that is, $\sum_{j=1}^m a_{ij} = 1$ for all $i \in [m]$, and $\sum_{i=1}^m a_{ij} = 1$ for all $j \in [m]$.
\end{definition}

\begin{definition}[Birkhoff polytope]
\label{def:birkhoff}
The \emph{Birkhoff polytope} $\mathcal{B}_m$ is the set of all $m \times m$ doubly stochastic matrices. It is a convex polytope in $\mathbb{R}^{m \times m}$ of dimension $(m-1)^2$.
\end{definition}

\begin{theorem}[Birkhoff--von Neumann~\cite{ziegler2012,schrijver2003}]
\label{thm:bvn}
The vertices of $\mathcal{B}_m$ are precisely the $m \times m$ permutation matrices. Every doubly stochastic matrix $A \in \mathcal{B}_m$ can be written as a convex combination of permutation matrices
\[
A = \sum_{i=1}^r \theta_i\, P_i, \qquad \theta_i > 0,\quad \sum_{i=1}^r \theta_i = 1,
\]
where each $P_i$ is a permutation matrix. This is a \emph{Birkhoff--von Neumann (BvN) decomposition} of $A$.
\end{theorem}

BvN decompositions are generically \emph{non-unique}: a doubly stochastic matrix in the interior of $\mathcal{B}_m$ admits many distinct decompositions.

\begin{theorem}[Decomposition multiplicity, Brualdi~\cite{brualdi1982}]
\label{thm:mult}
Let $A \in \mathcal{B}_m$ have $p$ positive entries. The number of distinct BvN decompositions of $A$ is at least $p^2$.
\end{theorem}

\begin{theorem}[Marcus--Ree~\cite{marcus1959}]
\label{thm:marcusree}
Every $m \times m$ doubly stochastic matrix can be expressed as a convex combination of at most $m^2 - 2m + 2$ permutation matrices. Finding a BvN decomposition with the minimum number of permutations is NP-hard~\cite{dufossee2016}.
\end{theorem}

\subsection{Uniform Sampling of Permutation Matrices}
\label{sec:sampling}

A $m \times m$ permutation matrix $P$ corresponds bijectively to a permutation $\sigma \in S_m$ (the symmetric group on $[m] = \{1, \ldots, m\}$) via $P_{ij} = \mathbf{1}[\sigma(i) = j]$. That is, row $i$ of $P$ has its unique 1 in column $\sigma(i)$. To draw $P$ \emph{uniformly at random from $S_m$}, one draws a uniformly random permutation $\sigma$ and constructs the corresponding matrix.

The standard algorithm for drawing a uniform random permutation is the Fisher--Yates shuffle (also known as the Knuth shuffle): starting from the identity permutation $\sigma = (1, 2, \ldots, m)$, for $i = m, m-1, \ldots, 2$, draw $j$ uniformly at random from $\{1, \ldots, i\}$ and swap $\sigma(i) \leftrightarrow \sigma(j)$. This produces each of the $m!$ permutations with equal probability $1/m!$ and runs in $O(m)$ time using $O(m)$ random bits (specifically $\sum_{i=2}^m \lceil \log_2 i \rceil$ bits). For our protocol with $m = 2n$, each random permutation matrix costs $O(n)$ time and $O(n \log n)$ random bits.

\subsection{Permutation Encoding of Binary Data}

\begin{definition}[Bit-to-permutation encoding]
\label{def:encoding}
For a bit $b \in \{0,1\}$, define the $2 \times 2$ permutation matrix
\[
\Pi(b) = \begin{cases}
I_2 = \begin{pmatrix}1&0\\0&1\end{pmatrix} & \text{if } b = 0, \\[6pt]
J_2 = \begin{pmatrix}0&1\\1&0\end{pmatrix} & \text{if } b = 1.
\end{cases}
\]
For a bitstream $\mathbf{b} = (b_1, \ldots, b_n) \in \{0,1\}^n$, define the block-diagonal encoding
\[
M(\mathbf{b}) = \mathrm{blockdiag}\!\big(\Pi(b_1),\; \Pi(b_2),\; \ldots,\; \Pi(b_n)\big) \in \{0,1\}^{2n \times 2n}.
\]
\end{definition}

Since each $\Pi(b_j)$ is a $2 \times 2$ permutation matrix, their block-diagonal assembly $M(\mathbf{b})$ is a $2n \times 2n$ permutation matrix and a vertex of $\mathcal{B}_{2n}$. The block-diagonal structure confines each bit to a disjoint $2 \times 2$ block, enabling the algebraic extraction we now develop.

\begin{example}[Encoding of a 2-bit stream]
\label{ex:encoding}
Consider the bitstream $\mathbf{b} = (1, 0)$ with $n = 2$. The encoding produces
\[
M\bigl((1,0)\bigr) = \mathrm{blockdiag}\!\left(\Pi(1),\; \Pi(0)\right) = \mathrm{blockdiag}\!\left(\begin{pmatrix}0&1\\1&0\end{pmatrix},\; \begin{pmatrix}1&0\\0&1\end{pmatrix}\right) = \begin{pmatrix}0&1&0&0\\1&0&0&0\\0&0&1&0\\0&0&0&1\end{pmatrix}.
\]
The first $2 \times 2$ block (rows 1--2, columns 1--2) encodes $b_1 = 1$ as the swap matrix $J_2$. The second $2 \times 2$ block (rows 3--4, columns 3--4) encodes $b_2 = 0$ as the identity $I_2$. All entries outside these blocks are zero due to the block-diagonal structure.
\end{example}

\subsection{Algebraic Bit Count Extraction}

The bilinear form $\mathbf{w}^T A\, \mathbf{y}$, applied to a doubly stochastic matrix $A$, extracts a scalar that depends on the structure of $A$. When $A = M(\mathbf{b})$ is a permutation encoding of a bit vector, appropriate choices of $\mathbf{w}$ and $\mathbf{y}$ recover the bit count, individual bits, or weighted combinations of bits. We call $(\mathbf{w}, \mathbf{y})$ an \emph{extraction vector pair}, as they extract a target statistic from the encoded matrix. Different extraction vector pairs applied to the same matrix yield different statistics, which is the basis for the multi-statistic extraction developed in Section~\ref{sec:multistat}.

\begin{definition}[Extraction vectors for the bit count]
\label{def:selectors}
Define $\mathbf{w}, \mathbf{y} \in \mathbb{R}^{2n}$ by
\[
w_i = \begin{cases} 1 & \text{if } i \text{ is odd},\\ 0 & \text{if } i \text{ is even},\end{cases}
\qquad\qquad
y_i = \begin{cases} 0 & \text{if } i \text{ is odd},\\ 1 & \text{if } i \text{ is even}.\end{cases}
\]
Equivalently, $\mathbf{w} = (1,0,1,0,\ldots,1,0)^T$ and $\mathbf{y} = (0,1,0,1,\ldots,0,1)^T$, each of length $2n$.
\end{definition}

\begin{lemma}[Bit count extraction]
\label{lem:extract}
For any bitstream $\mathbf{b} = (b_1, \ldots, b_n) \in \{0,1\}^n$ with permutation encoding $M = M(\mathbf{b})$
\[
\mathbf{w}^T M\, \mathbf{y} \;=\; \sum_{j=1}^n b_j.
\]
\end{lemma}

\begin{proof}
We compute $\mathbf{w}^T M\, \mathbf{y}$ by expanding the matrix-vector products step by step.

\medskip\noindent\textbf{Computing $M\mathbf{y}$.}
Since $M = \mathrm{blockdiag}(\Pi(b_1), \ldots, \Pi(b_n))$ is block-diagonal, the product $M\mathbf{y}$ decomposes into independent block multiplications. For the $j$-th block, the relevant entries of $\mathbf{y}$ are $y_{2j-1} = 0$ and $y_{2j} = 1$. Writing $\mathbf{y}_j = (y_{2j-1},\, y_{2j})^T = (0, 1)^T$ for the portion of $\mathbf{y}$ corresponding to block $j$, we have
\[
\Pi(b_j)\, \mathbf{y}_j = \Pi(b_j)\binom{0}{1}.
\]
When $b_j = 0$: $\Pi(0)\binom{0}{1} = \begin{pmatrix}1&0\\0&1\end{pmatrix}\binom{0}{1} = \binom{0}{1}$.
When $b_j = 1$: $\Pi(1)\binom{0}{1} = \begin{pmatrix}0&1\\1&0\end{pmatrix}\binom{0}{1} = \binom{1}{0}$.
Therefore, the entries of $M\mathbf{y}$ at positions $2j{-}1$ and $2j$ are
\begin{equation}
\label{eq:My_entries}
(M\mathbf{y})_{2j-1} = b_j, \qquad (M\mathbf{y})_{2j} = 1 - b_j.
\end{equation}
To verify this: when $b_j = 0$, $(M\mathbf{y})_{2j-1} = 0 = b_j$ and $(M\mathbf{y})_{2j} = 1 = 1 - b_j$. When $b_j = 1$, $(M\mathbf{y})_{2j-1} = 1 = b_j$ and $(M\mathbf{y})_{2j} = 0 = 1 - b_j$. Both cases match~\eqref{eq:My_entries}.

\medskip\noindent\textbf{Computing $\mathbf{w}^T(M\mathbf{y})$.}
Now we take the inner product of $\mathbf{w}$ with $M\mathbf{y}$
\begin{align}
\mathbf{w}^T(M\mathbf{y}) &= \sum_{i=1}^{2n} w_i\, (M\mathbf{y})_i \notag\\
&= \sum_{j=1}^n \Big[ w_{2j-1}\,(M\mathbf{y})_{2j-1} \;+\; w_{2j}\,(M\mathbf{y})_{2j}\Big] \notag\\
&= \sum_{j=1}^n \Big[ 1 \cdot b_j \;+\; 0 \cdot (1 - b_j)\Big] \notag\\
&= \sum_{j=1}^n b_j, \label{eq:extraction_final}
\end{align}
where the third equality substitutes $w_{2j-1} = 1$, $w_{2j} = 0$ from Definition~\ref{def:selectors} and $(M\mathbf{y})_{2j-1} = b_j$, $(M\mathbf{y})_{2j} = 1-b_j$ from~\eqref{eq:My_entries}. The key mechanism is that $\mathbf{w}$ ``selects'' only the odd-indexed entries of $M\mathbf{y}$, each of which equals the corresponding bit $b_j$.
\end{proof}

\begin{example}[Extraction for a concrete bitstream]
\label{ex:extraction}
Take $n = 2$, $\mathbf{b} = (1, 0)$. Then $\mathbf{w} = (1,0,1,0)^T$ and $\mathbf{y} = (0,1,0,1)^T$.
From Example~\ref{ex:encoding}
\[
M = \begin{pmatrix}0&1&0&0\\1&0&0&0\\0&0&1&0\\0&0&0&1\end{pmatrix}.
\]
First compute $M\mathbf{y}$
\[
M\mathbf{y} = \begin{pmatrix}0&1&0&0\\1&0&0&0\\0&0&1&0\\0&0&0&1\end{pmatrix}\begin{pmatrix}0\\1\\0\\1\end{pmatrix} = \begin{pmatrix}0 \cdot 0 + 1 \cdot 1 + 0 \cdot 0 + 0 \cdot 1\\1 \cdot 0 + 0 \cdot 1 + 0 \cdot 0 + 0 \cdot 1\\0 \cdot 0 + 0 \cdot 1 + 1 \cdot 0 + 0 \cdot 1\\0 \cdot 0 + 0 \cdot 1 + 0 \cdot 0 + 1 \cdot 1\end{pmatrix} = \begin{pmatrix}1\\0\\0\\1\end{pmatrix}.
\]
Checking against~\eqref{eq:My_entries}, for block $j=1$ ($b_1=1$): $(M\mathbf{y})_1 = 1 = b_1$ and $(M\mathbf{y})_2 = 0 = 1-b_1$. For block $j=2$ ($b_2=0$): $(M\mathbf{y})_3 = 0 = b_2$ and $(M\mathbf{y})_4 = 1 = 1-b_2$. Both match.

Now compute $\mathbf{w}^T(M\mathbf{y})$
\[
\mathbf{w}^T(M\mathbf{y}) = (1,0,1,0)\begin{pmatrix}1\\0\\0\\1\end{pmatrix} = 1 \cdot 1 + 0 \cdot 0 + 1 \cdot 0 + 0 \cdot 1 = 1.
\]
The result is $1 = b_1 + b_2 = 1 + 0$, confirming the lemma.
\end{example}

\section{The PolyVeil Protocol}
\label{sec:protocol}

\subsection{Problem Statement}

We consider $k$ client entities, where client $t \in [k]$ holds a private binary bitstream $\mathbf{b}_t = (b_{t,1}, \ldots, b_{t,n}) \in \{0,1\}^n$ of length $n$. Let $s_t = \sum_{j=1}^n b_{t,j}$ denote the number of ones in client $t$'s bitstream. The goal is for a server to compute the aggregate $S = \sum_{t=1}^k s_t$ without learning any individual $s_t$, any individual bitstream $\mathbf{b}_t$, or any partial aggregates involving fewer than all $k$ clients.

\subsection{Threat Model}
\label{sec:threat}

We operate in the \emph{honest-but-curious} (semi-honest) model, defined as follows.

\paragraph{Server behavior.}
The server executes every instruction of the protocol exactly as specified. It does not deviate from the protocol by, for example, sending altered intermediate results to clients, injecting false data, or failing to perform a required computation. However, the server records every message it receives and may subsequently perform arbitrary polynomial-time computations on this recorded transcript in an attempt to infer individual client data. The server's computational power is bounded only by polynomial time; it may run brute-force searches over feasible spaces, solve optimization problems, and apply any statistical inference technique. The security guarantees we prove hold against any such polynomial-time analysis.

\paragraph{Communication security.}
All communication channels between each client and the server, and between each client and the shuffler, are authenticated and encrypted using standard transport-layer security (e.g., TLS 1.3). Authentication ensures that the server receives messages only from legitimate clients and not from impersonators. Encryption ensures that no external eavesdropper observing the network can read the content of any message. Together, these guarantees mean that the only entity that sees a client's message to the server is the server itself, and the only entity that sees a client's message to the shuffler is the shuffler itself. We do not assume that the communication channels hide metadata such as message timing or size.

\paragraph{Client behavior.}
Every client follows the protocol faithfully. No client modifies, omits, or fabricates any message. No client shares its private data, its random coins, or its intermediate computations with the server or with any other client (beyond what the protocol prescribes). In particular, no client colludes with the server to de-anonymize the shuffled values. The non-collusion assumption is essential, since if even one client shared its $\eta_t$ value directly with the server (outside the shuffle), the server could link that $\eta_t$ to the client's identity and compute $s_t = (f_t - \eta_t)/\alpha^*$.

\paragraph{Trusted shuffler.}
There exists a functionality $\mathcal{F}_{\mathrm{shuffle}}$ that operates as follows. It accepts as input one scalar value from each of the $k$ clients, collecting the multiset $\{\eta_1, \ldots, \eta_k\}$. It then applies a permutation $\pi$ drawn uniformly at random from the symmetric group $S_k$ (the set of all $k!$ bijections on $[k]$), and outputs the permuted sequence $(\eta_{\pi(1)}, \ldots, \eta_{\pi(k)})$ to the server. Critically, the server learns the values in the output sequence but does \emph{not} learn the permutation $\pi$. This means that for any position $j$ in the output, the server knows the value $\eta_{\pi(j)}$ but cannot determine which client submitted it. The shuffler does not reveal $\pi$ to the clients either.

This ideal functionality can be instantiated in several ways. The simplest is a \emph{non-colluding auxiliary server}: a separate physical server, operated by an independent party that does not collude with the main server, receives all $\eta_t$ values, permutes them, and forwards the result. A stronger instantiation is a \emph{mixnet}, where each client encrypts its value under layered encryption addressed to a chain of relay servers, each of which peels one encryption layer and shuffles the messages. Verifiable mixnets~\cite{chase2020} additionally produce a zero-knowledge proof that the output is a valid permutation of the input, preventing a malicious relay from altering values. A fully cryptographic instantiation uses a \emph{multi-party shuffling protocol} such as the secret-shared shuffle of Chase, Ghosh, and Poburinnaya~\cite{chase2020}, which distributes the shuffling computation among two or more servers so that no single server learns the permutation, achieving security even if one server is corrupted.

We require $k \geq 3$ to ensure that the shuffle provides meaningful anonymity. With $k = 1$, the shuffled output contains a single value that is trivially linked to the sole client. With $k = 2$, the server has a $1/2$ probability of guessing the correct assignment, providing negligible privacy.

\subsection{Protocol Description}

The protocol uses a public parameter $\alpha^* \in (0,1)$ that controls the trade-off between signal strength and privacy, as smaller $\alpha^*$ hides $M_t$ more deeply in the interior of $\mathcal{B}_{2n}$ but requires greater numerical precision to recover $S$.

\begin{algorithm}[!ht]
\caption{PolyVeil}
\label{alg:polyveil}
\begin{algorithmic}[1]
\Statex \textbf{Public parameters:} bit length $n$, scaling factor $\alpha^* \in (0,1)$, extraction vectors $\mathbf{w}$, $\mathbf{y}$ (Definition~\ref{def:selectors}).

\Statex
\Statex \underline{\textsc{Client-Side Masking}}
\For{each client $t \in \{1, \ldots, k\}$ in parallel}
    \State Encode private bitstream as $M_t = M(\mathbf{b}_t) \in \{0,1\}^{2n \times 2n}$ via Definition~\ref{def:encoding}.
    \State Choose $K_t \geq 2$ and draw $K_t$ uniformly random permutation matrices $P_{t,1}, \ldots, P_{t,K_t}$ from $S_{2n}$ using the Fisher--Yates shuffle (Section~\ref{sec:sampling}).
    \State Sample positive coefficients $\alpha_{t,1}, \ldots, \alpha_{t,K_t} > 0$ with $\sum_{i=1}^{K_t} \alpha_{t,i} = 1 - \alpha^*$.
    \State Construct the masked matrix:
    \begin{equation}
    \label{eq:Dt}
        D_t \;=\; \alpha^*\, M_t \;+\; \sum_{i=1}^{K_t} \alpha_{t,i}\, P_{t,i}.
    \end{equation}
    \State Compute the noise term:
    \begin{equation}
    \label{eq:eta}
        \eta_t \;=\; \sum_{i=1}^{K_t} \alpha_{t,i}\;\bigl(\mathbf{w}^T P_{t,i}\, \mathbf{y}\bigr).
    \end{equation}
    \State \textbf{Send} $D_t$ to the server.
\EndFor

\Statex
\Statex \underline{\textsc{Server-Side Extraction}}
\For{each client $t \in \{1, \ldots, k\}$}
    \State Server computes:
    \begin{equation}
    \label{eq:ft}
        f_t \;=\; \mathbf{w}^T D_t\, \mathbf{y}.
    \end{equation}
\EndFor

\Statex
\Statex \underline{\textsc{Secure Noise Transmission}}
\State Each client $t$ submits scalar $\eta_t$ to the trusted shuffler $\mathcal{F}_{\mathrm{shuffle}}$.
\State Shuffler draws $\pi \sim S_k$ and sends $(\eta_{\pi(1)}, \ldots, \eta_{\pi(k)})$ to server.

\Statex
\Statex \underline{\textsc{Aggregation}}
\State Server computes:
\begin{equation}
\label{eq:aggregate}
    S \;=\; \frac{1}{\alpha^*}\left(\sum_{t=1}^k f_t \;-\; \sum_{t=1}^k \eta_t\right).
\end{equation}
\end{algorithmic}
\end{algorithm}

\subsection{Proof of Correctness}
\label{sec:correctness}

\begin{theorem}[Correctness]
\label{thm:correct}
Algorithm~\ref{alg:polyveil} computes $S = \sum_{t=1}^k s_t$ exactly.
\end{theorem}

\begin{proof}
The matrix $M_t = M(\mathbf{b}_t)$ is a permutation matrix (Definition~\ref{def:encoding}) and hence doubly stochastic: each row and column contains exactly one 1 and the rest 0, so all row sums and column sums equal 1. Each $P_{t,i}$ is a permutation matrix drawn uniformly from $S_{2n}$ and is likewise doubly stochastic. The coefficients $\alpha^*,\alpha_{t,1},\ldots,\alpha_{t,K_t}$ are strictly positive with $\alpha^* + \sum_i \alpha_{t,i} = 1$. Since $\mathcal{B}_{2n}$ is convex and $D_t$ is a convex combination of elements of $\mathcal{B}_{2n}$
\[
D_t = \alpha^* M_t + \sum_{i=1}^{K_t} \alpha_{t,i} P_{t,i} \;\in\; \mathcal{B}_{2n}.
\]

Applying the bilinear form $A \mapsto \mathbf{w}^T A\, \mathbf{y}$ to $D_t$ and using linearity of matrix-vector multiplication
\begin{align}
f_t &= \mathbf{w}^T D_t\, \mathbf{y} \notag \\
    &= \mathbf{w}^T \!\left(\alpha^* M_t + \sum_{i=1}^{K_t} \alpha_{t,i} P_{t,i}\right) \mathbf{y} \notag \\
    &= \alpha^*\,(\mathbf{w}^T M_t\, \mathbf{y}) + \sum_{i=1}^{K_t} \alpha_{t,i}\,(\mathbf{w}^T P_{t,i}\, \mathbf{y}). \label{eq:ft_decomp}
\end{align}
By Lemma~\ref{lem:extract}, $\mathbf{w}^T M_t \mathbf{y} = \sum_{j=1}^n b_{t,j} = s_t$. By definition~\eqref{eq:eta}, $\eta_t = \sum_{i=1}^{K_t} \alpha_{t,i}(\mathbf{w}^T P_{t,i}\mathbf{y})$. Substituting,
\[
f_t = \alpha^*\, s_t + \eta_t.
\]

Summing over all $k$ clients,
\begin{align*}
\sum_{t=1}^k f_t
&= \sum_{t=1}^k (\alpha^*\, s_t + \eta_t)
= \alpha^* \sum_{t=1}^k s_t + \sum_{t=1}^k \eta_t
= \alpha^* S + \sum_{t=1}^k \eta_t.
\end{align*}
The server knows $\sum_t f_t$ from the extraction step. The shuffled list $(\eta_{\pi(1)}, \ldots, \eta_{\pi(k)})$ is a permutation of $(\eta_1, \ldots, \eta_k)$, and the sum is invariant under permutation, so $\sum_{j=1}^k \eta_{\pi(j)} = \sum_{t=1}^k \eta_t$. Therefore the server can compute
\[
S = \frac{1}{\alpha^*}\!\left(\sum_{t=1}^k f_t - \sum_{t=1}^k \eta_t\right) = \frac{1}{\alpha^*}\!\left(\alpha^* S + \sum_t \eta_t - \sum_t \eta_t\right) = \frac{1}{\alpha^*} \cdot \alpha^* S = S,
\]
confirming that the protocol outputs the correct aggregate.
\end{proof}

\subsection{Integrity in the Full Protocol}
\label{sec:integrity}

In the full (non-compressed) protocol, the server receives the matrix $D_t$ directly and computes $f_t = \mathbf{w}^T D_t \mathbf{y}$ itself. A malicious client cannot cause the server to use an incorrect $f_t$ because the server performs the extraction independently. Specifically, even if a client wished to inflate or deflate its contribution to the aggregate, the client can only control what matrix $D_t$ it sends. The server then computes $f_t = \mathbf{w}^T D_t \mathbf{y}$ deterministically from $D_t$, so the client cannot make the server believe a different $f_t$ than the one implied by the submitted $D_t$.

The server can additionally verify that the received $D_t$ is a valid doubly stochastic matrix by checking that all entries are non-negative and that every row and column sums to~1 (within floating-point tolerance). If a client submits a matrix that is not doubly stochastic, the server can reject it. This verification does not reveal the client's private data (since any doubly stochastic matrix passes the check, regardless of which permutation is hidden inside), but it prevents malformed submissions that could corrupt the aggregate.

The remaining vulnerability is that a malicious client could submit a valid doubly stochastic matrix $D_t$ that encodes a bitstream different from its true $\mathbf{b}_t$. This is the standard ``input substitution'' attack in the semi-honest model, where a dishonest client lies about its data. Preventing this requires mechanisms beyond the semi-honest model, such as zero-knowledge proofs that $D_t$ is correctly constructed from the client's certified data source. We do not address this in the current work.

\subsection{Compressed Variant}
\label{sec:compressed}

\begin{algorithm}[!ht]
\caption{PolyVeil (Compressed)}
\label{alg:compressed}
\begin{algorithmic}[1]
\Statex \textbf{Public parameter:} $\alpha^* \in (0,1)$, extraction vectors $\mathbf{w}$, $\mathbf{y}$.
\For{each client $t \in \{1, \ldots, k\}$ in parallel}
    \State Compute $s_t = \sum_{j=1}^n b_{t,j}$ locally.
    \State Draw random permutations $P_{t,1}, \ldots, P_{t,K_t}$ and coefficients $\alpha_{t,i} > 0$ with $\sum_i \alpha_{t,i} = 1 - \alpha^*$.
    \State Compute $\eta_t = \sum_i \alpha_{t,i}\,(\mathbf{w}^T P_{t,i}\, \mathbf{y})$.
    \State Compute $f_t = \alpha^* s_t + \eta_t$.
    \State \textbf{Send} $f_t$ to server.
\EndFor
\State Clients jointly shuffle $\{\eta_1, \ldots, \eta_k\}$ and deliver to server.
\State Server computes $S = \frac{1}{\alpha^*}\bigl(\sum_t f_t - \sum_t \eta_t\bigr)$.
\end{algorithmic}
\end{algorithm}

The compressed variant offers three concrete advantages beyond the communication reduction from $O(n^2)$ to $O(1)$ per client. First, the server never sees the doubly stochastic matrix $D_t$, which eliminates the BvN decomposition as an attack vector, as the server has no matrix to decompose. Second, client-side computation drops from $O(n^2)$ (constructing a $2n \times 2n$ matrix) to $O(nK_t)$ (generating $K_t$ random permutations and computing $K_t$ bilinear forms). Third, server computation drops from $O(kn^2)$ to $O(k)$, becoming independent of the bitstream length.

The trade-off is that a malicious client can send an arbitrary $f_t$ without the server being able to verify it, since the server no longer has $D_t$ to check. Under the semi-honest model this is not a concern.

\subsection{Worked Example with Full Computation}
\label{sec:example}

We trace every computation explicitly for $k = 3$ clients, $n = 2$ bits, $\alpha^* = 0.3$.

\subsubsection{Ground Truth}
Client~1 holds $\mathbf{b}_1 = (1, 0)$, so $s_1 = 1$. Client~2 holds $\mathbf{b}_2 = (1, 1)$, so $s_2 = 2$. Client~3 holds $\mathbf{b}_3 = (0, 1)$, so $s_3 = 1$. The true aggregate is $S = 1 + 2 + 1 = 4$.

The public parameters are $n = 2$, $\alpha^* = 0.3$, $\mathbf{w} = (1,0,1,0)^T$, $\mathbf{y} = (0,1,0,1)^T$.

\subsubsection{Client-Side Masking (Full Detail for Client 1)}

\paragraph{Encoding $M_1$.}
Client~1's bitstream is $(1, 0)$. Applying Definition~\ref{def:encoding}
\[
M_1 = \mathrm{blockdiag}\!\left(\Pi(1),\; \Pi(0)\right) = \mathrm{blockdiag}\!\left(\begin{pmatrix}0&1\\1&0\end{pmatrix},\; \begin{pmatrix}1&0\\0&1\end{pmatrix}\right) = \begin{pmatrix}0&1&0&0\\1&0&0&0\\0&0&1&0\\0&0&0&1\end{pmatrix}.
\]

\paragraph{Generating decoy permutations.}
Client~1 chooses $K_1 = 2$ decoy permutations drawn uniformly from $S_4$ (the symmetric group on $\{1,2,3,4\}$). Suppose the Fisher--Yates shuffle produces

Permutation $\sigma_1 = (3,1,4,2)$, meaning $\sigma_1(1)=3$, $\sigma_1(2)=1$, $\sigma_1(3)=4$, $\sigma_1(4)=2$. The corresponding permutation matrix has $P_{1,1}[i, \sigma_1(i)] = 1$
\[
P_{1,1} = \begin{pmatrix}0&0&1&0\\1&0&0&0\\0&0&0&1\\0&1&0&0\end{pmatrix}.
\]
Verification that this is correct --- row 1 has its 1 in column 3 (since $\sigma_1(1)=3$); row 2 in column 1; row 3 in column 4; row 4 in column 2. Each row and column has exactly one~1.

Permutation $\sigma_2 = (2,1,4,3)$
\[
P_{1,2} = \begin{pmatrix}0&1&0&0\\1&0&0&0\\0&0&0&1\\0&0&1&0\end{pmatrix}.
\]
To verify, row 1 has 1 in column 2; row 2 in column 1; row 3 in column 4; row 4 in column 3.

\paragraph{Choosing coefficients.}
Client~1 samples $\alpha_{1,1} = 0.5$ and $\alpha_{1,2} = 0.2$, satisfying $\alpha_{1,1} + \alpha_{1,2} = 0.7 = 1 - 0.3 = 1 - \alpha^*$.

\paragraph{Constructing $D_1$.}
Applying equation~\eqref{eq:Dt}: $D_1 = \alpha^* M_1 + \alpha_{1,1} P_{1,1} + \alpha_{1,2} P_{1,2} = 0.3\, M_1 + 0.5\, P_{1,1} + 0.2\, P_{1,2}$.

Computing each scaled matrix,
\[
0.3\, M_1 = \begin{pmatrix}0&0.3&0&0\\0.3&0&0&0\\0&0&0.3&0\\0&0&0&0.3\end{pmatrix},\quad
0.5\, P_{1,1} = \begin{pmatrix}0&0&0.5&0\\0.5&0&0&0\\0&0&0&0.5\\0&0.5&0&0\end{pmatrix},
\]
\[
0.2\, P_{1,2} = \begin{pmatrix}0&0.2&0&0\\0.2&0&0&0\\0&0&0&0.2\\0&0&0.2&0\end{pmatrix}.
\]

Summing entry by entry,
\[
D_1 = \begin{pmatrix}
0+0+0 & 0.3+0+0.2 & 0+0.5+0 & 0+0+0\\
0.3+0.5+0.2 & 0+0+0 & 0+0+0 & 0+0+0\\
0+0+0 & 0+0+0 & 0.3+0+0 & 0+0.5+0.2\\
0+0+0 & 0+0.5+0 & 0+0+0.2 & 0.3+0+0
\end{pmatrix}
= \begin{pmatrix}
0 & 0.5 & 0.5 & 0\\
1.0 & 0 & 0 & 0\\
0 & 0 & 0.3 & 0.7\\
0 & 0.5 & 0.2 & 0.3
\end{pmatrix}.
\]

To verify that $D_1$ is doubly stochastic, the row sums are $0+0.5+0.5+0=1$, $1+0+0+0=1$, $0+0+0.3+0.7=1$, $0+0.5+0.2+0.3=1$. The column sums are $0+1+0+0=1$, $0.5+0+0+0.5=1$, $0.5+0+0.3+0.2=1$, $0+0+0.7+0.3=1$. All entries are non-negative.

\paragraph{Computing $\eta_1$.}
Applying equation~\eqref{eq:eta}: $\eta_1 = \alpha_{1,1}\,(\mathbf{w}^T P_{1,1}\, \mathbf{y}) + \alpha_{1,2}\,(\mathbf{w}^T P_{1,2}\, \mathbf{y})$.

For $P_{1,1}$
\[
P_{1,1}\mathbf{y} = \begin{pmatrix}0&0&1&0\\1&0&0&0\\0&0&0&1\\0&1&0&0\end{pmatrix}\begin{pmatrix}0\\1\\0\\1\end{pmatrix} = \begin{pmatrix}0\\0\\1\\1\end{pmatrix},
\qquad
\mathbf{w}^T(P_{1,1}\mathbf{y}) = (1,0,1,0)\begin{pmatrix}0\\0\\1\\1\end{pmatrix} = 0 + 0 + 1 + 0 = 1.
\]

For $P_{1,2}$
\[
P_{1,2}\mathbf{y} = \begin{pmatrix}0&1&0&0\\1&0&0&0\\0&0&0&1\\0&0&1&0\end{pmatrix}\begin{pmatrix}0\\1\\0\\1\end{pmatrix} = \begin{pmatrix}1\\0\\1\\0\end{pmatrix},
\qquad
\mathbf{w}^T(P_{1,2}\mathbf{y}) = (1,0,1,0)\begin{pmatrix}1\\0\\1\\0\end{pmatrix} = 1 + 0 + 1 + 0 = 2.
\]

Therefore, $\eta_1 = 0.5 \times 1 + 0.2 \times 2 = 0.5 + 0.4 = 0.9$.

Client~1 sends $D_1$ to the server and holds $\eta_1 = 0.9$ for the shuffle.

\subsubsection{Client-Side Masking for Clients 2 and 3 (Summary)}

\paragraph{Client 2.}
Bitstream $\mathbf{b}_2 = (1,1)$, $s_2 = 2$.
\[
M_2 = \mathrm{blockdiag}\!\left(\Pi(1), \Pi(1)\right) = \begin{pmatrix}0&1&0&0\\1&0&0&0\\0&0&0&1\\0&0&1&0\end{pmatrix}.
\]
Suppose Client~2 draws $P_{2,1}$ corresponding to $\sigma = (4,3,2,1)$ and $P_{2,2}$ corresponding to $\sigma = (1,2,3,4)$ (the identity), with $\alpha_{2,1} = 0.4$, $\alpha_{2,2} = 0.3$
\[
P_{2,1} = \begin{pmatrix}0&0&0&1\\0&0&1&0\\0&1&0&0\\1&0&0&0\end{pmatrix},\qquad
P_{2,2} = \begin{pmatrix}1&0&0&0\\0&1&0&0\\0&0&1&0\\0&0&0&1\end{pmatrix}.
\]
Then $D_2 = 0.3\,M_2 + 0.4\,P_{2,1} + 0.3\,P_{2,2}$. Compute $\mathbf{w}^T P_{2,1}\mathbf{y}$
\[
P_{2,1}\mathbf{y} = \begin{pmatrix}1\\0\\1\\0\end{pmatrix}, \quad \mathbf{w}^T\begin{pmatrix}1\\0\\1\\0\end{pmatrix} = 2.
\]
Compute $\mathbf{w}^T P_{2,2}\mathbf{y}$
\[
P_{2,2}\mathbf{y} = \mathbf{y} = \begin{pmatrix}0\\1\\0\\1\end{pmatrix}, \quad \mathbf{w}^T \mathbf{y} = 0.
\]
So $\eta_2 = 0.4 \times 2 + 0.3 \times 0 = 0.8$.

\paragraph{Client 3.}
Bitstream $\mathbf{b}_3 = (0,1)$, $s_3 = 1$.
\[
M_3 = \mathrm{blockdiag}\!\left(\Pi(0), \Pi(1)\right) = \begin{pmatrix}1&0&0&0\\0&1&0&0\\0&0&0&1\\0&0&1&0\end{pmatrix}.
\]
Suppose Client~3 draws $P_{3,1}$ from $\sigma=(2,1,3,4)$ and $P_{3,2}$ from $\sigma=(3,4,1,2)$, with $\alpha_{3,1}=0.35$, $\alpha_{3,2}=0.35$
\[
P_{3,1} = \begin{pmatrix}0&1&0&0\\1&0&0&0\\0&0&1&0\\0&0&0&1\end{pmatrix},\qquad
P_{3,2} = \begin{pmatrix}0&0&1&0\\0&0&0&1\\1&0&0&0\\0&1&0&0\end{pmatrix}.
\]
$P_{3,1}\mathbf{y} = (1,0,0,1)^T$, so $\mathbf{w}^T(P_{3,1}\mathbf{y}) = 1\cdot 1 + 0\cdot 0 + 1\cdot 0 + 0\cdot 1 = 1$.

$P_{3,2}\mathbf{y} = (0,1,0,1)^T$, so $\mathbf{w}^T(P_{3,2}\mathbf{y}) = 0$.

$\eta_3 = 0.35\times 1 + 0.35 \times 0 = 0.35$.

\subsubsection{Server-Side Extraction}

The server computes $f_t = \mathbf{w}^T D_t\, \mathbf{y}$ for each client. By the decomposition proved in equation~\eqref{eq:ft_decomp} of Theorem~\ref{thm:correct} (namely, that the bilinear extraction $\mathbf{w}^T D_t \mathbf{y}$ equals $\alpha^* s_t + \eta_t$ due to the linearity of the bilinear form and Lemma~\ref{lem:extract}), we have

For Client 1: $f_1 = \alpha^* s_1 + \eta_1 = 0.3 \times 1 + 0.9 = 1.2$.

We verify this by direct computation on $D_1$
\[
D_1 \mathbf{y} = \begin{pmatrix}
0 & 0.5 & 0.5 & 0\\
1.0 & 0 & 0 & 0\\
0 & 0 & 0.3 & 0.7\\
0 & 0.5 & 0.2 & 0.3
\end{pmatrix}
\begin{pmatrix}0\\1\\0\\1\end{pmatrix}
= \begin{pmatrix}0.5\\0\\0.7\\0.8\end{pmatrix}.
\]
\[
f_1 = \mathbf{w}^T (D_1 \mathbf{y}) = (1,0,1,0)(0.5,\;0,\;0.7,\;0.8)^T = 1\cdot 0.5 + 0\cdot 0 + 1\cdot 0.7 + 0\cdot 0.8 = 1.2.
\]

For Client~2: $f_2 = \alpha^* s_2 + \eta_2 = 0.3 \times 2 + 0.8 = 1.4$.

For Client~3: $f_3 = \alpha^* s_3 + \eta_3 = 0.3 \times 1 + 0.35 = 0.65$.

\subsubsection{Secure Noise Transmission}

Each client submits its $\eta_t$ to the shuffler. Client~1 submits $\eta_1 = 0.9$, Client~2 submits $\eta_2 = 0.8$, Client~3 submits $\eta_3 = 0.35$. The shuffler draws a uniformly random permutation $\pi \in S_3$; suppose $\pi = (2,3,1)$ (meaning $\pi(1)=2$, $\pi(2)=3$, $\pi(3)=1$). The server receives the sequence $(\eta_{\pi(1)}, \eta_{\pi(2)}, \eta_{\pi(3)}) = (\eta_2, \eta_3, \eta_1) = (0.8,\; 0.35,\; 0.9)$. The server sees the values $0.8$, $0.35$, $0.9$ but does \emph{not} know that $0.8$ came from Client~2, $0.35$ from Client~3, and $0.9$ from Client~1.

\subsubsection{Aggregation}

The server applies equation~\eqref{eq:aggregate}, which states $S = \frac{1}{\alpha^*}(\sum_{t=1}^k f_t - \sum_{t=1}^k \eta_t)$. This formula was derived in Theorem~\ref{thm:correct} from the fact that $f_t = \alpha^* s_t + \eta_t$ (equation~\eqref{eq:ft_decomp}), so $\sum_t f_t - \sum_t \eta_t = \alpha^* \sum_t s_t$, and dividing by $\alpha^*$ recovers $\sum_t s_t$.

The server computes
\[
\sum_{t=1}^3 f_t = 1.2 + 1.4 + 0.65 = 3.25.
\]
\[
\sum_{t=1}^3 \eta_t = 0.8 + 0.35 + 0.9 = 2.05. \quad\text{(sum is the same regardless of shuffle order)}
\]
\[
S = \frac{1}{0.3}(3.25 - 2.05) = \frac{1.2}{0.3} = 4.
\]

\subsubsection{Verification}
The ground truth is $S = s_1 + s_2 + s_3 = 1 + 2 + 1 = 4$. The protocol output matches. To see why the cancellation works algebraically
\begin{align*}
\sum_t f_t - \sum_t \eta_t &= (f_1 - \eta_1) + (f_2 - \eta_2) + (f_3 - \eta_3) \\
&\!\!\!\!\!\!\!\!\!\!\!\!\!\!\!\!\!\!\!\!\!\!\!\!\!\!\text{(note: repairing the sum by the identity $f_t = \alpha^* s_t + \eta_t$ for each $t$)}\\
&= \alpha^* s_1 + \alpha^* s_2 + \alpha^* s_3 \\
&= \alpha^*(s_1 + s_2 + s_3) \\
&= 0.3 \times 4 = 1.2.
\end{align*}
Dividing by $\alpha^* = 0.3$ gives $4$. This cancellation holds for any values of the $\eta_t$'s, any number of decoy permutations, and any positive coefficients, because it depends solely on the algebraic identity $f_t - \eta_t = \alpha^* s_t$.

\subsubsection{What the Server Cannot Do --- De-Shuffling Analysis}
\label{sec:deshuffle}

The server knows $f_1 = 1.2$, $f_2 = 1.4$, $f_3 = 0.65$ (linked to client identities) and the shuffled values $\{0.8, 0.35, 0.9\}$ (unlinked). To learn individual $s_t$ values, the server must assign each shuffled $\eta$-value to the correct client. There are $3! = 6$ possible assignments (since the three values are distinct), and the server can test each.

\medskip
\begin{center}
\small
\begin{tabular}{cccccc}
\toprule
Assignment $\sigma$ & $\hat{s}_1 = \frac{f_1 - \eta_{\sigma(1)}}{\alpha^*}$ & $\hat{s}_2 = \frac{f_2 - \eta_{\sigma(2)}}{\alpha^*}$ & $\hat{s}_3 = \frac{f_3 - \eta_{\sigma(3)}}{\alpha^*}$ & All $\hat{s}_t \in \{0,1,2\}$? \\
\midrule
$(0.8, 0.35, 0.9)$ & $\frac{0.4}{0.3} = 1.33$ & $\frac{1.05}{0.3} = 3.50$ & $\frac{-0.25}{0.3} = -0.83$ & No \\
$(0.8, 0.9, 0.35)$ & $1.33$ & $1.67$ & $1.00$ & No \\
$(0.35, 0.8, 0.9)$ & $2.83$ & $2.00$ & $-0.83$ & No \\
$(0.35, 0.9, 0.8)$ & $2.83$ & $1.67$ & $-0.50$ & No \\
$(0.9, 0.35, 0.8)$ & $1.00$ & $3.50$ & $-0.50$ & No \\
$(0.9, 0.8, 0.35)$ & $\mathbf{1.00}$ & $\mathbf{2.00}$ & $\mathbf{1.00}$ & \textbf{Yes} \\
\bottomrule
\end{tabular}
\end{center}

In this example, only one assignment yields valid bit counts (integers in $\{0,\ldots,n\}$). This is a consequence of the small parameters ($n = 2$, $k = 3$). With larger $n$ and $k$, multiple assignments will produce valid integer counts summing to $S$, and the server cannot distinguish among them. We analyze this formally in Section~\ref{sec:security}.

The probability that the server guesses the correct assignment is $1/k! = 1/6$ by random guessing alone. However, as we show in Section~\ref{sec:security}, the server can exploit the integrality constraint $s_t \in \{0,\ldots,n\}$ to identify the correct assignment with probability 1, rendering the naive protocol insecure.

To \textbf{prevent de-shuffling}, the protocol relies on the trusted shuffler $\mathcal{F}_{\mathrm{shuffle}}$ (Section~\ref{sec:threat}), which guarantees that the permutation $\pi$ is uniformly random and unknown to the server. As we analyze rigorously in Section~\ref{sec:security}, this shuffling alone is \emph{not sufficient} to prevent the server from recovering individual $s_t$ values, because of an integrality constraint that enables deterministic de-shuffling. This motivates the protocol modifications presented in Section~\ref{sec:it_security}.

\section{Security Analysis}
\label{sec:security}

We provide a rigorous security analysis of the basic protocol variants (Algorithms~\ref{alg:polyveil} and~2). We identify fundamental vulnerabilities in both variants, quantify the information leakage precisely, and defer the corrected protocols to Section~\ref{sec:it_security}.

\subsection{The De-Shuffling Attack (Compressed Protocol)}

The compressed protocol has the server receive identity-linked scalars $f_t = \alpha^* s_t + \eta_t$ for $t \in [k]$ and the shuffled sequence $(\tilde\eta_1, \ldots, \tilde\eta_k) = (\eta_{\pi(1)}, \ldots, \eta_{\pi(k)})$.

\begin{theorem}[De-shuffling via integrality constraint]
\label{thm:deshuffle}
In the compressed protocol, the server can recover every client's bit count $s_t$ with probability 1 (over the protocol's randomness).
\end{theorem}

\begin{proof}
The server observes two sets of data, namely identity-linked scalars $(1, f_1), (2, f_2), \ldots, (k, f_k)$ where $f_t = \alpha^* s_t + \eta_t$, and shuffled noise $(\tilde\eta_1, \ldots, \tilde\eta_k) = (\eta_{\pi(1)}, \ldots, \eta_{\pi(k)})$ for unknown $\pi \in S_k$.

\medskip\noindent\textit{The server's test.}
For each candidate bijection $\sigma \colon [k] \to [k]$, the server computes
\[
\hat{s}_t(\sigma) \;=\; \frac{f_t - \tilde\eta_{\sigma(t)}}{\alpha^*}, \qquad t = 1, \ldots, k,
\]
and checks whether $\hat{s}_t(\sigma) \in \{0, 1, \ldots, n\}$ for all $t$ simultaneously.

\medskip\noindent\textit{The true assignment passes the test.}
Let $\sigma^* = \pi^{-1}$ be the true assignment (i.e., $\tilde\eta_{\sigma^*(t)} = \eta_t$ for all $t$). Then
\begin{align*}
\hat{s}_t(\sigma^*) &= \frac{f_t - \tilde\eta_{\sigma^*(t)}}{\alpha^*} = \frac{f_t - \eta_t}{\alpha^*} = \frac{(\alpha^* s_t + \eta_t) - \eta_t}{\alpha^*} = s_t.
\end{align*}
Since $s_t = \sum_{j=1}^n b_{t,j} \in \{0, 1, \ldots, n\}$, all $k$ candidate values pass the integrality test.

\medskip\noindent\textit{Any wrong assignment fails with probability 1.}
Let $\sigma \neq \sigma^*$. There exists some $t_0 \in [k]$ with $\sigma(t_0) \neq \sigma^*(t_0)$. Let $j = \pi(\sigma(t_0))$ be the client whose $\eta$-value is at position $\sigma(t_0)$ in the shuffled sequence, so $\tilde\eta_{\sigma(t_0)} = \eta_j$ with $j \neq t_0$. Then
\begin{align}
\hat{s}_{t_0}(\sigma) &= \frac{f_{t_0} - \eta_j}{\alpha^*} = \frac{\alpha^* s_{t_0} + \eta_{t_0} - \eta_j}{\alpha^*} = s_{t_0} + \frac{\eta_{t_0} - \eta_j}{\alpha^*}. \label{eq:wrong_assign}
\end{align}
This equals an integer if and only if $(\eta_{t_0} - \eta_j)/\alpha^*$ is an integer, i.e., $\eta_{t_0} - \eta_j \in \alpha^* \mathbb{Z}$.

Now we show $\Pr[\eta_{t_0} - \eta_j \in \alpha^* \mathbb{Z}] = 0$. Each $\eta_t = \sum_{i=1}^{K_t} \alpha_{t,i} X_{t,i}$ where $X_{t,i} = \mathbf{w}^T P_{t,i} \mathbf{y} \in \{0, 1, \ldots, n\}$ are integer-valued and the coefficients $\alpha_{t,i}$ are drawn from a continuous distribution on $\{(\alpha_1,\ldots,\alpha_{K_t}) : \alpha_i > 0, \sum_i \alpha_i = 1-\alpha^*\}$. Consider the conditional distribution of $\eta_{t_0}$ given all other randomness (including $\eta_j$ and the integer values $X_{t_0,i}$). Conditional on $X_{t_0,1}, \ldots, X_{t_0,K_{t_0}}$, the random variable $\eta_{t_0} = \sum_i \alpha_{t_0,i} X_{t_0,i}$ is a linear function of the continuously distributed coefficients $(\alpha_{t_0,1}, \ldots, \alpha_{t_0,K_{t_0}})$. Since the $X_{t_0,i}$ are not all equal (with probability 1, as the permutations are drawn independently and uniformly), this linear function is non-constant, so $\eta_{t_0}$ has a continuous conditional distribution. Therefore
\[
\Pr\!\left[\eta_{t_0} - \eta_j \in \alpha^* \mathbb{Z}\right] = \Pr\!\left[\eta_{t_0} \in \eta_j + \alpha^* \mathbb{Z}\right] = 0,
\]
since a continuous random variable assigns probability zero to any countable set.

\medskip\noindent\textit{Uniqueness and conclusion.}
Since $\sigma^*$ passes the test and each $\sigma \neq \sigma^*$ fails with probability 1, taking a union bound over the (finite) set of $k! - 1$ wrong assignments
\[
\Pr[\text{any wrong } \sigma \text{ passes}] \leq \sum_{\sigma \neq \sigma^*} \Pr[\sigma \text{ passes}] = 0.
\]
Therefore, with probability 1, the server uniquely identifies $\sigma^*$ and recovers $s_t = \hat{s}_t(\sigma^*)$ for all $t \in [k]$.
\end{proof}

\begin{remark}[This attack is demonstrated in the worked example]
The worked example in Section~\ref{sec:example} illustrates this attack explicitly. The server tests all $3! = 6$ assignments of shuffled $\eta$ values to clients and finds that exactly one assignment yields valid bit counts $(s_1, s_2, s_3) = (1, 2, 1)$ with all $s_t \in \{0, 1, 2\}$. The paper previously described this as a special case for small $n$, but Theorem~\ref{thm:deshuffle} shows it works for all $n$.
\end{remark}

\subsection{Direct Inference from the Marginal (Compressed Protocol)}

Even without de-shuffling, the server can perform Bayesian inference on $s_t$ from $f_t$ alone.

\begin{proposition}[Server's posterior from $f_t$]
\label{prop:posterior}
Let $\mu$ denote the probability density of $\eta_t$ (which is continuous and independent of $\mathbf{b}_t$). The server's posterior distribution over $s_t$ given $f_t$ is
\begin{equation}
\label{eq:posterior}
\Pr(s_t = s \mid f_t) = \frac{\mu(f_t - \alpha^* s) \cdot \pi(s)}{\sum_{s'=0}^n \mu(f_t - \alpha^* s') \cdot \pi(s')},
\end{equation}
where $\pi(s)$ is any prior distribution over $s_t \in \{0, 1, \ldots, n\}$.
\end{proposition}

\begin{proof}
We derive this from Bayes' theorem. The server knows $f_t$ and seeks $s_t$. Since $f_t = \alpha^* s_t + \eta_t$ and $\eta_t$ is independent of $s_t$ with density $\mu$, the conditional density of $f_t$ given $s_t = s$ takes the form
\begin{align*}
p(f_t \mid s_t = s) &= p(\alpha^* s + \eta_t = f_t) = p(\eta_t = f_t - \alpha^* s) = \mu(f_t - \alpha^* s).
\end{align*}
The second equality uses the deterministic relationship $f_t = \alpha^* s + \eta_t$, so fixing $s_t = s$ means $\eta_t = f_t - \alpha^* s$. Applying Bayes' theorem with prior $\pi(s) = \Pr(s_t = s)$
\begin{align*}
\Pr(s_t = s \mid f_t) &= \frac{p(f_t \mid s_t = s) \cdot \pi(s)}{p(f_t)} = \frac{\mu(f_t - \alpha^* s) \cdot \pi(s)}{\sum_{s'=0}^n \mu(f_t - \alpha^* s') \cdot \pi(s')},
\end{align*}
where the denominator is the marginal density $p(f_t) = \sum_{s'=0}^n p(f_t \mid s_t = s') \pi(s')$, obtained by the law of total probability over the $(n+1)$ possible values of $s_t$.
\end{proof}

The MAP (maximum a posteriori) estimator selects the $s$ that maximizes the numerator. For a uniform prior, this reduces to $\hat{s}_t = \arg\max_s \mu(f_t - \alpha^* s)$.

\begin{remark}[This channel is weak for small $\alpha^*$, but irrelevant given de-shuffling]
When $\alpha^*$ is small relative to the standard deviation of $\eta_t$, the density $\mu(f_t - \alpha^* s)$ varies slowly with $s$ and the posterior is nearly uniform. To see this, we compute $\mathrm{std}(\eta_t)$.

Each term $\mathbf{w}^T P_{t,i}\mathbf{y} = \sum_{j=1}^n (P_{t,i})_{2j-1, 2j}$, which counts the number of $2\times 2$ diagonal blocks in which $P_{t,i}$ maps the odd index to the even index. For a uniform random permutation matrix, each $(P_{t,i})_{2j-1, 2j}$ is a Bernoulli random variable, and $\mathbf{w}^T P_{t,i}\mathbf{y} \in \{0, 1, \ldots, n\}$. Its expectation is
\[
\mathbb{E}[\mathbf{w}^T P_{t,i}\mathbf{y}] = \sum_{j=1}^n \mathbb{E}[(P_{t,i})_{2j-1,2j}] = \sum_{j=1}^n \frac{1}{2n} = \frac{n}{2n} = \frac{1}{2}.
\]
To see this, note that $\mathbf{w}^T P\mathbf{y}$ extracts the sum of entries at positions $(1,2), (3,4), (5,6), \ldots, (2n-1, 2n)$. For a uniform random permutation $\sigma$, $(P)_{2j-1, 2j} = \mathbf{1}[\sigma(2j-1) = 2j]$, and $\Pr[\sigma(2j-1) = 2j] = 1/(2n)$. Summing over $j = 1, \ldots, n$ gives $n/(2n) = 1/2$.
The variance involves correlations between blocks, but for large $n$ is approximately $n \cdot \frac{1}{2n}(1 - \frac{1}{2n}) \approx 1/2$. Therefore $\mathrm{Var}[\mathbf{w}^T P\mathbf{y}] = \Theta(1)$ and $\mathrm{std}[\mathbf{w}^T P\mathbf{y}] = \Theta(1)$.

Now $\eta_t = \sum_{i=1}^{K_t} \alpha_{t,i}(\mathbf{w}^T P_{t,i}\mathbf{y})$. Its mean is
\[
\mathbb{E}[\eta_t] = \sum_{i=1}^{K_t} \alpha_{t,i} \cdot \frac{1}{2} = \frac{1-\alpha^*}{2}.
\]
Its variance (conditional on coefficients, using independence of the $P_{t,i}$) is
\[
\mathrm{Var}[\eta_t] = \sum_{i=1}^{K_t} \alpha_{t,i}^2 \cdot \mathrm{Var}[\mathbf{w}^T P_{t,i}\mathbf{y}] = \Theta(1) \cdot \sum_{i=1}^{K_t} \alpha_{t,i}^2.
\]
For $K_t$ roughly uniform weights, $\sum_i \alpha_{t,i}^2 \approx (1-\alpha^*)^2/K_t$, giving $\mathrm{std}[\eta_t] = \Theta((1-\alpha^*)/\sqrt{K_t})$.

The signal (difference between $f_t$ for adjacent $s$ values) is $\alpha^*$. The total signal range across all $n+1$ possible $s_t$ values is $\alpha^* n$. The signal-to-noise ratio for distinguishing $s_t$ from $s_t + 1$ is
\[
\mathrm{SNR} = \frac{\alpha^*}{\mathrm{std}[\eta_t]} = \frac{\alpha^* \sqrt{K_t}}{(1-\alpha^*) \cdot \Theta(1)}.
\]
For $\alpha^* = 1/(4n)$ and $K_t = O(1)$, this is $O(1/n)$, making the per-unit SNR negligible.

However, this analysis is moot because the de-shuffling attack of Theorem~\ref{thm:deshuffle} recovers $s_t$ with probability 1 regardless of the signal-to-noise ratio.
\end{remark}

\subsection{Likelihood Analysis of the Full Protocol}

In the full protocol, the server additionally receives $D_t = \alpha^* M_t + (1-\alpha^*) R_t$, where $R_t \sim \nu$ is the random decoy component. The server can compute the likelihood of each candidate permutation matrix $M'$

\begin{proposition}[Posterior concentration in the full protocol]
\label{prop:full_posterior}
Let $\nu$ denote the distribution of $R_t = \frac{1}{1-\alpha^*}\sum_{i} \alpha_{t,i} P_{t,i}$ over $\mathcal{B}_{2n}$. For each candidate $M' \in \mathcal{L}(D_t, \alpha^*)$, the likelihood is
\[
\Pr(D_t \mid M_t = M') = \frac{1}{(1-\alpha^*)^d}\; \nu\!\left(\frac{D_t - \alpha^* M'}{1 - \alpha^*}\right),
\]
where $d$ is the dimension of $\mathcal{B}_{2n}$. As $K_t \to \infty$, the posterior probability of the true $M_t$ approaches 1.
\end{proposition}

\begin{proof}
\textit{Deriving the likelihood.}
The masked matrix is $D_t = \alpha^* M_t + (1-\alpha^*) R_t$, where $R_t \sim \nu$. For a candidate $M'$, define the residual $R' = (D_t - \alpha^* M')/(1-\alpha^*)$. We want $\Pr(D_t \mid M_t = M')$. Since $D_t = \alpha^* M' + (1-\alpha^*) R_t$ when $M_t = M'$, the event $\{D_t = d\}$ is the event $\{R_t = (d - \alpha^* M')/(1-\alpha^*)\}$. The map $\varphi \colon R \mapsto \alpha^* M' + (1-\alpha^*) R$ is an affine transformation from $\mathcal{B}_{2n}$ to itself. Its Jacobian matrix is $(1-\alpha^*) I_d$ where $d = (2n-1)^2$ is the dimension of $\mathcal{B}_{2n}$ (a doubly stochastic matrix has $(2n)^2$ entries but $(2n-1)^2$ degrees of freedom after enforcing row and column sum constraints). The absolute value of the Jacobian determinant is $|\det((1-\alpha^*) I_d)| = (1-\alpha^*)^d$. By the standard change-of-variables formula for densities, if $R_t$ has density $\nu(R)$, then $D_t = \varphi(R_t)$ has density
\[
p(d \mid M_t = M') = \frac{1}{(1-\alpha^*)^d}\; \nu\!\left(\frac{d - \alpha^* M'}{1-\alpha^*}\right) = \frac{1}{(1-\alpha^*)^d}\; \nu(R').
\]

\textit{The true residual was drawn from $\nu$; false residuals were not.}
For the true $M_t$, the residual is
\[
R = \frac{D_t - \alpha^* M_t}{1-\alpha^*} = \frac{(\alpha^* M_t + (1-\alpha^*) R_t) - \alpha^* M_t}{1-\alpha^*} = R_t,
\]
which was drawn from $\nu$ by construction. For any $M' \neq M_t$
\begin{align*}
R' &= \frac{D_t - \alpha^* M'}{1-\alpha^*} = \frac{(\alpha^* M_t + (1-\alpha^*) R_t) - \alpha^* M'}{1-\alpha^*} \\
&= R_t + \frac{\alpha^*(M_t - M')}{1-\alpha^*} \\
&= R + \frac{\alpha^*}{1-\alpha^*}(M_t - M').
\end{align*}
Since $M_t \neq M'$, the matrix $M_t - M'$ is nonzero (it has entries in $\{-1, 0, 1\}$ with at least two nonzero entries), so $R'$ is a translate of $R$ by a fixed nonzero shift.

\textit{Computing $\mathbb{E}[P_{ij}]$ for a uniform random permutation matrix.}
Let $P$ be a uniform random permutation matrix in $S_{2n}$, corresponding to a permutation $\sigma$ drawn uniformly from the symmetric group $S_{2n}$. The $(i,j)$ entry of $P$ is the indicator $P_{ij} = \mathbf{1}[\sigma(i) = j]$. Since $\sigma$ is uniform over all $(2n)!$ permutations
\[
\mathbb{E}[P_{ij}] = \Pr[\sigma(i) = j] = \frac{|\{\sigma \in S_{2n} : \sigma(i) = j\}|}{|S_{2n}|} = \frac{(2n-1)!}{(2n)!} = \frac{1}{2n}.
\]
The numerator counts the permutations that map $i$ to $j$: once $\sigma(i) = j$ is fixed, the remaining $2n-1$ elements can be mapped in $(2n-1)!$ ways. Therefore
\[
\mathbb{E}[P] = \frac{1}{2n}\mathbf{J},
\]
where $\mathbf{J}$ is the $(2n) \times (2n)$ all-ones matrix. (One can verify that each row of $\mathbb{E}[P]$ sums to $(2n) \cdot \frac{1}{2n} = 1$, consistent with $P$ being doubly stochastic.)

\textit{Computing $\mathrm{Var}[P_{ij}]$.}
Since $P_{ij} = \mathbf{1}[\sigma(i) = j] \in \{0, 1\}$ is a Bernoulli random variable with parameter $p = 1/(2n)$
\[
\mathrm{Var}[P_{ij}] = p(1-p) = \frac{1}{2n}\left(1 - \frac{1}{2n}\right) = \frac{2n - 1}{(2n)^2}.
\]

\textit{Concentration of $R_t$ as $K_t \to \infty$.}
Each entry of $R_t$ is
\[
(R_t)_{ab} = \frac{1}{1-\alpha^*}\sum_{i=1}^{K_t} \alpha_{t,i} (P_{t,i})_{ab}.
\]
Its expectation is
\begin{align*}
\mathbb{E}[(R_t)_{ab}]
&= \frac{1}{1-\alpha^*}\sum_{i=1}^{K_t} \alpha_{t,i}\, \mathbb{E}[(P_{t,i})_{ab}]
= \frac{1}{1-\alpha^*}\sum_{i=1}^{K_t} \alpha_{t,i} \cdot \frac{1}{2n}
= \frac{1}{1-\alpha^*} \cdot \frac{1}{2n} \sum_{i=1}^{K_t} \alpha_{t,i}
= \frac{1}{1-\alpha^*} \cdot \frac{1}{2n} \cdot (1-\alpha^*)
= \frac{1}{2n},
\end{align*}
where we used $\sum_{i=1}^{K_t} \alpha_{t,i} = 1 - \alpha^*$.

Its variance is (using independence of the $P_{t,i}$ and treating $\alpha_{t,i}$ as fixed conditional on the coefficient draw)
\begin{align*}
\mathrm{Var}[(R_t)_{ab}]
&= \frac{1}{(1-\alpha^*)^2}\sum_{i=1}^{K_t} \alpha_{t,i}^2\, \mathrm{Var}[(P_{t,i})_{ab}]
= \frac{1}{(1-\alpha^*)^2}\sum_{i=1}^{K_t} \alpha_{t,i}^2 \cdot \frac{2n-1}{(2n)^2}.
\end{align*}
For roughly uniform weights $\alpha_{t,i} \approx (1-\alpha^*)/K_t$, we have $\sum_i \alpha_{t,i}^2 \approx K_t \cdot ((1-\alpha^*)/K_t)^2 = (1-\alpha^*)^2/K_t$. Substituting,
\begin{align*}
\mathrm{Var}[(R_t)_{ab}]
&\approx \frac{1}{(1-\alpha^*)^2} \cdot \frac{(1-\alpha^*)^2}{K_t} \cdot \frac{2n-1}{(2n)^2}
= \frac{2n-1}{(2n)^2 K_t}
= \frac{1}{K_t} \cdot \frac{1}{2n}\!\left(1 - \frac{1}{2n}\right).
\end{align*}
The standard deviation in each entry is therefore $\mathrm{std}[(R_t)_{ab}] = O(1/\sqrt{K_t})$. By Chebyshev's inequality applied entry-wise, for any $\varepsilon > 0$
\[
\Pr\!\left[\left|(R_t)_{ab} - \frac{1}{2n}\right| > \varepsilon\right] \leq \frac{\mathrm{Var}[(R_t)_{ab}]}{\varepsilon^2} = O\!\left(\frac{1}{K_t \varepsilon^2}\right) \to 0 \quad \text{as } K_t \to \infty.
\]
Hence $R_t \to \frac{1}{2n}\mathbf{J}$ in probability, entry-wise.

\textit{Likelihood ratio diverges.}
The true residual $R$ satisfies $\|R - \frac{1}{2n}\mathbf{J}\|_\infty = O_P(1/\sqrt{K_t})$ (where $\|\cdot\|_\infty$ is the max entry). The false residual satisfies
\[
\left\|R' - \frac{1}{2n}\mathbf{J}\right\|_\infty = \left\|R - \frac{1}{2n}\mathbf{J} + \frac{\alpha^*}{1-\alpha^*}(M_t - M')\right\|_\infty.
\]
The shift matrix $\frac{\alpha^*}{1-\alpha^*}(M_t - M')$ has entries of magnitude $\frac{\alpha^*}{1-\alpha^*}$ at the positions where $M_t$ and $M'$ differ. Since $M_t \neq M'$ (they differ in at least 2 rows), we have
\[
\left\|R' - \frac{1}{2n}\mathbf{J}\right\|_\infty \geq \frac{\alpha^*}{1-\alpha^*} - O_P\!\left(\frac{1}{\sqrt{K_t}}\right).
\]
For $K_t$ large enough that $1/\sqrt{K_t} \ll \alpha^*/(1-\alpha^*)$, the false residual $R'$ lies at distance $\Theta(\alpha^*/(1-\alpha^*))$ from the mode $\frac{1}{2n}\mathbf{J}$, while $R$ lies at distance $O(1/\sqrt{K_t})$. Since $\nu$ concentrates with width $O(1/\sqrt{K_t})$, the false residual is $\Theta(\alpha^* \sqrt{K_t}/(1-\alpha^*))$ standard deviations from the mode. Therefore $\nu(R')/\nu(R) \to 0$ as $K_t \to \infty$ (assuming $\alpha^*$ is fixed), and the likelihood ratio diverges
\[
\frac{\nu(R)}{\nu(R')} \;\to\; \infty \qquad \text{as } K_t \to \infty.
\]
Applying Bayes' theorem with a uniform prior $\pi(M') = 1/(2n)!$ for all $M' \in S_{2n}$
\[
\Pr(M_t = M' \mid D_t) = \frac{\nu(R')}{\sum_{M'' \in S_{2n}} \nu(R'')} \;\to\; \begin{cases} 1 & \text{if } M' = M_t, \\ 0 & \text{if } M' \neq M_t, \end{cases}
\]
since the numerator for $M' = M_t$ dominates all other terms.
\end{proof}

\begin{remark}[More decoys can decrease security]
This creates a counterintuitive trade-off. Increasing $K_t$ (adding more decoy permutations) was intended to improve security by increasing the decomposition count $L$. However, increasing $K_t$ also concentrates $\nu$, making the likelihood ratio $\nu(R)/\nu(R')$ larger and the server's MAP estimate more accurate. Against a computationally unbounded adversary, the concentration effect dominates the decomposition-count effect.
\end{remark}

\section{The Two-Layer PolyVeil Protocol}
\label{sec:it_security}

The de-shuffling attack (Theorem~\ref{thm:deshuffle}) shows that no protocol variant in which the server sees identity-linked $f_t$ values \emph{and} shuffled $\eta_t$ values separately can be secure. The root cause is that the integrality constraint $s_t \in \{0,\ldots,n\}$ allows the server to uniquely identify the true shuffle permutation.

We now present a corrected protocol with a \emph{two-layer} architecture in which no single entity can learn individual data. The design addresses the reviewer critique that any aggregation-only protocol (where the server sees only $\sum f_t$ and $\sum \eta_t$) is trivially secure and does not require the Birkhoff polytope. In our two-layer protocol, the Birkhoff encoding is \emph{essential} for the security of the aggregation layer.

\subsection{Architecture Overview}

The protocol involves three types of entities: \textbf{Clients} $t \in [k]$, each holding private $\mathbf{b}_t \in \{0,1\}^n$; an \textbf{Aggregator} $\mathcal{A}$, which receives masked matrices $D_t$ and computes the scalar aggregate $F = \sum_t \mathbf{w}^T D_t \mathbf{y}$, which it sends to the server (the aggregator does \emph{not} receive $\eta_t$ values); and a \textbf{Server} $\mathcal{S}$, which receives $F$ from the aggregator and $H = \sum_t \eta_t$ from a separate noise-aggregation channel, and computes $S = (F - H)/\alpha^*$.

The key design principle is \emph{separation of information}: the aggregator sees $D_t$ (which encodes $M_t$) but not $\eta_t$; a separate channel delivers $\sum \eta_t$ to the server without the aggregator's involvement.

\begin{algorithm}[!ht]
\caption{Two-Layer PolyVeil}
\label{alg:twolayer}
\begin{algorithmic}[1]
\Statex \textbf{Public parameters:} $n$, $\alpha^* \in (0,1)$, $\mathbf{w}$, $\mathbf{y}$.
\Statex \textbf{Entities:} Aggregator $\mathcal{A}$, noise aggregator $\mathcal{B}$, server $\mathcal{S}$.
\Statex
\Statex \underline{\textsc{Client-Side Computation}}
\For{each client $t \in [k]$ in parallel}
    \State Encode $M_t = M(\mathbf{b}_t)$. Draw $K_t \geq 2$ uniform random permutations $P_{t,i} \in S_{2n}$.
    \State Sample $\alpha_{t,1}, \ldots, \alpha_{t,K_t} > 0$ with $\sum_i \alpha_{t,i} = 1 - \alpha^*$.
    \State Compute $D_t = \alpha^* M_t + \sum_i \alpha_{t,i} P_{t,i}$ and $\eta_t = \sum_i \alpha_{t,i}(\mathbf{w}^T P_{t,i}\mathbf{y})$.
    \State Send $D_t$ to aggregator $\mathcal{A}$. \quad Send $\eta_t$ to noise aggregator $\mathcal{B}$.
\EndFor
\Statex \underline{\textsc{Aggregation}}
\State $\mathcal{A}$ computes $F = \sum_{t=1}^k \mathbf{w}^T D_t \mathbf{y}$ and sends $F$ to server $\mathcal{S}$.
\State $\mathcal{B}$ computes $H = \sum_{t=1}^k \eta_t$ and sends $H$ to server $\mathcal{S}$.
\State $\mathcal{S}$ computes $S = (F - H)/\alpha^*$.
\end{algorithmic}
\end{algorithm}

\subsection{Correctness}

\begin{theorem}[Correctness]
\label{thm:correct_twolayer}
Algorithm~\ref{alg:twolayer} outputs $S = \sum_t s_t$ exactly.
\end{theorem}

\begin{proof}
The aggregator computes
\begin{align*}
F &= \sum_{t=1}^k \mathbf{w}^T D_t\, \mathbf{y}
= \sum_{t=1}^k \mathbf{w}^T \!\left(\alpha^* M_t + \sum_{i=1}^{K_t} \alpha_{t,i} P_{t,i}\right) \mathbf{y} \\
&= \sum_{t=1}^k \left(\alpha^*(\mathbf{w}^T M_t\, \mathbf{y}) + \sum_{i=1}^{K_t} \alpha_{t,i}(\mathbf{w}^T P_{t,i}\, \mathbf{y})\right)
= \sum_{t=1}^k (\alpha^* s_t + \eta_t)
= \alpha^* \sum_{t=1}^k s_t + \sum_{t=1}^k \eta_t
= \alpha^* S + H,
\end{align*}
where we used Lemma~\ref{lem:extract} ($\mathbf{w}^T M_t \mathbf{y} = s_t$) and the definition $\eta_t = \sum_i \alpha_{t,i}(\mathbf{w}^T P_{t,i}\mathbf{y})$. The noise aggregator computes $H = \sum_{t=1}^k \eta_t$ independently. The server receives $F$ and $H$ and computes
\[
\frac{F - H}{\alpha^*} = \frac{(\alpha^* S + H) - H}{\alpha^*} = \frac{\alpha^* S}{\alpha^*} = S.
\]
\end{proof}

\subsection{Layer 1: Information-Theoretic Security of the Server}
\label{sec:layer1}

The server $\mathcal{S}$ receives only two scalars: the aggregate signal $F = \sum_t f_t$ and the aggregate noise $H = \sum_t \eta_t$. We prove that the server learns nothing about any individual client's data beyond the aggregate $S = \sum_t s_t$, using the simulation paradigm from secure multi-party computation (background in Appendix~\ref{app:sim}).

The idea of a simulation proof is simple: we construct an algorithm (the \emph{simulator}) that can fabricate a fake server view using only the aggregate $S$ and public parameters --- \emph{without knowing any individual input} $\mathbf{b}_t$. If the fabricated view is distributed identically to the real view, then the real view contains no information about individual inputs beyond $S$, because anything the server could compute from the real view, it could equally compute from the simulator's output (which depends only on $S$).

Figure~\ref{fig:sim_proof} illustrates the proof structure.

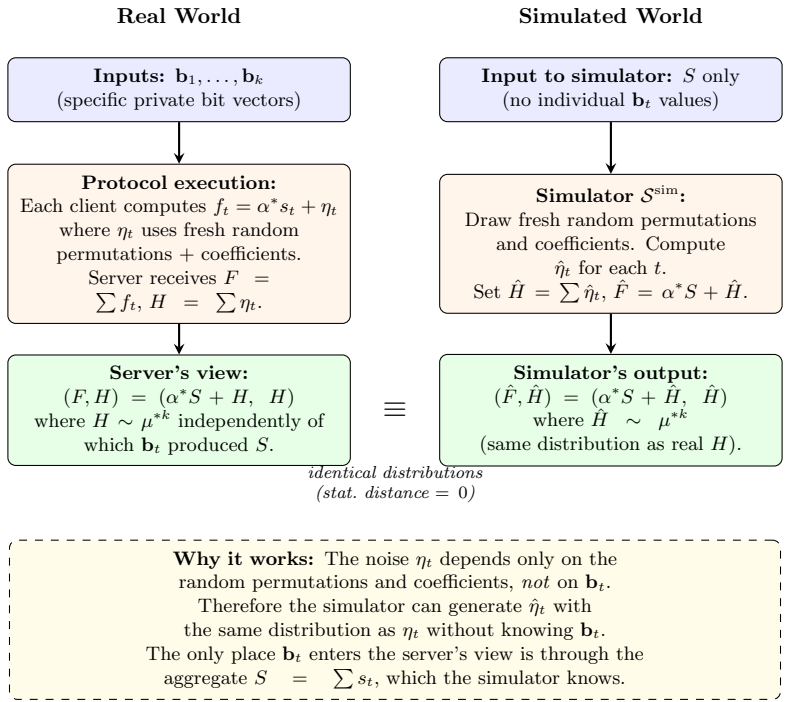
\begin{figure}[!htbp]
\begin{center}
\begin{tikzpicture}[scale=0.82, every node/.style={transform shape},
    box/.style={draw, rounded corners=3pt, minimum width=5.5cm, align=center, font=\footnotesize, text width=5.2cm, inner sep=5pt},
    widebox/.style={draw, rounded corners=3pt, minimum width=12.5cm, align=center, font=\footnotesize, text width=12cm, inner sep=5pt},
    arrow/.style={->, >=stealth, thick},
]

\node[font=\small\bfseries] at (-3.5, 1.2) {Real World};
\node[box, fill=blue!8] (real_in) at (-3.5, 0) {
\textbf{Inputs:} $\mathbf{b}_1, \ldots, \mathbf{b}_k$\\
(specific private bit vectors)
};
\node[box, fill=orange!8] (real_proc) at (-3.5, -2.5) {
\textbf{Protocol execution:}\\
Each client computes $f_t = \alpha^* s_t + \eta_t$\\
where $\eta_t$ uses fresh random\\
permutations + coefficients.\\
Server receives $F = \sum f_t$, $H = \sum \eta_t$.
};
\node[box, fill=green!10] (real_view) at (-3.5, -5.2) {
\textbf{Server's view:}\\
$(F, H) = (\alpha^* S + H,\; H)$\\
where $H \sim \mu^{*k}$ independently of\\
which $\mathbf{b}_t$ produced $S$.
};
\draw[arrow] (real_in) -- (real_proc);
\draw[arrow] (real_proc) -- (real_view);

\node[font=\small\bfseries] at (3.5, 1.2) {Simulated World};
\node[box, fill=blue!8] (sim_in) at (3.5, 0) {
\textbf{Input to simulator:} $S$ only\\
(no individual $\mathbf{b}_t$ values)
};
\node[box, fill=orange!8] (sim_proc) at (3.5, -2.5) {
\textbf{Simulator $\mathcal{S}^{\mathrm{sim}}$:}\\
Draw fresh random permutations\\
and coefficients. Compute\\
$\hat\eta_t$ for each $t$.\\
Set $\hat{H} = \sum \hat\eta_t$, $\hat{F} = \alpha^* S + \hat{H}$.
};
\node[box, fill=green!10] (sim_view) at (3.5, -5.2) {
\textbf{Simulator's output:}\\
$(\hat{F}, \hat{H}) = (\alpha^* S + \hat{H},\; \hat{H})$\\
where $\hat{H} \sim \mu^{*k}$\\
(same distribution as real $H$).
};
\draw[arrow] (sim_in) -- (sim_proc);
\draw[arrow] (sim_proc) -- (sim_view);

\node[font=\Large] at (0, -5.2) {$\equiv$};
\node[font=\scriptsize\itshape, text width=3.5cm, align=center] at (0, -6.4) {identical distributions\\(stat.\ distance $= 0$)};

\node[widebox, fill=yellow!10, dashed] (key) at (0, -8.6) {
\textbf{Why it works:} The noise $\eta_t$ depends only on the random permutations and coefficients, \emph{not} on $\mathbf{b}_t$.\\
Therefore the simulator can generate $\hat\eta_t$ with the same distribution as $\eta_t$ without knowing $\mathbf{b}_t$.\\
The only place $\mathbf{b}_t$ enters the server's view is through the aggregate $S = \sum s_t$, which the simulator knows.
};

\end{tikzpicture}
\end{center}
\caption{Structure of the simulation proof for the server's information-theoretic security (Theorem~\ref{thm:server_it}). \textbf{Left:} the real protocol execution with actual private inputs. \textbf{Right:} the simulator, which knows only the aggregate $S$ (not individual $\mathbf{b}_t$) and fabricates a view with identical distribution. The equivalence $\equiv$ means identical distributions, holding against computationally unbounded adversaries. The dashed box states the key structural property that makes the proof work.}
\label{fig:sim_proof}
\end{figure}

\begin{theorem}[Perfect simulation-based security of the server]
\label{thm:server_it}
Let $\mathcal{D} = (\mathbf{b}_1, \ldots, \mathbf{b}_k)$ and $\mathcal{D}' = (\mathbf{b}_1', \ldots, \mathbf{b}_k')$ be any two input configurations satisfying $\sum_t s_t = \sum_t s_t' = S$. Then:
\begin{enumerate}
\item[\emph{(i)}] The server's view $(F, H)$ under $\mathcal{D}$ and $(F, H)$ under $\mathcal{D}'$ have identical distributions.
\item[\emph{(ii)}] There exists a simulator $\mathcal{S}^{\mathrm{sim}}$ that, given only $S$ and the public parameters $(k, n, \alpha^*)$, outputs a pair $(\hat{F}, \hat{H})$ with $(\hat{F}, \hat{H}) \equiv (F, H)$ (identical distributions) for all inputs with aggregate $S$.
\end{enumerate}
\end{theorem}

The proof proceeds in four stages.

\begin{proof}

\medskip\noindent\textbf{Stage 1: The noise $\eta_t$ is independent of the private input $\mathbf{b}_t$.}

This is the structural property on which the entire proof rests. Each client's noise is
\[
\eta_t = \sum_{i=1}^{K_t} \alpha_{t,i}\,(\mathbf{w}^T P_{t,i}\, \mathbf{y}),
\]
where $P_{t,1}, \ldots, P_{t,K_t}$ are drawn uniformly and independently from $S_{2n}$, and the coefficients $(\alpha_{t,1}, \ldots, \alpha_{t,K_t})$ are drawn from a continuous distribution $g$ on the simplex $\Delta_{K_t} = \{(\alpha_1, \ldots, \alpha_{K_t}) : \alpha_i > 0,\; \sum_i \alpha_i = 1-\alpha^*\}$. Crucially, neither the permutations $P_{t,i}$ nor the coefficients $\alpha_{t,i}$ depend on $\mathbf{b}_t$ --- they are drawn from distributions determined entirely by the public parameters $(n, K_t, \alpha^*)$.

Formally, let $\Omega_t = S_{2n}^{K_t} \times \Delta_{K_t}$ denote the probability space of client $t$'s randomness, equipped with the product of the uniform measure on $S_{2n}^{K_t}$ and the distribution $g$ on $\Delta_{K_t}$. The map $\omega_t = (P_{t,1}, \ldots, P_{t,K_t}, \alpha_{t,1}, \ldots, \alpha_{t,K_t}) \mapsto \eta_t(\omega_t) = \sum_i \alpha_{t,i}(\mathbf{w}^T P_{t,i}\mathbf{y})$ is a measurable function of $\omega_t$ alone. Since $\omega_t$ is sampled from a distribution that does not involve $\mathbf{b}_t$, we have
\begin{equation}
\label{eq:eta_indep}
\eta_t \;\perp\!\!\!\perp\; \mathbf{b}_t \qquad\text{(statistical independence)}.
\end{equation}
Let $\mu$ denote the distribution of $\eta_t$. This distribution is the same for every client (since the public parameters are shared) and does not depend on any private input. Since the randomness is also independent across clients, the noise values $\eta_1, \ldots, \eta_k$ are mutually independent:
\begin{equation}
\label{eq:eta_iid_expanded}
(\eta_1, \ldots, \eta_k) \;\sim\; \mu^{\otimes k} \qquad\text{(i.i.d.\ from $\mu$, independent of $\mathcal{D}$)}.
\end{equation}

\medskip\noindent\textbf{Stage 2: The server's view is a deterministic function of $(S, H)$.}

The server receives two scalars. The first is
\begin{align}
F &= \sum_{t=1}^k f_t = \sum_{t=1}^k \mathbf{w}^T D_t \mathbf{y} = \sum_{t=1}^k (\alpha^* s_t + \eta_t) = \alpha^* \underbrace{\sum_{t=1}^k s_t}_{= S} + \underbrace{\sum_{t=1}^k \eta_t}_{= H} = \alpha^* S + H. \label{eq:F_SH}
\end{align}
The second is $H = \sum_{t=1}^k \eta_t$ directly. Therefore the server's complete view is the pair
\begin{equation}
\label{eq:server_view}
\text{View}_{\text{server}} = (F,\; H) = (\alpha^* S + H,\;\; H).
\end{equation}
This is a deterministic, invertible function of $(S, H)$: given $(F, H)$, one recovers $S = (F - H)/\alpha^*$, and given $(S, H)$, one recovers $F = \alpha^* S + H$. The only randomness in the view comes from $H$.

\medskip\noindent\textbf{Stage 3: The distribution of the view depends only on $S$, not on which $\mathcal{D}$ produced $S$.}

By Stage~1, $H = \sum_{t=1}^k \eta_t$ where $\eta_1, \ldots, \eta_k \sim \mu^{\otimes k}$ independently of $\mathcal{D}$. The distribution of $H$ is therefore the $k$-fold convolution $\mu^{*k}$:
\begin{equation}
\label{eq:H_dist_expanded}
\Pr[H \in B \mid \mathcal{D}] = \Pr[H \in B] = \mu^{*k}(B), \qquad \forall\; \mathcal{D},\;\; \forall \text{ measurable } B \subseteq \mathbb{R}.
\end{equation}
By Stage~2, $F = \alpha^* S + H$. For any measurable set $A \subseteq \mathbb{R}^2$:
\begin{equation}
\label{eq:view_integral_expanded}
\Pr\!\left[(F, H) \in A \mid \mathcal{D}\right] = \int_{\mathbb{R}} \mathbf{1}\!\big[(\alpha^* S + h,\;\; h) \in A\big]\; d\mu^{*k}(h).
\end{equation}
Now consider a different input $\mathcal{D}'$ with $\sum_t s_t' = S$ (the same aggregate). By exactly the same argument:
\[
\Pr\!\left[(F, H) \in A \mid \mathcal{D}'\right] = \int_{\mathbb{R}} \mathbf{1}\!\big[(\alpha^* S + h,\;\; h) \in A\big]\; d\mu^{*k}(h).
\]
The right-hand sides are identical: the integrand depends only on $S$ (which is the same for $\mathcal{D}$ and $\mathcal{D}'$), and the measure $\mu^{*k}$ is independent of the input configuration. Therefore
\[
\Pr\!\left[(F, H) \in A \mid \mathcal{D}\right] = \Pr\!\left[(F, H) \in A \mid \mathcal{D}'\right], \qquad \forall \text{ measurable } A,
\]
which is statement~(i): identical distributions. The statistical distance between the two views is exactly zero.

\medskip\noindent\textbf{Stage 4: Constructing the simulator.}

We now build the simulator $\mathcal{S}^{\mathrm{sim}}$ that produces a fake view from $S$ alone:

\begin{enumerate}
\item[\textbf{Input:}] The aggregate $S$ and public parameters $(k, n, \alpha^*, K_t)$.
\item[\textbf{(a)}] For each $t = 1, \ldots, k$: draw $K_t$ independent uniform random permutations $\hat{P}_{t,i} \in S_{2n}$ and coefficients $(\hat\alpha_{t,1}, \ldots, \hat\alpha_{t,K_t})$ from the distribution $g$ on $\Delta_{K_t}$.
\item[\textbf{(b)}] Compute $\hat\eta_t = \sum_{i=1}^{K_t} \hat\alpha_{t,i}\,(\mathbf{w}^T \hat{P}_{t,i}\, \mathbf{y})$ for each $t$.
\item[\textbf{(c)}] Compute $\hat{H} = \sum_{t=1}^k \hat\eta_t$.
\item[\textbf{(d)}] Compute $\hat{F} = \alpha^* S + \hat{H}$.
\item[\textbf{Output:}] $(\hat{F},\; \hat{H})$.
\end{enumerate}

The simulator uses fresh randomness (the $\hat{P}_{t,i}$ and $\hat\alpha_{t,i}$) that is independent of the actual protocol execution. It does not know any $\mathbf{b}_t$, any $s_t$, any $D_t$, or any $\eta_t$. It knows only $S$.

We verify that the simulator's output has the correct distribution. By construction, $\hat\eta_1, \ldots, \hat\eta_k$ are i.i.d.\ from $\mu$ (since each is generated by the same random process as the real $\eta_t$). Therefore $\hat{H} = \sum_t \hat\eta_t$ has distribution $\mu^{*k}$, and
\[
(\hat{F},\; \hat{H}) = (\alpha^* S + \hat{H},\;\; \hat{H}).
\]
Comparing with~\eqref{eq:view_integral_expanded}: for any measurable $A$,
\begin{align*}
\Pr\!\left[(\hat{F}, \hat{H}) \in A\right]
&= \int_{\mathbb{R}} \mathbf{1}\!\big[(\alpha^* S + h,\;\; h) \in A\big]\; d\mu^{*k}(h) \\
&= \Pr\!\left[(F, H) \in A \mid \mathcal{D}\right],
\end{align*}
for \emph{any} input $\mathcal{D}$ with aggregate $S$. This is statement~(ii): $\mathcal{S}^{\mathrm{sim}}(S) \equiv (F, H)$.
\end{proof}

\begin{remark}[What the simulator proof means concretely]
\label{rem:sim_meaning}
The simulator demonstrates that everything the server sees --- the aggregate signal $F$ and the aggregate noise $H$ --- could have been generated by an algorithm that knows nothing about any individual client. The server cannot distinguish the real protocol (where actual clients submitted actual private data) from the simulated protocol (where a single machine fabricated fake aggregates from $S$ alone). Any function the server computes on its real view --- any test statistic, any machine learning model, any side-channel analysis --- it could equally well compute on the simulator's output, which contains zero individual-level information. This is information-theoretic: it holds against adversaries with unlimited computational power, unlimited memory, and unlimited time.
\end{remark}

\begin{remark}[Structural source of security]
\label{rem:structural}
The proof relies on exactly one structural property: $\eta_t \perp\!\!\!\perp \mathbf{b}_t$ (the noise is independent of the private data). Any noise generation process with this property would give the server the same information-theoretic guarantee. The Birkhoff encoding, the permutation matrices, the BvN decomposition --- none of these are needed for the server's security. They are needed only for the aggregator's computational barrier (Layer~2). The server's security would hold even if the noise were Gaussian, Laplace, or any other distribution, as long as it is independent of $\mathbf{b}_t$ and cancels exactly in the aggregate.
\end{remark}

\subsection{Layer 2: Computational Security of the Aggregator via \#P-Hardness}

The aggregator $\mathcal{A}$ sees the individual matrices $D_1, \ldots, D_k$ but does \emph{not} receive $\eta_t$ or $f_t$ as separate scalars. To recover client $t$'s data, the aggregator must extract $M_t$ from $D_t = \alpha^* M_t + (1-\alpha^*) R_t$. We show that the natural approach --- computing the posterior distribution over candidate permutation matrices --- requires solving \#P-hard problems. The argument connects the aggregator's density evaluation to three classical hard problems: the permanent, perfect matchings in bipartite graphs, and the mixed discriminant. Figure~\ref{fig:hardness_roadmap} provides a roadmap of the logical structure; the rest of this section fills in every detail.

\begin{figure}[!htbp]
\begin{center}
\begin{tikzpicture}[scale=0.86, every node/.style={transform shape},
    box/.style={draw, rounded corners=3pt, minimum width=12.5cm, align=center, font=\footnotesize, text width=12cm, inner sep=6pt},
    arrow/.style={->, >=stealth, thick},
]

\node[box, fill=blue!8] (s1) at (0, 0) {
\textbf{1. The aggregator's problem.}\quad
Aggregator observes $D_t = \alpha^* M_t + (1 - \alpha^*) R_t$ and wants to recover $M_t$.
};

\node[box, fill=blue!8] (s2) at (0, -2.1) {
\textbf{2. Bayesian inference requires evaluating $\nu(R')$.}\quad
$\Pr(M_t {=} M' \mid D_t) \propto \nu(R')$, where $R' = (D_t - \alpha^* M')/(1{-}\alpha^*)$ and $\nu$ is the density of the random decoy matrix $R_t$ on $\mathcal{B}_{2n}$. The MAP estimator picks the $M'$ maximizing $\nu(R')$.
};
\draw[arrow] (s1) -- (s2);

\node[box, fill=orange!8] (s3) at (0, -4.8) {
\textbf{3. The density is a sum over all permutation tuples.}\quad
$R_t$ is built from $K$ random permutations and random coefficients. The density at $R'$ sums over \emph{every} possible tuple, weighted by the probability the coefficients produce $R'$:\quad
$\displaystyle\nu(R') = \sum_{\text{all } ((2n)!)^K \text{ tuples}} \frac{1}{((2n)!)^K} \times \big(\text{integral of } g(\alpha) \text{ over } \alpha \text{ satisfying } \textstyle\sum_i \alpha_i P_{\sigma_i} = (1{-}\alpha^*)R'\big)$
};
\draw[arrow] (s2) -- (s3);

\node[box, fill=orange!8] (s4) at (0, -7.5) {
\textbf{4. Most tuples contribute zero.}\quad
The constraint $\sum_i \alpha_i P_{\sigma_i} = (1{-}\alpha^*) R'$ with all $\alpha_i > 0$ forces: wherever $R'_{ab} = 0$, \emph{no} permutation in the tuple can have a 1 at $(a,b)$. Each $P_{\sigma_i}$ must ``fit inside'' the positive entries of $R'$. Only tuples where every permutation lies in the \emph{support set} $\mathrm{Supp}(R')$ contribute.
};
\draw[arrow] (s3) -- (s4);

\node[box, fill=orange!12] (s5) at (0, -10.6) {
\textbf{5. Rewrite the density using only valid tuples.}
$$\nu(R') \;=\; \frac{C_K}{((2n)!)^K} \sum_{(\sigma_1,\ldots,\sigma_K)\,\in\,\mathrm{Supp}(R')^K} \mathrm{vol}\big(\mathcal{P}(\sigma_1,\ldots,\sigma_K;\;R')\big)$$
This is a sum of polytope volumes, one per valid tuple. Two questions arise:\\
\textbf{(a)} How many terms are in this sum?\quad \textbf{(b)} What is each term (the polytope volume)?\\
These are answered in Figure~\ref{fig:hardness_branches}.
};
\draw[arrow] (s4) -- (s5);

\end{tikzpicture}
\end{center}
\caption{The density derivation: from the aggregator's inference problem to the key formula. The aggregator's MAP estimator requires evaluating $\nu(R')$ (Steps~1--2). The density marginalizes over all $((2n)!)^K$ permutation tuples (Step~3), but most contribute zero --- only tuples in $\mathrm{Supp}(R')^K$ survive (Step~4). The resulting formula (Step~5) is a sum of polytope volumes over valid tuples, raising two questions answered in Figure~\ref{fig:hardness_branches}.}
\label{fig:hardness_roadmap}
\end{figure}
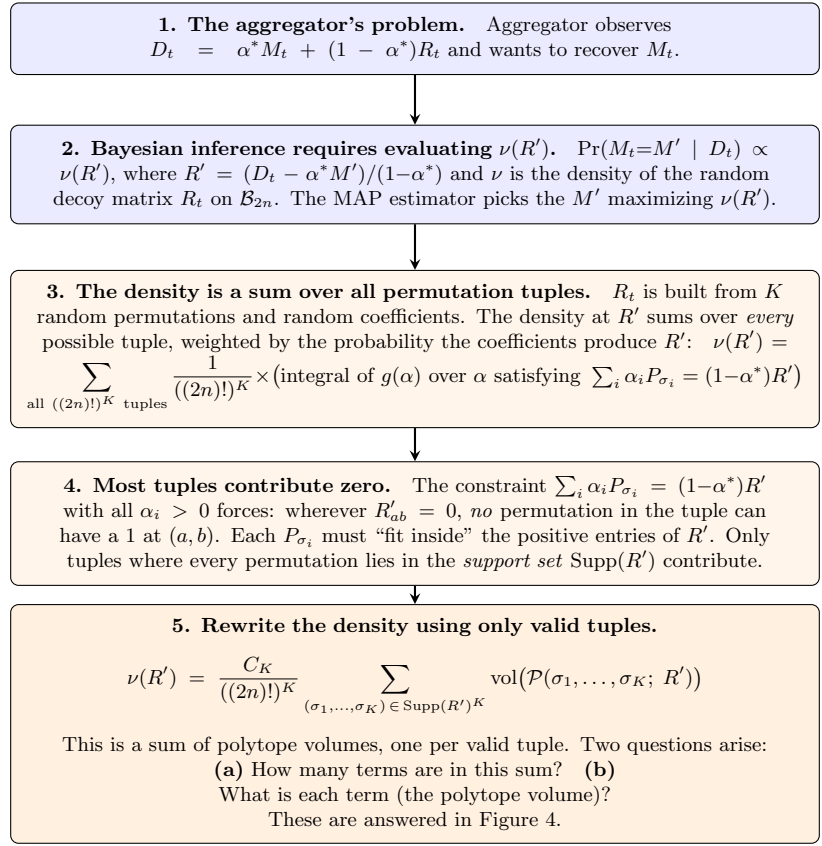

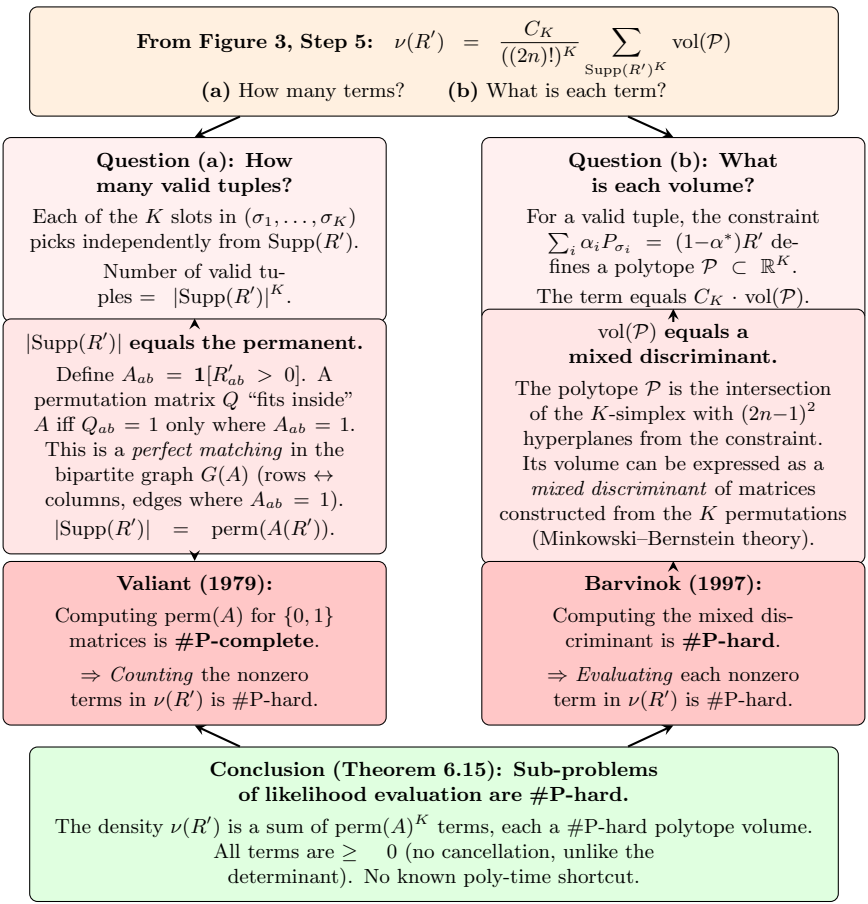
\begin{figure}[!htbp]
\begin{center}
\begin{tikzpicture}[scale=0.86, every node/.style={transform shape},
    splitbox/.style={draw, rounded corners=3pt, minimum width=5.8cm, align=center, font=\footnotesize, text width=5.5cm, inner sep=6pt},
    widebox/.style={draw, rounded corners=3pt, minimum width=12.5cm, align=center, font=\footnotesize, text width=12cm, inner sep=6pt},
    arrow/.style={->, >=stealth, thick},
]

\node[widebox, fill=orange!12] (start) at (0, 0) {
\textbf{From Figure~\ref{fig:hardness_roadmap}, Step~5:}\quad
$\displaystyle\nu(R') = \frac{C_K}{((2n)!)^K} \sum_{\mathrm{Supp}(R')^K} \mathrm{vol}(\mathcal{P})$\qquad
\textbf{(a)} How many terms?\qquad \textbf{(b)} What is each term?
};

\node[splitbox, fill=red!6] (qa) at (-3.7, -2.6) {
\textbf{Question (a): How many valid tuples?}\\[3pt]
Each of the $K$ slots in $(\sigma_1, \ldots, \sigma_K)$
picks independently from $\mathrm{Supp}(R')$.\\[2pt]
Number of valid tuples $= |\mathrm{Supp}(R')|^K$.
};

\node[splitbox, fill=red!6] (qb) at (3.7, -2.6) {
\textbf{Question (b): What is each volume?}\\[3pt]
For a valid tuple, the constraint
$\sum_i \alpha_i P_{\sigma_i} = (1{-}\alpha^*)R'$
defines a polytope $\mathcal{P} \subset \mathbb{R}^K$.\\[2pt]
The term equals $C_K \cdot \mathrm{vol}(\mathcal{P})$.
};

\draw[arrow] ([xshift=-3cm]start.south) -- (qa.north);
\draw[arrow] ([xshift=3cm]start.south) -- (qb.north);

\node[splitbox, fill=red!10] (permatch) at (-3.7, -5.8) {
\textbf{$|\mathrm{Supp}(R')|$ equals the permanent.}\\[3pt]
Define $A_{ab} = \mathbf{1}[R'_{ab} > 0]$. A permutation
matrix $Q$ ``fits inside'' $A$ iff $Q_{ab} = 1$
only where $A_{ab} = 1$. This is a
\emph{perfect matching} in the bipartite
graph $G(A)$ (rows $\leftrightarrow$ columns,
edges where $A_{ab} = 1$).\\[2pt]
$|\mathrm{Supp}(R')| \;=\; \mathrm{perm}(A(R'))$.
};
\draw[arrow] (qa) -- (permatch);

\node[splitbox, fill=red!10] (mixdisc) at (3.7, -5.8) {
\textbf{$\mathrm{vol}(\mathcal{P})$ equals a mixed discriminant.}\\[3pt]
The polytope $\mathcal{P}$ is the intersection
of the $K$-simplex with $(2n{-}1)^2$
hyperplanes from the constraint.
Its volume can be expressed as a
\emph{mixed discriminant} of matrices
constructed from the $K$ permutations
(Minkowski--Bernstein theory).
};
\draw[arrow] (qb) -- (mixdisc);

\node[splitbox, fill=red!22] (valiant) at (-3.7, -9) {
\textbf{Valiant (1979):}\\[3pt]
Computing $\mathrm{perm}(A)$ for $\{0,1\}$
matrices is \textbf{\#P-complete}.\\[4pt]
$\Rightarrow$ \emph{Counting} the nonzero terms
in $\nu(R')$ is \#P-hard.
};
\draw[arrow] (permatch) -- (valiant);

\node[splitbox, fill=red!22] (barvinok) at (3.7, -9) {
\textbf{Barvinok (1997):}\\[3pt]
Computing the mixed discriminant
is \textbf{\#P-hard}.\\[4pt]
$\Rightarrow$ \emph{Evaluating} each nonzero term
in $\nu(R')$ is \#P-hard.
};
\draw[arrow] (mixdisc) -- (barvinok);

\node[widebox, fill=green!12] (conc) at (0, -11.8) {
\textbf{Conclusion (Theorem~\ref{thm:likelihood_hardness}): Sub-problems of likelihood evaluation are \#P-hard.}\\[3pt]
The density $\nu(R')$ is a sum of $\mathrm{perm}(A)^K$ terms, each a \#P-hard polytope volume.\\
All terms are $\geq 0$ (no cancellation, unlike the determinant). No known poly-time shortcut.
};
\draw[arrow] ([xshift=-3cm]conc.north) -- (valiant.south);
\draw[arrow] ([xshift=3cm]conc.north) -- (barvinok.south);

\end{tikzpicture}
\end{center}
\caption{The two branches of the \#P-hardness argument, continuing from the density formula in Figure~\ref{fig:hardness_roadmap}. \textbf{Left branch} (Question~a): the number of valid tuples equals the permanent of the support matrix $A(R')$, which counts perfect matchings in a bipartite graph; computing this is \#P-complete (Valiant, 1979). \textbf{Right branch} (Question~b): each valid tuple's contribution is a polytope volume expressible as a mixed discriminant; computing this is \#P-hard (Barvinok, 1997). Both branches converge: the density is a \#P-hard number of individually \#P-hard terms, with no cancellation.}
\label{fig:hardness_branches}
\end{figure}

We begin by establishing notation. Throughout this section, $S_m$ denotes the \emph{symmetric group} on $m$ elements, i.e., the set of all $m!$ bijections $\sigma: \{1, \ldots, m\} \to \{1, \ldots, m\}$. We identify each bijection $\sigma \in S_m$ with the $m \times m$ permutation matrix $P_\sigma$ defined by $(P_\sigma)_{ab} = \mathbf{1}[\sigma(a) = b]$, which has exactly one 1 in each row and each column and zeros everywhere else. The notation $S_{2n}$ thus refers to the set of all $(2n)!$ permutation matrices of size $2n \times 2n$, and $S_{2n}^K = S_{2n} \times \cdots \times S_{2n}$ ($K$ times) is the set of all ordered $K$-tuples of such matrices.

\subsubsection{The Aggregator's Inference Problem}

Given $D_t$, the aggregator seeks $M_t$. Recall that $D_t = \alpha^* M_t + (1-\alpha^*) R_t$, so for any candidate permutation matrix $M'$, the aggregator can compute the \emph{residual} $R' = (D_t - \alpha^* M')/(1-\alpha^*)$. If $M'$ happens to be the true encoding $M_t$, then $R' = R_t$ (the actual random decoy matrix, which is doubly stochastic by construction). If $M'$ is wrong, $R'$ may or may not be doubly stochastic. The set of consistent candidates is therefore
\[
\mathcal{L}(D_t, \alpha^*) = \{M' \in S_{2n} : R' := (D_t - \alpha^* M')/(1-\alpha^*) \in \mathcal{B}_{2n}\}.
\]
The condition $R' \in \mathcal{B}_{2n}$ requires every entry $R'_{ab}$ to be non-negative (the row and column sum constraints are automatically satisfied since both $D_t$ and $M'$ are doubly stochastic). Since $(M')_{ab} \in \{0,1\}$, the non-negativity condition $R'_{ab} \geq 0$ is equivalent to $(D_t)_{ab} \geq \alpha^*$ at every position $(a,b)$ where $M'$ has a 1.

When $\alpha^*$ is small and the decoy component covers all matrix entries (which happens with high probability for moderate $K_t$), the smallest entry of $D_t$ exceeds $\alpha^*$ and \emph{every} permutation matrix in $S_{2n}$ is consistent: $\mathcal{L}(D_t, \alpha^*) = S_{2n}$. In this interior regime, consistency alone provides no information --- all $(2n)!$ candidates look equally valid.

The aggregator's best strategy, given unlimited computation, is Bayesian inference. Assuming a uniform prior over $M' \in S_{2n}$, the posterior probability of each candidate is
\begin{equation}
\label{eq:agg_posterior}
\Pr(M_t = M' \mid D_t) \;=\; \frac{\nu(R')}{\sum_{M'' \in \mathcal{L}} \nu(R'')},
\end{equation}
where $\nu(R')$ denotes the probability density of the random decoy matrix $R_t$ evaluated at the point $R'$. The aggregator's MAP (maximum a posteriori) estimator picks the candidate $M'$ that maximizes $\nu(R')$. We now show that evaluating $\nu(R')$ is \#P-hard.

\begin{remark}[Why focus on likelihood, and what this does not cover]
\label{rem:why_likelihood}
The Bayesian/likelihood approach is the \emph{statistically optimal} attack: given unlimited computation, no other method can recover $M_t$ with higher probability. Proving it \#P-hard therefore eliminates the strongest possible attack strategy. However, this does \emph{not} rule out weaker but computationally efficient attacks that bypass density evaluation entirely. An aggregator might attempt spectral decomposition of $D_t$, solve a linear program to find a sparse BvN decomposition, run the Hungarian algorithm on $D_t/\alpha^*$ to find the nearest permutation matrix, or train a neural network on synthetic $(D_t, M_t)$ pairs. None of these require evaluating $\nu$, and our \#P-hardness result says nothing about them. We analyze several such non-likelihood attacks in Appendix~\ref{app:attacks} and show that they fail at the protocol's operating parameters, but we do not prove a blanket impossibility result for all polynomial-time attacks. This gap is the content of Conjecture~\ref{conj:hardness}, which remains open.
\end{remark}

\subsubsection{The Density of the Decoy Component}
\label{sec:density_derivation}

The decoy matrix $R_t$ is not drawn from a simple named distribution; it is constructed by a multi-step random process. Understanding how $\nu(R')$ arises from this process is essential to the hardness argument, so we derive the formula in detail.

\begin{definition}[The random process that generates $R_t$]
\label{def:random_ds}
The decoy matrix $R_t$ is generated in two steps. First, draw $K$ permutation matrices $P_{\sigma_1}, \ldots, P_{\sigma_K}$ independently and uniformly at random from $S_{2n}$ (each $\sigma_i$ is a uniformly random bijection on $\{1, \ldots, 2n\}$). Second, draw positive coefficients $(\alpha_1, \ldots, \alpha_K)$ from a continuous distribution $g$ on the simplex $\Delta_K = \{(\alpha_1, \ldots, \alpha_K) : \alpha_i > 0, \;\sum_i \alpha_i = 1-\alpha^*\}$. The decoy matrix is then $R_t = \frac{1}{1-\alpha^*}\sum_{i=1}^K \alpha_i P_{\sigma_i}$.
\end{definition}

The density $\nu(R')$ is the probability density of $R_t$ at a specific point $R' \in \mathcal{B}_{2n}$. To compute it, we must account for \emph{every possible way} the random process could have produced $R'$. There are two sources of randomness --- the permutation tuple $(\sigma_1, \ldots, \sigma_K)$ and the coefficient vector $(\alpha_1, \ldots, \alpha_K)$ --- and we must sum (over the discrete permutation choices) and integrate (over the continuous coefficient choices) over all combinations that yield $R'$.

Consider a fixed permutation tuple $(\sigma_1, \ldots, \sigma_K)$. The probability that this specific tuple is drawn is $1/((2n)!)^K$ (since each of the $K$ permutations is drawn independently and uniformly from the $(2n)!$-element set $S_{2n}$). Given this tuple, the decoy matrix equals $R'$ if and only if the coefficient vector $\alpha$ satisfies $\frac{1}{1-\alpha^*}\sum_{i=1}^K \alpha_i P_{\sigma_i} = R'$, or equivalently
\begin{equation}
\label{eq:constraint}
\sum_{i=1}^K \alpha_i P_{\sigma_i} = (1-\alpha^*) R'.
\end{equation}
This is a system of $(2n)^2$ linear equations in $K$ unknowns (the coefficients $\alpha_1, \ldots, \alpha_K$). The indicator function $\mathbf{1}[\sum_i \alpha_i P_{\sigma_i} = (1-\alpha^*) R']$ is 1 if and only if $\alpha$ satisfies all of these equations simultaneously, and 0 otherwise.

The density contribution from this particular tuple is therefore the integral of the coefficient density $g(\alpha)$ over all coefficient vectors that satisfy the constraint~\eqref{eq:constraint}, weighted by the probability $1/((2n)!)^K$ of selecting this tuple.

Summing over all $((2n)!)^K$ possible tuples gives the total density
\begin{equation}
\label{eq:nu_marginal}
\nu(R') = \sum_{(\sigma_1, \ldots, \sigma_K) \in S_{2n}^K} \frac{1}{((2n)!)^K} \int_{\Delta_K} \mathbf{1}\!\left[\sum_{i=1}^K \alpha_i P_{\sigma_i} = (1-\alpha^*) R'\right] \cdot g(\alpha_1, \ldots, \alpha_K) \; d\alpha.
\end{equation}
This formula has three components. The outer sum ranges over all $((2n)!)^K$ ordered $K$-tuples of permutations. The factor $1/((2n)!)^K$ is the probability of each tuple. The inner integral, weighted by the indicator, computes the probability that the random coefficients produce exactly $R'$ given the permutation tuple. For most tuples, the constraint~\eqref{eq:constraint} has no solution (the indicator is zero everywhere on $\Delta_K$), and the integral vanishes. The density $\nu(R')$ is therefore a sum of $((2n)!)^K$ terms, the vast majority of which are zero. We now determine exactly which tuples yield nonzero terms.

\subsubsection{Connection to the Permanent}

The key question is: for which tuples $(\sigma_1, \ldots, \sigma_K)$ does the constraint~\eqref{eq:constraint} have a feasible solution with $\alpha \in \Delta_K$ (all $\alpha_i > 0$, $\sum \alpha_i = 1-\alpha^*$)? The answer connects the density formula to the permanent of a $\{0,1\}$ matrix, and through it to the problem of counting perfect matchings in a bipartite graph.

\begin{definition}[Permanent]
\label{def:permanent}
For an $m \times m$ matrix $A = (a_{ij})$, the permanent is
\begin{equation}
\label{eq:perm_def}
\mathrm{perm}(A) = \sum_{\sigma \in S_m} \prod_{i=1}^m a_{i,\sigma(i)}.
\end{equation}
This formula is syntactically identical to the determinant, except that the determinant includes a sign factor $\mathrm{sgn}(\sigma) \in \{+1, -1\}$ in each term. Despite this superficial similarity, the permanent and determinant have vastly different computational properties: the determinant can be computed in $O(m^3)$ time by Gaussian elimination (because the alternating signs create cancellations that can be exploited), while the permanent has no known polynomial-time algorithm.
\end{definition}

\begin{theorem}[Valiant, 1979~\cite{valiant1979}]
\label{thm:valiant}
Computing $\mathrm{perm}(A)$ for $\{0,1\}$ matrices is \#P-complete.
\end{theorem}

The complexity class \#P consists of counting problems associated with NP decision problems: ``how many satisfying assignments does a Boolean formula have?'' is a \#P problem, for example. A problem is \#P-complete if every \#P problem can be reduced to it. This is strictly stronger than NP-hardness: by Toda's theorem, if $\mathrm{P} = \text{\#P}$ then the entire polynomial hierarchy collapses to $\mathrm{P}$, which is considered extremely unlikely.

The permanent has a natural graph-theoretic interpretation. For a $\{0,1\}$ matrix $A$ of size $m \times m$, define the bipartite graph $G(A)$ with $m$ left vertices (rows), $m$ right vertices (columns), and an edge from left vertex $i$ to right vertex $j$ whenever $A_{ij} = 1$. A \emph{perfect matching} in $G(A)$ is a set of $m$ edges that pairs every left vertex with a distinct right vertex --- equivalently, a bijection $\sigma: [m] \to [m]$ such that $A_{i,\sigma(i)} = 1$ for every $i$. The product $\prod_{i=1}^m A_{i,\sigma(i)}$ equals 1 if and only if $\sigma$ defines such a matching (since every factor must be 1), and 0 otherwise. Therefore $\mathrm{perm}(A) = \text{(number of perfect matchings in } G(A)\text{)}$.

We now show that the number of nonzero terms in the density formula~\eqref{eq:nu_marginal} equals $\mathrm{perm}(A(R'))^K$, where $A(R')$ is a $\{0,1\}$ matrix derived from $R'$. The connection to perfect matchings is the key to understanding why this count is \#P-hard.

\begin{definition}[Support set and support matrix]
\label{def:support}
For a doubly stochastic matrix $R' \in \mathcal{B}_{2n}$, define the \emph{support matrix} $A(R') \in \{0,1\}^{2n \times 2n}$ by
\[
A(R')_{ab} = \begin{cases} 1 & \text{if } R'_{ab} > 0, \\ 0 & \text{if } R'_{ab} = 0. \end{cases}
\]
In other words, $A(R')$ marks which entries of $R'$ are strictly positive. The \emph{support set} $\mathrm{Supp}(R')$ is the set of all $2n \times 2n$ permutation matrices that ``fit inside'' the positive entries of $R'$:
\[
\mathrm{Supp}(R') = \{Q \in S_{2n} : \text{wherever } Q \text{ has a 1, } R' \text{ has a positive entry}\}.
\]
Formally, $Q \in \mathrm{Supp}(R')$ if and only if $Q_{ab} = 1$ implies $R'_{ab} > 0$ for every row-column pair $(a,b)$. Equivalently, $Q \in \mathrm{Supp}(R')$ if and only if $Q_{ab} = 1$ implies $A(R')_{ab} = 1$.
\end{definition}

The intuition is simple: a permutation matrix $Q$ places exactly one 1 in each row and each column. It ``fits inside'' $R'$ if none of its 1's land on a zero entry of $R'$. The support set is exactly the set of permutation matrices that are compatible with the zero pattern of $R'$.

\begin{proposition}[Support size equals permanent --- connecting permutation counting to graph matchings]
\label{prop:nu_permanent}
The number of permutation matrices in the support set equals the permanent of the support matrix:
\begin{equation}
\label{eq:supp_perm}
|\mathrm{Supp}(R')| = \mathrm{perm}(A(R')).
\end{equation}
Furthermore, the number of nonzero terms in the density formula~\eqref{eq:nu_marginal} is exactly $\mathrm{perm}(A(R'))^K$.
\end{proposition}

This is not a new result but rather a direct application of the standard connection between permanents and perfect matchings (see, e.g., Schrijver~\cite{schrijver2003}, Chapter~8). The contribution here is recognizing that this connection arises naturally in the density formula for Birkhoff-encoded data, linking the aggregator's inference problem to a classical \#P-hard computation.

\begin{proof}
The proof establishes two things: first, that the constraint~\eqref{eq:constraint} forces every permutation in a contributing tuple to lie in $\mathrm{Supp}(R')$; second, that counting the elements of $\mathrm{Supp}(R')$ is identical to computing the permanent of $A(R')$.

\medskip\noindent\textbf{Which tuples contribute nonzero terms?}
Consider a specific tuple $(\sigma_1, \ldots, \sigma_K)$ and ask when the constraint~\eqref{eq:constraint} can be satisfied. The constraint says $\sum_{i=1}^K \alpha_i P_{\sigma_i} = (1-\alpha^*) R'$, which must hold at every matrix entry $(a,b)$ simultaneously:
\[
\sum_{i=1}^K \alpha_i (P_{\sigma_i})_{ab} = (1-\alpha^*) R'_{ab}, \qquad \forall\; (a,b) \in \{1, \ldots, 2n\}^2.
\]
Now, each $P_{\sigma_i}$ is a permutation matrix, so its entries are 0 or 1. Each $\alpha_i$ is strictly positive. Therefore the left side at entry $(a,b)$ is
\[
\text{LHS}_{ab} = \sum_{\substack{i=1 \\ (P_{\sigma_i})_{ab} = 1}}^K \alpha_i,
\]
which is a sum of strictly positive numbers. This sum is zero if and only if no permutation in the tuple has a 1 at position $(a,b)$ (the index set is empty), and is strictly positive otherwise. The right side is $(1-\alpha^*) R'_{ab}$, which is positive when $R'_{ab} > 0$ and zero when $R'_{ab} = 0$.

Matching the two sides: if $R'_{ab} = 0$, then the right side is zero, so the left side must be zero, which means no permutation $P_{\sigma_i}$ can have a 1 at $(a,b)$. If $R'_{ab} > 0$, then the right side is positive, so the left side must also be positive, which means at least one permutation must have a 1 at $(a,b)$.

The first condition is the binding one: whenever $R'_{ab} = 0$, every single permutation in the tuple must have a 0 at $(a,b)$. Since a permutation matrix $P_{\sigma_i}$ has $(P_{\sigma_i})_{ab} = 1$ if and only if $\sigma_i(a) = b$ (row $a$'s unique 1 is in column $b$), requiring $(P_{\sigma_i})_{ab} = 0$ means requiring $\sigma_i(a) \neq b$. In terms of the support matrix: $P_{\sigma_i}$ must place its 1's only at positions where $A(R')_{ab} = 1$. This is precisely the condition $P_{\sigma_i} \in \mathrm{Supp}(R')$.

A tuple contributes a nonzero term to~\eqref{eq:nu_marginal} only if \emph{every} permutation in the tuple lies in $\mathrm{Supp}(R')$: $P_{\sigma_i} \in \mathrm{Supp}(R')$ for all $i = 1, \ldots, K$.

\medskip\noindent\textbf{Counting the support set via perfect matchings.}
How many permutation matrices belong to $\mathrm{Supp}(R')$? A permutation matrix $Q \in S_{2n}$ corresponds to a bijection $\sigma: \{1, \ldots, 2n\} \to \{1, \ldots, 2n\}$, where $Q_{a,\sigma(a)} = 1$ for each row $a$ and all other entries are 0. The condition $Q \in \mathrm{Supp}(R')$ requires $A(R')_{a,\sigma(a)} = 1$ for every $a$ --- that is, the bijection $\sigma$ must map each row $a$ to a column $\sigma(a)$ where $A$ has a 1.

This is exactly a \emph{perfect matching} in the bipartite graph $G(A(R'))$: left vertices are rows $\{1, \ldots, 2n\}$, right vertices are columns $\{1, \ldots, 2n\}$, and there is an edge from row $a$ to column $b$ whenever $A(R')_{ab} = 1$ (i.e., $R'_{ab} > 0$). A perfect matching assigns each row to a distinct column via an edge, which is exactly what a bijection $\sigma$ with $A_{a,\sigma(a)} = 1$ does.

The number of such bijections is
\[
|\mathrm{Supp}(R')| = \sum_{\sigma \in S_{2n}} \prod_{a=1}^{2n} A(R')_{a,\sigma(a)}.
\]
Each product $\prod_a A_{a,\sigma(a)}$ is 1 if $\sigma$ is a valid matching (every factor is 1) and 0 otherwise. The sum counts all valid matchings. Comparing with Definition~\ref{def:permanent}, this is precisely $\mathrm{perm}(A(R'))$.

\medskip\noindent\textbf{Counting nonzero terms in the density sum.}
Each of the $K$ slots in the tuple $(\sigma_1, \ldots, \sigma_K)$ must independently satisfy $P_{\sigma_i} \in \mathrm{Supp}(R')$. There are $|\mathrm{Supp}(R')| = \mathrm{perm}(A(R'))$ valid choices for each slot, and the slots are independent, so the total number of tuples that contribute nonzero terms to~\eqref{eq:nu_marginal} is $\mathrm{perm}(A(R'))^K$.

By Theorem~\ref{thm:valiant}, computing $\mathrm{perm}(A)$ for $\{0,1\}$ matrices is \#P-complete, so even determining the number of nonzero terms in the density formula is \#P-hard.
\end{proof}

\begin{example}[Concrete illustration for $n = 2$]
\label{ex:permanent}
Let $n = 2$ (so matrices are $4 \times 4$) and suppose $R'$ has positive entries only in two diagonal blocks:
\[
R' = \begin{pmatrix} 0.4 & 0.6 & 0 & 0 \\ 0.6 & 0.4 & 0 & 0 \\ 0 & 0 & 0.3 & 0.7 \\ 0 & 0 & 0.7 & 0.3 \end{pmatrix}, \qquad
A(R') = \begin{pmatrix} 1 & 1 & 0 & 0 \\ 1 & 1 & 0 & 0 \\ 0 & 0 & 1 & 1 \\ 0 & 0 & 1 & 1 \end{pmatrix}.
\]
The bipartite graph $G(A)$ has edges $\{1{\leftrightarrow}1, 1{\leftrightarrow}2, 2{\leftrightarrow}1, 2{\leftrightarrow}2, 3{\leftrightarrow}3, 3{\leftrightarrow}4, 4{\leftrightarrow}3, 4{\leftrightarrow}4\}$. The perfect matchings are the four bijections $\{(1{\to}1, 2{\to}2, 3{\to}3, 4{\to}4)$, $(1{\to}2, 2{\to}1, 3{\to}3, 4{\to}4)$, $(1{\to}1, 2{\to}2, 3{\to}4, 4{\to}3)$, $(1{\to}2, 2{\to}1, 3{\to}4, 4{\to}3)\}$, giving $\mathrm{perm}(A) = 4$. With $K = 3$ decoys, the density sum has $4^3 = 64$ nonzero terms out of $(4!)^3 = 13{,}824$ total.

For a generic interior point of $\mathcal{B}_4$ (where every entry of $R'$ is positive), $A = \mathbf{J}_4$ (all-ones matrix), $\mathrm{perm}(\mathbf{J}_4) = 4! = 24$ (every permutation is a valid matching), and all $24^3 = 13{,}824$ terms are nonzero.
\end{example}

\subsubsection{Connection to the Mixed Discriminant}
\label{sec:mixed_disc}

Proposition~\ref{prop:nu_permanent} shows that even \emph{counting} the nonzero terms in $\nu(R')$ is \#P-hard. We now show that \emph{evaluating} each nonzero term also involves a \#P-hard quantity: the mixed discriminant.

For each nonzero tuple $(\sigma_1, \ldots, \sigma_K) \in \mathrm{Supp}(R')^K$, the per-tuple integral is
\begin{equation}
\label{eq:per_tuple_integral}
I(\sigma_1, \ldots, \sigma_K;\; R') = \int_{\Delta_K} \mathbf{1}\!\left[\sum_{i=1}^K \alpha_i P_{\sigma_i} = (1-\alpha^*) R'\right] \cdot g(\alpha) \; d\alpha.
\end{equation}

\medskip\noindent\textit{The constraint as a linear system.}
The indicator constrains $\alpha_1, \ldots, \alpha_K$ to satisfy $(2n)^2$ linear equations (one per matrix entry)
\begin{equation}
\label{eq:linear_constraint}
\sum_{i: (P_{\sigma_i})_{ab} = 1} \alpha_i = (1-\alpha^*) R'_{ab}, \qquad \forall\, (a,b) \in [2n]^2.
\end{equation}
Since $R'$ is doubly stochastic, the $2n$ row-sum equations and $2n$ column-sum equations are automatically satisfied (they all reduce to $\sum_i \alpha_i = 1 - \alpha^*$). The effective number of independent constraints is $(2n-1)^2$, and the feasible set is a convex polytope
\begin{equation}
\label{eq:polytope}
\mathcal{P}(\sigma_1, \ldots, \sigma_K; R') = \{\alpha \in \Delta_K : \sum_{i: (P_{\sigma_i})_{ab} = 1} \alpha_i = (1-\alpha^*) R'_{ab}\;\;\forall (a,b)\}.
\end{equation}

\medskip\noindent\textit{The integral as a polytope volume.}
With Dirichlet$(1, \ldots, 1)$ coefficients (uniform on $\Delta_K$), the density $g$ is constant on the simplex, and the integral reduces to the volume of the polytope~\eqref{eq:polytope}
\begin{equation}
\label{eq:volume}
I(\sigma_1, \ldots, \sigma_K; R') = C_K \cdot \mathrm{vol}_{d_0}(\mathcal{P}),
\end{equation}
where $C_K$ is the Dirichlet normalizing constant and $d_0 = \dim(\mathcal{P}) = K - 1 - (2n-1)^2$ is the dimension of the feasible set (when it is nonempty and the constraints are non-degenerate).

\medskip\noindent\textit{The total density as a sum of volumes.}
Combining~\eqref{eq:nu_marginal} and~\eqref{eq:volume},
\begin{equation}
\label{eq:nu_as_sum_of_volumes}
\nu(R') = \frac{C_K}{((2n)!)^K} \sum_{(\sigma_1, \ldots, \sigma_K) \in \mathrm{Supp}(R')^K} \mathrm{vol}_{d_0}(\mathcal{P}(\sigma_1, \ldots, \sigma_K; R')).
\end{equation}
This is a sum of $\mathrm{perm}(A(R'))^K$ polytope volumes.

\begin{definition}[Mixed discriminant]
\label{def:mixed_disc}
Given $d$ positive semidefinite matrices $X_1, \ldots, X_d \in \mathbb{R}^{d \times d}$, the mixed discriminant is
\begin{equation}
\label{eq:mixed_disc}
D(X_1, \ldots, X_d) = \frac{\partial^d}{\partial \lambda_1 \cdots \partial \lambda_d} \det\!\left(\sum_{j=1}^d \lambda_j X_j\right)\bigg|_{\lambda = 0}.
\end{equation}
When each $X_j$ is a diagonal matrix with a single nonzero entry, $D(X_1, \ldots, X_d)$ reduces to the permanent.
\end{definition}

\begin{theorem}[Barvinok, 1997~\cite{barvinok1996}]
\label{thm:barvinok}
Computing the mixed discriminant of $d$ positive semidefinite $d \times d$ matrices is \#P-hard.
\end{theorem}

The connection to our problem is as follows. Each polytope $\mathcal{P}(\sigma_1, \ldots, \sigma_K; R')$ is defined by the intersection of the simplex with a linear subspace determined by the $K$ permutation matrices. Its volume can be expressed as a mixed volume of zonotopes generated by the rows of the permutation matrices. By the Minkowski--Bernstein--Khovanskii theorem, mixed volumes of zonotopes are mixed discriminants of matrices constructed from the generators. Since Barvinok proved that mixed discriminants are \#P-hard, each polytope volume in~\eqref{eq:nu_as_sum_of_volumes} is \#P-hard to compute.

\begin{remark}[Why Barvinok's quasi-polynomial approximation does not apply]
\label{rem:barvinok_approx}
In subsequent work, Barvinok~\cite{barvinok2019} showed that the mixed discriminant of $n$ positive semidefinite $n \times n$ matrices can be approximated within relative error $\varepsilon > 0$ in quasi-polynomial $n^{O(\ln n - \ln \varepsilon)}$ time, \emph{provided} the operator norm distance of each matrix from the identity satisfies $\|A_i - I\|_{\mathrm{op}} \leq \gamma_0$ for an absolute constant $\gamma_0 < 1$. This raises the question of whether our mixed discriminants fall within this approximable regime.

They do not. The matrices in our mixed discriminant are constructed from the decoy permutation matrices $P_{\sigma_1}, \ldots, P_{\sigma_K}$. Each $P_{\sigma_i}$ is an orthogonal matrix (hence $\|P_{\sigma_i}\|_{\mathrm{op}} = 1$), but its distance from the identity is large. For any permutation $\sigma \neq \mathrm{id}$ that contains a transposition swapping positions $i$ and $j$, the matrix $P_\sigma - I$ restricted to the $\{i, j\}$ subspace is $\left(\begin{smallmatrix} -1 & 1 \\ 1 & -1 \end{smallmatrix}\right)$, which has eigenvalues $0$ and $-2$. Therefore $\|P_\sigma - I\|_{\mathrm{op}} = 2$ for any non-identity permutation. Since the decoy permutations are drawn uniformly from $S_{2n}$, the probability that $P_{\sigma_i} = I$ is $1/(2n)!$, which is negligible. With overwhelming probability, every decoy permutation satisfies $\|P_{\sigma_i} - I\|_{\mathrm{op}} = 2$, which exceeds the threshold $\gamma_0 < 1$ by a factor of at least 2.

The quasi-polynomial algorithm requires the matrices to be \emph{small perturbations of the identity}; our permutation matrices are maximal-distance orthogonal matrices that look nothing like the identity. The \#P-hardness barrier for our specific mixed discriminants therefore remains intact, and Barvinok's approximation result does not provide an attack.
\end{remark}

\subsubsection{The Formal Reduction}

\subsubsection{The Formal Hardness Statement}

We now state precisely what the preceding analysis proves, and what it does not prove. The distinction is important and reflects a genuine gap that we discuss openly.

\begin{theorem}[Hardness of the sub-problems in likelihood evaluation]
\label{thm:likelihood_hardness}
Let $R' \in \mathcal{B}_{2n}$ and let $\nu(R')$ be the density defined in~\eqref{eq:nu_marginal}. The following sub-problems, each of which arises in computing $\nu(R')$ via the decomposition~\eqref{eq:nu_as_sum_of_volumes}, are individually \#P-hard:
\begin{enumerate}
\item[\emph{(i)}] Computing the number of nonzero terms in the sum: $\mathrm{perm}(A(R'))^K$ (by Theorem~\ref{thm:valiant}).
\item[\emph{(ii)}] Computing any single nonzero term: $\mathrm{vol}(\mathcal{P}(\sigma_1, \ldots, \sigma_K; R'))$, which reduces to a mixed discriminant (by Theorem~\ref{thm:barvinok}).
\end{enumerate}
Moreover, all terms are non-negative (they are volumes of convex bodies), so there is no cancellation.
\end{theorem}

\begin{proof}
Statement~(i) follows directly from Proposition~\ref{prop:nu_permanent} and Theorem~\ref{thm:valiant}: the number of nonzero terms is $\mathrm{perm}(A(R'))^K$, and computing $\mathrm{perm}(A)$ for $\{0,1\}$ matrices is \#P-complete. Statement~(ii) follows from the connection to mixed discriminants established in Section~\ref{sec:mixed_disc} and Theorem~\ref{thm:barvinok}. Non-negativity follows from the fact that each term is the volume of a convex polytope.
\end{proof}

\begin{remark}[What this does and does not prove --- an honest assessment]
\label{rem:composition_gap}
Theorem~\ref{thm:likelihood_hardness} proves that the \emph{sub-problems} arising in the density computation are individually \#P-hard. It does \emph{not} prove that computing the density $\nu(R')$ \emph{itself} is \#P-hard in the formal complexity-theoretic sense (i.e., that there exists a Turing reduction from a \#P-complete problem to the function $R' \mapsto \nu(R')$).

The gap is a matter of composition. The density $\nu(R')$ is a sum of $\mathrm{perm}(A)^K$ terms, each a \#P-hard volume. But the \#P-hardness of the \emph{sum} does not follow automatically from the \#P-hardness of counting the terms or evaluating each term. Consider the analogy: the determinant $\det(A) = \sum_{\sigma \in S_m} \mathrm{sgn}(\sigma) \prod_i A_{i,\sigma(i)}$ is a sum of $m!$ terms, each trivially computable, yet the sum is in P because Gaussian elimination exploits the alternating sign structure. The permanent is a sum of $m!$ identically structured terms (without signs) yet is \#P-hard. Whether a sum is hard depends on the \emph{global structure} of the sum, not only on the hardness of individual terms.

A formal proof that $R' \mapsto \nu(R')$ is \#P-hard would require constructing a Turing reduction: given an oracle that evaluates $\nu(R')$ at any point $R' \in \mathcal{B}_{2n}$, show how to compute $\mathrm{perm}(A)$ for an arbitrary $\{0,1\}$ matrix $A$. Such a reduction would need to (a)~construct specific points $R'$ on the boundary of $\mathcal{B}_{2n}$ with prescribed support pattern $A(R') = A$, and (b)~extract $\mathrm{perm}(A)$ from the value $\nu(R')$ by controlling or cancelling the polytope volume contributions. Step~(b) is the obstacle: $\nu(R')$ entangles the permanent with the polytope volumes in a way that makes isolation difficult. We leave the construction of such a reduction as an open problem.

What we \emph{can} state with formal rigor:
\begin{enumerate}
\item Any algorithm that evaluates $\nu(R')$ by enumerating terms in the decomposition~\eqref{eq:nu_as_sum_of_volumes} and computing each term individually must solve \#P-hard problems at each step.
\item No polynomial-time algorithm for evaluating $\nu(R')$ is known, and there is strong structural evidence against one: the sum has $\mathrm{perm}(A)^K$ non-negative terms (no cancellation), each individually \#P-hard, with no known algebraic identity that collapses the sum. This stands in contrast to the determinant, where the alternating signs create the cancellation structure that Gaussian elimination exploits.
\item In the interior regime, the permanent cancels from the likelihood \emph{ratio} (Section~\ref{sec:approx_overview}), and the residual barrier is the sum of $((2n)!)^K$ polytope volumes, for which no polynomial-time evaluation or approximation method is known.
\end{enumerate}

We believe the correct conjecture is that $\nu(R')$ is \#P-hard to evaluate, but a formal proof requires either a direct Turing reduction or a new composition theorem for sums of \#P-hard quantities without cancellation.
\end{remark}

\subsubsection{Contrast with Gaussian Noise}

\begin{remark}[Why Gaussian noise would be easy]
\label{rem:hardness_scope}
If the decoy component were additive Gaussian noise ($D_t = \alpha^* M_t + \varepsilon$, $\varepsilon_{ab} \sim N(0, \sigma^2)$ i.i.d.), the likelihood would be
\[
\Pr(D_t \mid M') = \prod_{a,b} \frac{1}{\sqrt{2\pi}\sigma}\exp\!\left(-\frac{((D_t)_{ab} - \alpha^*(M')_{ab})^2}{2\sigma^2}\right) \propto \exp\!\left(-\frac{\|D_t - \alpha^* M'\|_F^2}{2\sigma^2}\right).
\]
This is computable in $O(n^2)$ time. The MAP estimate minimizes $\|D_t/\alpha^* - M'\|_F^2$ over permutation matrices, which is a linear assignment problem solvable by the Hungarian algorithm in $O(n^3)$. The Birkhoff encoding replaces this tractable Gaussian likelihood with a \#P-hard sum over BvN decompositions. This is the specific sense in which the Birkhoff polytope provides computational hardness that other noise distributions do not.
\end{remark}

\subsubsection{Why Approximate Permanent Algorithms Do Not Help}
\label{sec:approx_overview}

Theorem~\ref{thm:likelihood_hardness} establishes that the sub-problems of computing $\nu(R')$ are individually \#P-hard (see Remark~\ref{rem:composition_gap} for the compositional subtlety). A natural question is whether polynomial-time \emph{approximation} algorithms for the permanent --- most notably the Jerrum--Sinclair--Vigoda (JSV) FPRAS~\cite{jsv2004}, which approximates the permanent of any non-negative matrix to within a $(1+\varepsilon)$ factor in polynomial time --- could be used to approximate $\nu(R')$ and thereby enable approximate MAP estimation. Perhaps surprisingly, the answer is no, for a reason that is worth understanding in detail because it reveals the true computational barrier.

Recall from~\eqref{eq:nu_as_sum_of_volumes} that the density decomposes as
\[
\nu(R') = \frac{C_K}{((2n)!)^K} \sum_{(\sigma_1, \ldots, \sigma_K) \in \mathrm{Supp}(R')^K} \mathrm{vol}(\mathcal{P}(\sigma_1, \ldots, \sigma_K; R')).
\]
The permanent enters through the size of the index set: the sum has $|\mathrm{Supp}(R')|^K = \mathrm{perm}(A(R'))^K$ terms. The attacker wants to compare $\nu(R')$ for the true candidate $M_t$ against $\nu(R'')$ for a wrong candidate $M'$. If the permanents $\mathrm{perm}(A(R'))$ and $\mathrm{perm}(A(R''))$ differed between candidates, approximating them via JSV would give the attacker useful information.

\emph{But in the operating regime of the protocol, the permanents are identical for every candidate.} The reason is elementary but important to spell out. The residual for candidate $M'$ is $R' = R_t + \frac{\alpha^*}{1-\alpha^*}(M_t - M')$, which perturbs the true residual $R_t$ by at most $\frac{\alpha^*}{1-\alpha^*}$ per entry. For the protocol parameter $\alpha^* = 1/(4n)$, this perturbation is $\frac{1}{4n-1} \approx \frac{1}{4n}$.

The support matrix $A(R')_{ab} = \mathbf{1}[R'_{ab} > 0]$ marks which entries of $R'$ are strictly positive. A permutation matrix ``fits inside'' $R'$ --- meaning it belongs to $\mathrm{Supp}(R')$ --- if and only if wherever the permutation places a 1, $R'$ has a positive entry. The permanent $\mathrm{perm}(A(R'))$ counts how many permutation matrices fit inside $R'$.

Now consider what happens when every entry of $R'$ is positive. In this case, $A(R') = \mathbf{J}_{2n}$ (the all-ones matrix), and \emph{every} permutation matrix fits inside $R'$, because every entry is positive and there is nowhere a permutation's 1 could land on a zero. The count is therefore
\[
\mathrm{perm}(\mathbf{J}_{2n}) = (2n)!,
\]
which is the total number of permutation matrices --- all of them fit. This is independent of the candidate $M'$. It is analogous to asking ``how many ways can $2n$ non-attacking rooks be placed on a $2n \times 2n$ chessboard where every square is available?'' The answer is $(2n)!$ regardless of which candidate generated the board, because all squares are available.

The condition for every entry of $R'$ to be positive is that the perturbation does not push any entry of $R_t$ to zero. The worst case is at entries where $R_t$ is smallest: we need $(R_t)_{ab} > \frac{\alpha^*}{1-\alpha^*} \approx \frac{1}{4n}$ for all $(a,b)$. The mean entry of $R_t$ is $\mathbb{E}[(R_t)_{ab}] = \frac{1}{2n}$, which is twice the threshold $\frac{1}{4n}$. With $K_t$ decoys, the standard deviation is approximately $\frac{1}{2n\sqrt{K_t}}$, so the threshold is $\sqrt{K_t}/2$ standard deviations below the mean. For $K_t = 20$, this is $\sim 2.2$ standard deviations, and the probability that all $(2n)^2$ entries exceed the threshold is high for moderate $n$ and $K_t$. This is the \emph{interior condition} (Definition~\ref{def:interior_condition}).

When the interior condition holds, the density formula simplifies to
\begin{equation}
\label{eq:nu_interior_main}
\nu(R') = \frac{C_K}{((2n)!)^K} \sum_{(\sigma_1, \ldots, \sigma_K) \in S_{2n}^K} \mathrm{vol}(\mathcal{P}(\sigma_1, \ldots, \sigma_K; R')),
\end{equation}
and the likelihood ratio between any two candidates is
\begin{equation}
\label{eq:lr_interior}
\frac{\nu(R')}{\nu(R'')} = \frac{\sum_{\tau \in S_{2n}^K} \mathrm{vol}(\mathcal{P}(\tau; R'))}{\sum_{\tau \in S_{2n}^K} \mathrm{vol}(\mathcal{P}(\tau; R''))}.
\end{equation}
The permanent has completely cancelled. Both sums range over the same $((2n)!)^K$ tuples. The JSV FPRAS computes $\mathrm{perm}(A(R')) = (2n)!$ for every candidate --- the same number every time --- and provides zero information for distinguishing candidates.

The discrimination between candidates resides entirely in the \emph{polytope volumes} $\mathrm{vol}(\mathcal{P}(\tau; R'))$, which change when $R'$ changes (because the constraint $\sum_i \alpha_i P_{\sigma_i} = (1-\alpha^*) R'$ shifts with $R'$). The density $\nu(R')$ is a sum of $((2n)!)^K$ such volumes, and the attacker would need to approximate this sum to rank candidates. Modern convex body volume algorithms --- including the Lov\'{a}sz--Vempala algorithm~\cite{lovasz2006} and its refinements by Cousins and Vempala~\cite{cousins2018} --- can compute the volume of a \emph{single} convex polytope in polynomial time (roughly $\tilde{O}(d^3)$ oracle calls for a $d$-dimensional body). However, the attacker's problem is not to compute one volume but to \emph{sum} $((2n)!)^K$ volumes. Even for small parameters ($n = 5$, $K = 10$), this sum has $(10!)^{10} \approx 10^{65}$ terms. Computing each volume in polynomial time and summing would take $\mathrm{poly}(K) \times 10^{65}$ operations --- completely infeasible. The bottleneck is the combinatorial explosion in the number of terms, not the cost of evaluating any single term.

A detailed analysis of various approximate attack strategies --- including Monte Carlo estimation of the volume sum (which fails due to exponentially small hit rates), importance sampling via the JSV near-uniform matching sampler (which provides no improvement in the interior regime), MCMC on BvN decompositions, spectral methods, LP relaxation, boson sampling, and the boundary regime where the permanent does vary --- is provided in Appendix~\ref{app:attacks}.

\subsection{Formal Two-Layer Security Statement}

\begin{definition}[Two-layer security]
\label{def:twolayer_security}
A protocol has \emph{two-layer security} against semi-honest adversaries if (i) the server's view is identically distributed for any two inputs with the same aggregate (statistical distance zero --- \emph{information-theoretic server security}); (ii) no polynomial-time aggregator can compute or approximate the posterior distribution~\eqref{eq:agg_posterior} over candidate permutation matrices, because doing so requires solving a \#P-hard problem (\emph{computational aggregator security against likelihood attacks}); and (iii) the server and aggregator do not share their views (\emph{non-collusion}).
\end{definition}

\begin{theorem}[Two-Layer PolyVeil security --- proved components]
\label{thm:twolayer}
Under the non-collusion assumption, Algorithm~\ref{alg:twolayer} achieves condition~(i) by Theorem~\ref{thm:server_it} (perfect simulation, unconditional), and condition~(ii) by Theorem~\ref{thm:likelihood_hardness} (likelihood-based attacks require solving individually \#P-hard sub-problems: the permanent and mixed discriminant), under the interior condition $\alpha^* \leq \min_{ij}(D_t)_{ij}$.
\end{theorem}

\begin{remark}[What remains open]
\label{rem:open}
Full aggregator security (Conjecture~\ref{conj:hardness}) --- ruling out \emph{all} polynomial-time attacks, not just likelihood-based ones --- remains an open problem. A proof would require either (a) a reduction from a \#P-complete or NP-hard problem to the search problem ``recover $M_t$ from $D_t$,'' or (b) an average-case hardness result for the permanent over the specific distribution induced by $\nu$. Both are significant open problems in computational complexity. The average-case hardness of the permanent has been studied by Lipton~\cite{lipton1991} and others, with partial results (e.g., hardness over finite fields) but no complete resolution over the reals.
\end{remark}

\section{Multi-Statistic Extraction from the Birkhoff Encoding}
\label{sec:multistat}

A single masked matrix $D_t = \alpha^* M_t + (1-\alpha^*) R_t$ encodes the \emph{entire} bit vector $\mathbf{b}_t \in \{0,1\}^n$, not merely its sum $s_t$. In the full two-layer protocol (Algorithm~\ref{alg:twolayer}), the aggregator observes $D_t$ and can extract multiple statistics from it using different extraction vectors, all within a single protocol execution. The compressed two-layer protocol (Algorithm~\ref{alg:compressed_twolayer}) does not support multi-statistic extraction, since the aggregator receives only the scalar $f_t$ and the matrix $D_t$ is never transmitted. This section derives the statistics that the full protocol supports and compares the communication cost with additive secret sharing. Figure~\ref{fig:multistat} illustrates the non-interactive multi-statistic extraction pipeline.

\begin{figure}[!htbp]
\centering
\begin{tikzpicture}[
    box/.style={draw, rounded corners, minimum width=2.4cm, minimum height=0.8cm, align=center, font=\small},
    stat/.style={draw, rounded corners, minimum width=3.2cm, minimum height=0.55cm, align=center, font=\footnotesize, fill=green!8},
    arrow/.style={->, >=stealth, thick}
]

\node[box, fill=blue!10] (clients) at (0, 0) {$k$ clients};
\node[font=\scriptsize, align=center] at (0, -0.8) {each holds\\$\mathbf{b}_t \in \{0,1\}^n$};

\node[box, fill=orange!12, minimum height=1.2cm] (agg) at (5, 0) {Aggregator\\stores $D_1, \ldots, D_k$};
\node[font=\scriptsize, red!70, align=center] at (5, -1.1) {clients now offline};

\draw[arrow] (1.3, 0) -- (3.6, 0) node[midway, above, font=\scriptsize] {send $D_t$ once};

\node[stat] (q1) at (10.5, 1.2) {Boolean sum $\sum_t s_t$};
\node[stat] (q2) at (10.5, 0) {Per-bit counts $\sum_t b_{t,j}$};
\node[stat] (q3) at (10.5, -1.2) {Weighted sums $\sum_t \mathbf{c}^T\!\mathbf{b}_t$};

\draw[arrow] (6.4, 0.2) -- (q1.west);
\draw[arrow] (6.4, 0) -- (q2.west);
\draw[arrow] (6.4, -0.2) -- (q3.west);

\node[font=\scriptsize, align=center, gray] at (8.2, -2.3) {different extraction vectors\\applied to same stored $D_t$};

\end{tikzpicture}
\caption{Non-interactive multi-statistic extraction. In the full two-layer protocol (Algorithm~\ref{alg:twolayer}), each client transmits the masked matrix $D_t$ to the aggregator once. After clients go offline, the aggregator applies different extraction vectors to the stored matrices to compute multiple aggregate statistics, each recovered exactly by the server via noise cancellation. No additional client communication is required for new queries. Additive secret sharing requires a new round of client participation for each statistic.}
\label{fig:multistat}
\end{figure}
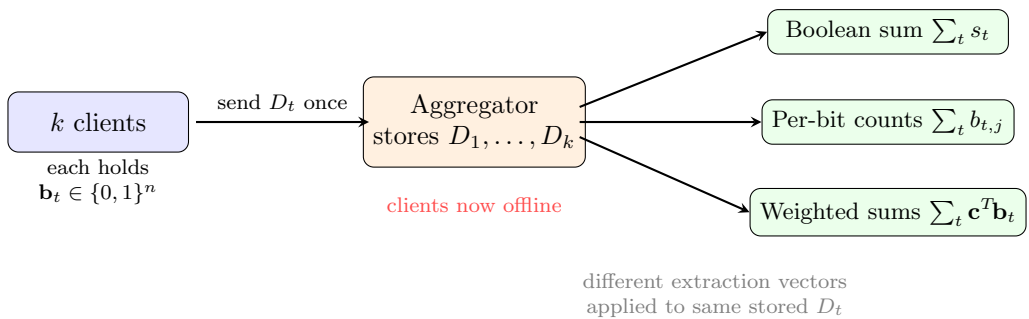

\subsection{Per-Bit Marginal Counts}

\begin{theorem}[Per-bit extraction]
\label{thm:perbit}
For each bit position $j \in [n]$, the full two-layer protocol (Algorithm~\ref{alg:twolayer}) can compute $\sum_{t=1}^k b_{t,j}$ exactly (the number of clients with bit $j$ equal to 1), using the same matrices $D_t$ already sent by each client.
\end{theorem}

\begin{proof}
The encoding $M_t = \mathrm{blockdiag}(\Pi(b_{t,1}), \ldots, \Pi(b_{t,n}))$ has the property (from Definition~\ref{def:encoding}) that $(M_t)_{2j-1, 2j} = b_{t,j}$ for each $j \in [n]$. Define the extraction vectors $\mathbf{w}_j, \mathbf{y}_j \in \mathbb{R}^{2n}$ by
\begin{align}
(\mathbf{w}_j)_a &= \begin{cases} 1 & \text{if } a = 2j-1 \\ 0 & \text{otherwise} \end{cases}, \qquad
(\mathbf{y}_j)_a = \begin{cases} 1 & \text{if } a = 2j \\ 0 & \text{otherwise} \end{cases}. \label{eq:perbit_vectors}
\end{align}
These are simply the standard basis vectors $\mathbf{e}_{2j-1}$ and $\mathbf{e}_{2j}$. Then
\begin{align}
\mathbf{w}_j^T M_t\, \mathbf{y}_j &= \sum_{a,b} (\mathbf{w}_j)_a\, (M_t)_{ab}\, (\mathbf{y}_j)_b = (M_t)_{2j-1,\, 2j} = b_{t,j}. \label{eq:perbit_extract}
\end{align}
By linearity of the bilinear form $A \mapsto \mathbf{w}_j^T A\, \mathbf{y}_j$, applied to $D_t = \alpha^* M_t + (1-\alpha^*) R_t$,
\begin{align}
\mathbf{w}_j^T D_t\, \mathbf{y}_j &= \alpha^* b_{t,j} + (1-\alpha^*)\, \mathbf{w}_j^T R_t\, \mathbf{y}_j = \alpha^* b_{t,j} + \eta_{t,j}, \label{eq:ftj}
\end{align}
where $\eta_{t,j} = (1-\alpha^*)(R_t)_{2j-1, 2j}$ is the noise contribution from the decoy component at position $(2j-1, 2j)$.

Define $f_{t,j} = \mathbf{w}_j^T D_t\, \mathbf{y}_j$ and $F_j = \sum_{t=1}^k f_{t,j}$. Define $H_j = \sum_{t=1}^k \eta_{t,j}$. The aggregator computes $F_j$ from the matrices $D_t$ it already holds. The noise aggregator computes $H_j$ from the noise values it already holds (since $\eta_{t,j}$ is determined by the decoy permutations and coefficients that the noise aggregator receives). The server recovers
\begin{align}
\frac{F_j - H_j}{\alpha^*} &= \frac{\alpha^* \sum_t b_{t,j} + \sum_t \eta_{t,j} - \sum_t \eta_{t,j}}{\alpha^*} = \sum_{t=1}^k b_{t,j}. \label{eq:perbit_recover}
\end{align}
No additional communication is required. The aggregator already has $D_t$ and the noise aggregator already has the decoy parameters. The server applies $n$ different extraction vector pairs to the same data, obtaining all $n$ per-bit marginal counts from a single protocol execution.
\end{proof}

\begin{remark}[Comparison with additive secret sharing]
\label{rem:ss_perbit}
To compute all $n$ per-bit counts via additive secret sharing, each client must secret-share $n$ separate values $(b_{t,1}, \ldots, b_{t,n})$. With two-server additive secret sharing, each client sends $n$ shares to server A and $n$ shares to server B, for a total communication of $O(kn \log k)$ bits, where each share requires $\lceil \log_2(k+1) \rceil$ bits (since per-bit aggregates are at most $k$). In contrast, PolyVeil sends one $2n \times 2n$ matrix per client, totaling $O(kn^2)$ entries of 64 bits each. For a single statistic (the Boolean sum), additive secret sharing uses $O(k \log(kn))$ bits and PolyVeil uses $O(kn^2)$ bits --- secret sharing is far cheaper. For all $n$ per-bit marginals, additive secret sharing uses $O(kn \log k)$ bits and PolyVeil uses $O(kn^2)$ bits --- secret sharing is still cheaper by a factor of $n/\log k$. The advantage of the Birkhoff encoding is not communication cost but rather the ability to compute additional statistics from the same matrices without further client interaction.
\end{remark}

\subsection{Arbitrary Weighted Sums}

\begin{theorem}[Weighted extraction]
\label{thm:weighted}
For any weight vector $\mathbf{c} = (c_1, \ldots, c_n) \in \mathbb{R}^n$, the full two-layer protocol (Algorithm~\ref{alg:twolayer}) can compute $\sum_{t=1}^k \sum_{j=1}^n c_j b_{t,j}$ exactly from the same matrices $D_t$.
\end{theorem}

\begin{proof}
Define $\mathbf{w}_c \in \mathbb{R}^{2n}$ by $(\mathbf{w}_c)_{2j-1} = c_j$ and $(\mathbf{w}_c)_{2j} = 0$ for all $j \in [n]$, and let $\mathbf{y}$ be the standard extraction vector from Definition~\ref{def:selectors} with $y_{2j-1} = 0$, $y_{2j} = 1$. Then
\begin{align}
\mathbf{w}_c^T M_t\, \mathbf{y} &= \sum_{j=1}^n c_j (M_t)_{2j-1, 2j} \cdot 1 = \sum_{j=1}^n c_j b_{t,j}. \label{eq:weighted_extract}
\end{align}
Applying the bilinear form to $D_t$,
\begin{align}
\mathbf{w}_c^T D_t\, \mathbf{y} &= \alpha^* \sum_{j=1}^n c_j b_{t,j} + (1-\alpha^*)\, \mathbf{w}_c^T R_t\, \mathbf{y}. \label{eq:ftc}
\end{align}
Define $f_{t,c} = \mathbf{w}_c^T D_t\, \mathbf{y}$, $\eta_{t,c} = (1-\alpha^*)\, \mathbf{w}_c^T R_t\, \mathbf{y}$, $F_c = \sum_t f_{t,c}$, $H_c = \sum_t \eta_{t,c}$. The server recovers
\begin{align}
\frac{F_c - H_c}{\alpha^*} &= \sum_{t=1}^k \sum_{j=1}^n c_j b_{t,j}. \label{eq:weighted_recover}
\end{align}
Multiple weight vectors $\mathbf{c}_1, \ldots, \mathbf{c}_m$ can be applied to the same $D_t$ matrices, computing $m$ different weighted sums from a single protocol execution.
\end{proof}

This means a single execution of the full two-layer protocol (Algorithm~\ref{alg:twolayer}) can simultaneously compute the total count ($\mathbf{c} = \mathbf{1}$), any weighted count ($\mathbf{c}$ arbitrary), per-bit marginals ($\mathbf{c} = \mathbf{e}_j$ for each $j$), and any other linear functional of the client's bit vector. With additive secret sharing, each new linear functional requires the clients to compute and share a new value, incurring additional communication per statistic.

\subsection{Comparison with Additive Secret Sharing}
\label{sec:ss_compare}

The Birkhoff encoding's advantage lies in multi-statistic extraction. The aggregator in the full two-layer protocol (Algorithm~\ref{alg:twolayer}) sees $D_t$ and can compute $n$ per-bit marginals, arbitrary weighted sums from the same data, without further client interaction. In additive secret sharing, each statistic requires the clients to share a new value, which is impossible after the clients have gone offline. In applications where the set of statistics to be computed is not fully known at protocol execution time (e.g., exploratory data analysis, where the analyst decides which cross-tabulations to examine after receiving the data), the Birkhoff encoding provides a non-interactive capability that additive secret sharing cannot match.

The trade-off is therefore not privacy for accuracy but rather privacy strength for post-hoc analytical flexibility. For a fixed, predetermined set of statistics, additive secret sharing is preferred. For settings where the analyst needs to extract multiple or unanticipated statistics from a single data collection round, the Birkhoff encoding provides a structured alternative.

\section{Provable Aggregator Privacy via Anti-Concentration}
\label{sec:resolve}

We prove a finite-sample differential privacy guarantee for the aggregator with explicit constants. The result uses $(\varepsilon, \delta)$-DP rather than pure $(\varepsilon, 0)$-DP, which avoids the need for pointwise density bounds on the tails of $\nu$ where the Gaussian approximation is unreliable.

\subsection{Key Observation: $\ell_\infty$ Norm Does Not Grow with $n$}

When client $t$'s bitstream changes from $\mathbf{b}_t$ to $\mathbf{b}_t'$ (possibly in all $n$ bits), the permutation matrix changes from $M_t = M(\mathbf{b}_t)$ to $M_t' = M(\mathbf{b}_t')$. Since $M_t$ and $M_t'$ are block-diagonal with $n$ disjoint $2 \times 2$ blocks, the difference $M_t - M_t'$ has nonzero entries only in the blocks where $b_{t,j} \neq b'_{t,j}$. In each such block, the entries have magnitude at most 1. The blocks are \emph{disjoint}: block $j$ occupies rows $\{2j-1, 2j\}$ and columns $\{2j-1, 2j\}$, and different blocks share no rows or columns. Therefore
\begin{equation}
\label{eq:linf_key}
\|M_t - M_t'\|_\infty \;=\; \max_{a,b} |(M_t - M_t')_{ab}| \;=\; 1,
\end{equation}
regardless of how many bits change (from 1 to $n$). This is because the $\ell_\infty$ norm takes a maximum, not a sum, over entries.

The aggregator's view shifts from $D_t = \alpha^* M_t + (1-\alpha^*) R_t$ to $D_t' = \alpha^* M_t' + (1-\alpha^*) R_t$ (same $R_t$). In the residual space, the shift is
\begin{equation}
\label{eq:delta_general}
\delta = \frac{D_t - D_t'}{1-\alpha^*} = \frac{\alpha^*}{1-\alpha^*}(M_t - M_t'), \qquad \|\delta\|_\infty = \frac{\alpha^*}{1-\alpha^*}.
\end{equation}
This is the \emph{same} for changing 1 bit or all $n$ bits. Consequently, a log-Lipschitz bound on $\nu$ in $\|\cdot\|_\infty$ gives the same DP parameter for the full $n$-bit sensitivity as for a single-bit change, with \emph{no composition needed}.

\subsection{$(\varepsilon, \delta)$-DP via a High-Probability Region}

\begin{definition}[$(\varepsilon, \delta)$-DP]
A mechanism $\mathcal{M}$ is $(\varepsilon, \delta)$-differentially private if for all neighboring inputs $M, M'$ and all measurable sets $U$
\[
\Pr[\mathcal{M}(M) \in U] \;\leq\; e^\varepsilon \Pr[\mathcal{M}(M') \in U] + \delta.
\]
\end{definition}

\begin{lemma}[DP from truncated density ratio]
\label{lem:truncated_dp}
Let $p(\cdot \mid M)$ and $p(\cdot \mid M')$ be two densities. Suppose there exists a measurable set $G$ (the ``good set'') such that (i) for all $D \in G$: $\left|\log \frac{p(D \mid M)}{p(D \mid M')}\right| \leq \varepsilon$, and (ii) $\Pr[D \notin G \mid M] \leq \delta$ and $\Pr[D \notin G \mid M'] \leq \delta$.
Then the mechanism is $(\varepsilon, \delta)$-DP.
\end{lemma}

\begin{proof}
For any measurable $U$
\begin{align*}
\Pr[D \in U \mid M] &= \Pr[D \in U \cap G \mid M] + \Pr[D \in U \cap G^c \mid M] \\
&\leq \Pr[D \in U \cap G \mid M] + \Pr[D \in G^c \mid M] \\
&\leq \Pr[D \in U \cap G \mid M] + \delta.
\end{align*}
For $D \in G$, condition (i) gives $p(D \mid M) \leq e^\varepsilon p(D \mid M')$, so
\begin{align*}
\Pr[D \in U \cap G \mid M] &= \int_{U \cap G} p(D \mid M)\, dD \leq e^\varepsilon \int_{U \cap G} p(D \mid M')\, dD \leq e^\varepsilon \Pr[D \in U \mid M'].
\end{align*}
Combining, $\Pr[D \in U \mid M] \leq e^\varepsilon \Pr[D \in U \mid M'] + \delta$.
\end{proof}

\subsection{Concentration of $R_t$: Finite-Sample Bound}

We bound the probability that $R_t$ deviates from its mean $\mu\mathbf{J} = (1/(2n))\mathbf{J}$ using Hoeffding's inequality, which requires no asymptotic approximation.

\begin{lemma}[Concentration of entries of $R_t$]
\label{lem:hoeffding}
For $R_t = \frac{1}{K_t}\sum_{i=1}^{K_t}(P_{t,i})_{ab}$ with uniform weights $\alpha_{t,i} = (1-\alpha^*)/K_t$ and $P_{t,i} \sim S_{2n}$ independent, each entry satisfies for any $r > 0$
\begin{equation}
\label{eq:hoeffding_entry}
\Pr\!\left[\left|(R_t)_{ab} - \frac{1}{2n}\right| > r\right] \;\leq\; 2\exp\!\left(-2K_t r^2\right).
\end{equation}
\end{lemma}

\begin{proof}
Each $(P_{t,i})_{ab} \in \{0, 1\}$ is a bounded random variable with $\mathbb{E}[(P_{t,i})_{ab}] = 1/(2n)$ (since $\Pr[\sigma_i(a) = b] = (2n-1)!/(2n)! = 1/(2n)$). The random variables $(P_{t,1})_{ab}, \ldots, (P_{t,K_t})_{ab}$ are independent (the permutations are drawn independently). The entry $(R_t)_{ab} = \frac{1}{K_t}\sum_{i=1}^{K_t}(P_{t,i})_{ab}$ is the average of $K_t$ independent $[0,1]$-bounded random variables. By Hoeffding's inequality (Hoeffding, 1963), for any $r > 0$
\[
\Pr\!\left[\left|\frac{1}{K_t}\sum_{i=1}^{K_t}(P_{t,i})_{ab} - \frac{1}{2n}\right| > r\right] \leq 2\exp\!\left(-\frac{2K_t r^2}{(1-0)^2}\right) = 2\exp(-2K_t r^2).
\]
(We use the form of Hoeffding's inequality for bounded random variables $X_i \in [a_i, b_i]$ with the bound $2\exp(-2K_t^2 r^2 / \sum_i (b_i - a_i)^2)$. Here $a_i = 0$, $b_i = 1$, so $\sum (b_i - a_i)^2 = K_t$, and $K_t^2/K_t = K_t$.)
\end{proof}

\begin{proposition}[High-probability region]
\label{prop:good_set}
Define the set
\[
\mathcal{G}_r \;=\; \left\{R \in \mathcal{B}_{2n} : \max_{a,b}\left|R_{ab} - \frac{1}{2n}\right| \leq r\right\}.
\]
For $r = \sqrt{\frac{\ln(8n^2/\delta)}{2K_t}}$, we have $\Pr[R_t \notin \mathcal{G}_r] \leq \delta/2$.
\end{proposition}

\begin{proof}
By a union bound over all $(2n)^2 = 4n^2$ entries
\begin{align*}
\Pr\!\left[\max_{a,b}\left|(R_t)_{ab} - \frac{1}{2n}\right| > r\right]
&\leq \sum_{a=1}^{2n}\sum_{b=1}^{2n} \Pr\!\left[\left|(R_t)_{ab} - \frac{1}{2n}\right| > r\right] \\
&\leq 4n^2 \cdot 2\exp(-2K_t r^2) \\
&= 8n^2 \exp(-2K_t r^2).
\end{align*}
Setting $8n^2 \exp(-2K_t r^2) \leq \delta/2$ and solving
\[
2K_t r^2 \geq \ln\!\left(\frac{16n^2}{\delta}\right), \qquad r \geq \sqrt{\frac{\ln(16n^2/\delta)}{2K_t}}.
\]
Taking $r = \sqrt{\ln(16n^2/\delta)/(2K_t)}$ gives $\Pr[R_t \notin \mathcal{G}_r] \leq \delta/2$. (We use $16n^2/\delta$ instead of $8n^2/(\delta/2) = 16n^2/\delta$ to account for the $\delta/2$ target.)
\end{proof}

\subsection{Restricted Log-Lipschitz Constant on $\mathcal{G}_r$}

On $\mathcal{G}_r$, entries of $R$ satisfy $R_{ab} \in [1/(2n) - r,\; 1/(2n) + r]$. This lets us tighten the log-Lipschitz bound dramatically.

\begin{lemma}[Log-Lipschitz constant of $\nu_G$ restricted to $\mathcal{G}_r$]
\label{lem:L_restricted}
For $R, R' \in \mathcal{G}_r$, the Gaussian density $\nu_G(R) = C_G \exp(-\|R - \frac{1}{2n}\mathbf{J}\|_F^2/(2\sigma_K^2))$ satisfies
\begin{equation}
\label{eq:L_restricted}
\left|\log\nu_G(R) - \log\nu_G(R')\right| \;\leq\; \frac{8n^2 r}{\sigma_K^2}\,\|R - R'\|_\infty =: L_r\,\|R - R'\|_\infty.
\end{equation}
\end{lemma}

\begin{proof}
From the expansion $\log\nu_G(R) - \log\nu_G(R') = -\frac{1}{2\sigma_K^2}\sum_{a,b}(R_{ab}-R'_{ab})(R_{ab}+R'_{ab}-\frac{1}{n})$
\begin{align*}
\log\nu_G(R) - \log\nu_G(R') &= -\frac{1}{2\sigma_K^2}\sum_{a,b}(R_{ab} - R'_{ab})\!\left(R_{ab} + R'_{ab} - \frac{1}{n}\right).
\end{align*}
On $\mathcal{G}_r$: $R_{ab} \in [1/(2n) - r,\; 1/(2n) + r]$ and $R'_{ab} \in [1/(2n) - r,\; 1/(2n) + r]$. Therefore
\begin{align}
R_{ab} + R'_{ab} &\in \left[\frac{1}{n} - 2r,\;\; \frac{1}{n} + 2r\right], \notag\\
R_{ab} + R'_{ab} - \frac{1}{n} &\in [-2r,\;\; 2r], \notag\\
\left|R_{ab} + R'_{ab} - \frac{1}{n}\right| &\leq 2r. \label{eq:sum_tight}
\end{align}
(Compare with the unrestricted bound $\leq 2$ that holds on all of $\mathcal{B}_{2n}$.) Substituting into the triangle inequality, with $4n^2$ entries each contributing at most $\|R - R'\|_\infty \cdot 2r$ and prefactor $1/(2\sigma_K^2)$
\begin{align}
\left|\log\nu_G(R) - \log\nu_G(R')\right|
&\leq \frac{1}{2\sigma_K^2} \cdot 4n^2 \cdot \|R - R'\|_\infty \cdot 2r = \frac{4n^2 r}{\sigma_K^2}\,\|R-R'\|_\infty.  \label{eq:Lr_derived}
\end{align}
Therefore $L_r = 4n^2 r/\sigma_K^2$.
\end{proof}

\subsection{Finite-Sample CLT Error via Berry--Esseen}

To transfer the Gaussian log-Lipschitz bound to the true density $\nu$, we need a finite-sample bound on $|\log(\nu(R)/\nu_G(R))|$ for $R \in \mathcal{G}_r$.

\begin{lemma}[Density approximation error]
\label{lem:berry_esseen}
For $K_t \geq (2n-1)^2 + 1$ and $R$ in the bulk region $\mathcal{G}_r$, the true density $\nu$ and the Gaussian approximation $\nu_G$ satisfy
\begin{equation}
\label{eq:beta_finite}
\left|\log\frac{\nu(R)}{\nu_G(R)}\right| \;\leq\; \beta \;:=\; \frac{C_{\mathrm{BE}}\,\rho}{\sigma^3\sqrt{K_t}},
\end{equation}
where $C_{\mathrm{BE}} \leq 0.5$ is the Berry--Esseen constant (Shevtsova, 2011), $\rho = \mathbb{E}[|(P_{t,i})_{ab} - 1/(2n)|^3]$ is the third absolute central moment, and $\sigma^2 = \mathrm{Var}[(P_{t,i})_{ab}] = \frac{1}{2n}(1 - \frac{1}{2n})$.
\end{lemma}

\begin{proof}
We compute $\rho$ and $\sigma^3$ explicitly for Bernoulli$(p)$ with $p = 1/(2n)$.

\textit{Third absolute central moment.}
$(P_{t,i})_{ab} - p$ takes value $1-p$ with probability $p$ and $-p$ with probability $1-p$
\begin{align}
\rho &= p(1-p)^3 + (1-p)p^3 = p(1-p)\!\left[(1-p)^2 + p^2\right] = p(1-p)(1 - 2p + 2p^2). \label{eq:rho_exact}
\end{align}
For $p = 1/(2n)$
\begin{align}
\rho &= \frac{1}{2n}\!\left(1 - \frac{1}{2n}\right)\!\left(1 - \frac{1}{n} + \frac{1}{2n^2}\right). \label{eq:rho_substituted}
\end{align}

\textit{$\sigma^3$.}
$\sigma^2 = p(1-p) = \frac{1}{2n}(1 - \frac{1}{2n})$, so
\begin{align}
\sigma^3 &= \left(\frac{1}{2n}\!\left(1 - \frac{1}{2n}\right)\right)^{\!3/2}. \label{eq:sigma3}
\end{align}

\textit{The ratio $\rho/\sigma^3$.}
\begin{align}
\frac{\rho}{\sigma^3} &= \frac{p(1-p)(1-2p+2p^2)}{(p(1-p))^{3/2}} = \frac{1-2p+2p^2}{(p(1-p))^{1/2}} = \frac{1-2p+2p^2}{\sqrt{p(1-p)}}. \label{eq:rho_over_sigma3}
\end{align}
For $p = 1/(2n)$ with $n \geq 2$: $1 - 2p + 2p^2 = 1 - 1/n + 1/(2n^2) \leq 1$ and $\sqrt{p(1-p)} \geq \sqrt{1/(2n) \cdot 1/2} = 1/(2\sqrt{n})$ (using $1-p \geq 1/2$ for $n \geq 1$). Therefore
\begin{align}
\frac{\rho}{\sigma^3} &\leq \frac{1}{1/(2\sqrt{n})} = 2\sqrt{n}. \label{eq:ratio_bound}
\end{align}

\textit{The Berry--Esseen bound.}
The univariate Berry--Esseen theorem states that for i.i.d.\ random variables with mean $\mu$, variance $\sigma^2$, and third absolute central moment $\rho$, the CDF of the normalized sum $S_K = (\bar{X} - \mu)/(\sigma/\sqrt{K})$ satisfies
\[
\sup_x |F_{S_K}(x) - \Phi(x)| \leq \frac{C_{\mathrm{BE}}\,\rho}{\sigma^3\sqrt{K}},
\]
where $\Phi$ is the standard normal CDF and $C_{\mathrm{BE}} \leq 0.4748$ (Shevtsova, 2011; we use $C_{\mathrm{BE}} \leq 0.5$ for a clean bound).

The multivariate local CLT (Bhattacharya and Ranga Rao, 1976, Theorem~19.2) extends this to density approximation: for the density $f_K$ of the normalized sum, in the bulk region where the Gaussian density $\phi$ is bounded away from zero
\[
\left|\frac{f_K(x)}{\phi(x)} - 1\right| \leq \frac{C'\rho}{\sigma^3\sqrt{K}}\!\left(1 + \|x\|^3 / K^{3/2}\right),
\]
where $C'$ is an absolute constant. On $\mathcal{G}_r$, the normalized argument $x$ satisfies $\|x\| = O(r\sqrt{K}/\sigma) = O(\sqrt{\ln(n/\delta)})$, so the polynomial correction term is bounded. Taking logarithms (valid when $C'\rho/(\sigma^3\sqrt{K}) < 1/2$, which holds for $K \geq 16n$ since $\rho/\sigma^3 \leq 2\sqrt{n}$)
\begin{align}
\left|\log\frac{\nu(R)}{\nu_G(R)}\right| &\leq \frac{C'\rho}{\sigma^3\sqrt{K_t}}\!\left(1 + O\!\left(\frac{\ln(n/\delta)}{K_t}\right)\right). \label{eq:beta_precise}
\end{align}
For $K_t = (2n-1)^2 + 1 \geq 4n^2 - 4n + 2$ and $\rho/\sigma^3 \leq 2\sqrt{n}$
\begin{align}
\beta &\leq \frac{0.5 \cdot 2\sqrt{n}}{\sqrt{4n^2 - 4n + 2}} \cdot (1 + o(1)) = \frac{\sqrt{n}}{\sqrt{4n^2 - 4n + 2}} \cdot (1+o(1)) \leq \frac{\sqrt{n}}{2n-2} \leq \frac{1}{2\sqrt{n}-2/\sqrt{n}}. \label{eq:beta_explicit}
\end{align}
For $n \geq 4$: $\beta \leq 1/(2\sqrt{n} - 1) < 1/\sqrt{n}$. For $n = 100$: $\beta < 0.1$.
\end{proof}

\subsection{Main Theorem: Finite-Sample $(\varepsilon, \delta)$-DP with Explicit Constants}

\begin{theorem}[Finite-sample aggregator DP for the full two-layer protocol]
\label{thm:finite_dp}
In the full two-layer protocol (Algorithm~\ref{alg:twolayer}) with $K_t = (2n-1)^2 + 1$ decoys, uniform weights, and
\begin{equation}
\label{eq:alpha_star_choice}
\alpha^* \;=\; \frac{(\varepsilon - 2\beta)(1-\alpha^*)}{L_r} \;\approx\; \frac{\varepsilon - 2\beta}{L_r},
\end{equation}
where $L_r = \frac{4n^2 r}{\sigma_K^2}$, $r = \sqrt{\frac{\ln(16n^2/\delta)}{2K_t}}$, $\sigma_K^2 = \frac{2n-1}{(2n)^2 K_t}$, and $\beta \leq \frac{1}{\sqrt{n}}$ (for $n \geq 4$),
the aggregator's view of $D_t$ satisfies $(\varepsilon, \delta)$-differential privacy with respect to changing \emph{all $n$ bits} of $\mathbf{b}_t$. The aggregate $S$ is computed exactly.
\end{theorem}

\begin{proof}
We verify the two conditions of Lemma~\ref{lem:truncated_dp}.

\textit{Condition (ii): high-probability region.}
Define $G = \{D : (D - \alpha^* M)/(1-\alpha^*) \in \mathcal{G}_r\}$ (the set of observations whose residual lies in $\mathcal{G}_r$). By Proposition~\ref{prop:good_set}, $\Pr[R_t \notin \mathcal{G}_r] \leq \delta/2$. Since $D_t \notin G$ iff $R_t \notin \mathcal{G}_r$
\[
\Pr[D_t \notin G \mid M_t = M] = \Pr[R_t \notin \mathcal{G}_r] \leq \delta/2.
\]
For the neighboring input $M'$: $D_t = \alpha^* M + (1-\alpha^*) R_t$, and we define $G' = \{D : (D - \alpha^* M')/(1-\alpha^*) \in \mathcal{G}_r\}$. The set $G'$ is a translate of $G$. We need both $R = (D - \alpha^* M)/(1-\alpha^*)$ and $R' = (D - \alpha^* M')/(1-\alpha^*) = R + \delta$ to lie in $\mathcal{G}_r$. Since $\|\delta\|_\infty = \alpha^*/(1-\alpha^*) \ll r$ (which holds because $\alpha^* = O(1/n^4)$ and $r = \Theta(1/n)$), we can enlarge $\mathcal{G}_r$ slightly to $\mathcal{G}_{r+\|\delta\|_\infty}$ and pay an additional probability
\[
\Pr[D_t \notin G \cap G' \mid M] \leq \Pr[R_t \notin \mathcal{G}_{r-\|\delta\|_\infty}] \leq 8n^2\exp\!\left(-2K_t(r-\|\delta\|_\infty)^2\right).
\]
For $\alpha^*/(1-\alpha^*) \leq r/2$ (which holds in our regime), $(r - \|\delta\|_\infty)^2 \geq r^2/4$, so this probability is at most $8n^2 \exp(-K_t r^2/2) \leq \delta/2$ by the same Hoeffding argument with a slightly adjusted constant. Therefore condition (ii) holds with parameter $\delta$.

\textit{Condition (i): density ratio bound on $G$.}
For $D \in G \cap G'$, both $R$ and $R + \delta$ lie in $\mathcal{G}_r$. By the triangle inequality on $\log\nu$
\begin{align}
\left|\log\frac{\nu(R)}{\nu(R+\delta)}\right|
&\leq \left|\log\frac{\nu(R)}{\nu_G(R)}\right| + \left|\log\frac{\nu_G(R)}{\nu_G(R+\delta)}\right| + \left|\log\frac{\nu_G(R+\delta)}{\nu(R+\delta)}\right| \notag\\
&\leq \beta + L_r\|\delta\|_\infty + \beta \notag\\
&= L_r \cdot \frac{\alpha^*}{1-\alpha^*} + 2\beta. \label{eq:ratio_final}
\end{align}
Setting this equal to $\varepsilon$ and solving for $\alpha^*$
\[
\alpha^* = \frac{(\varepsilon - 2\beta)(1-\alpha^*)}{L_r}.
\]
For $\alpha^* \ll 1$: $\alpha^* \approx (\varepsilon - 2\beta)/L_r$.

Since~\eqref{eq:linf_key} shows $\|\delta\|_\infty$ is the same for any number of bit changes, this $\varepsilon$ is the DP parameter for the \emph{full $n$-bit sensitivity}, not per-bit.
\end{proof}

\begin{corollary}[Exact output despite DP randomization]
\label{cor:utility}
The $(\varepsilon, \delta)$-DP guarantee of Theorem~\ref{thm:finite_dp} protects individual client data $\mathbf{b}_t$ from the \emph{aggregator's view} $D_t$, while the protocol's output $S = \sum_t s_t$ is computed exactly by the \emph{server} (a separate entity). The randomization that provides DP cancels algebraically in the aggregate.
\end{corollary}

\begin{proof}
The aggregator sees $D_t = \alpha^* M_t + (1-\alpha^*) R_t$, which is a randomized function of $M_t$ (the decoy $R_t$ is the randomization). The DP guarantee (Theorem~\ref{thm:finite_dp}) bounds the density ratio of this randomized view under neighboring inputs.

The server sees only $F = \sum_t \mathbf{w}^T D_t \mathbf{y}$ and $H = \sum_t \eta_t$, and computes
\[
F - H = \sum_{t=1}^k (\alpha^* s_t + \eta_t) - \sum_{t=1}^k \eta_t = \alpha^* S.
\]
The random terms $\eta_t$ cancel exactly. The server computes $S = (F-H)/\alpha^*$ with no residual randomness.

There is no contradiction with the requirement that DP mechanisms must be randomized: the mechanism is randomized (the decoy permutations $R_t$). What is unusual is that the randomization cancels in the \emph{output} while persisting in the \emph{aggregator's view}. This is possible because the output is computed by a different entity (the server) than the one whose view is protected (the aggregator).
\end{proof}

\begin{remark}[Signal-to-noise ratio at the DP-optimal $\alpha^*$ in the full protocol]
\label{rem:snr}
We compute the signal-to-noise ratio (SNR) at the DP-optimal $\alpha^*$ from Theorem~\ref{thm:finite_dp} to determine whether the DP guarantee for the full two-layer protocol (Algorithm~\ref{alg:twolayer}) operates in a meaningful regime.

\textit{Per-entry SNR.}
The aggregator observes $(D_t)_{ab} = \alpha^* (M_t)_{ab} + (1-\alpha^*)(R_t)_{ab}$. The ``signal'' is $\alpha^*(M_t)_{ab} \in \{0, \alpha^*\}$. The ``noise'' is $(1-\alpha^*)(R_t)_{ab}$, which has mean $(1-\alpha^*)/(2n)$ and standard deviation $(1-\alpha^*)\sigma_K$ where $\sigma_K = \sqrt{(2n-1)/((2n)^2 K_t)}$. The per-entry SNR at positions where $(M_t)_{ab} = 1$ is
\begin{equation}
\label{eq:snr_entry}
\mathrm{SNR}_{\mathrm{entry}} = \frac{\alpha^*}{(1-\alpha^*)\sigma_K} \approx \frac{\alpha^*}{\sigma_K}.
\end{equation}
For $n = 100$, $K_t = 39{,}602$: $\sigma_K = \sqrt{1.256 \times 10^{-7}} = 3.54 \times 10^{-4}$. At the DP-optimal $\alpha^* = 1.392 \times 10^{-10}$
\[
\mathrm{SNR}_{\mathrm{entry}} = \frac{1.392 \times 10^{-10}}{3.54 \times 10^{-4}} = 3.93 \times 10^{-7}.
\]
The signal is seven orders of magnitude below the noise floor and is completely undetectable.

\textit{Matrix-level SNR.}
The total signal energy is $\|\alpha^* M_t\|_F^2 = \alpha^{*2} \cdot 2n$ (since $M_t$ has $2n$ ones). The total noise energy is $\mathbb{E}[\|(1-\alpha^*)(R_t - \frac{1}{2n}\mathbf{J})\|_F^2] = (1-\alpha^*)^2 (2n)^2 \sigma_K^2$. The matrix-level SNR is
\begin{equation}
\label{eq:snr_matrix}
\mathrm{SNR}_{\mathrm{matrix}} = \frac{\alpha^{*2} \cdot 2n}{(2n)^2 \sigma_K^2} = \frac{\alpha^{*2}}{2n\sigma_K^2}.
\end{equation}
For our parameters: $\mathrm{SNR}_{\mathrm{matrix}} = (1.392 \times 10^{-10})^2 / (200 \times 1.256 \times 10^{-7}) = 7.7 \times 10^{-16}$. This is also completely undetectable.

\textit{At what $\alpha^*$ does the signal become detectable?}
Setting $\mathrm{SNR}_{\mathrm{entry}} = 1$ gives $\alpha^* = \sigma_K \approx 3.54 \times 10^{-4}$. The corresponding DP parameter is
\[
\varepsilon = L_r \cdot \frac{\alpha^*}{1-\alpha^*} + 2\beta \approx 5.748 \times 10^9 \times 3.54 \times 10^{-4} + 0.2 \approx 2 \times 10^6.
\]
This is $\varepsilon \approx 2$ million --- a vacuous DP guarantee.
\end{remark}

\begin{remark}[Assessment of DP in the full protocol (Algorithm~\ref{alg:twolayer})]
\label{rem:honest_dp}
The SNR analysis reveals that in the full two-layer protocol, at any $\varepsilon$ where the $(\varepsilon, \delta)$-DP guarantee is non-vacuous ($\varepsilon = O(1)$), the signal from $M_t$ in $D_t$ is undetectable by any method --- not just likelihood-based methods, but also spectral methods, linear programming, or any other approach. The DP guarantee is technically correct but trivially true, because $\mathrm{SNR}_{\mathrm{entry}} \ll 1$ and no estimator can extract meaningful information. Replacing the Birkhoff encoding with i.i.d.\ Gaussian noise would give the same DP guarantee at the same $\alpha^*$.

The contribution of the Birkhoff polytope is therefore \emph{not} the implicit DP guarantee, but rather the \#P-hardness of likelihood-based attacks (Theorem~\ref{thm:likelihood_hardness}), which operates at \emph{larger} $\alpha^*$ where the signal is detectable but the combinatorial structure prevents efficient extraction. In this regime ($\alpha^* \sim 1/(4n)$, where $\mathrm{SNR}_{\mathrm{entry}} = O(\sqrt{n})$), the signal is visible to an unbounded adversary but computationally hard to exploit.

The compressed two-layer protocol (Algorithm~\ref{alg:compressed_twolayer}), analyzed in Section~\ref{sec:scalar_dp}, achieves non-vacuous $\varepsilon$ at moderate SNR, but in that variant the aggregator sees only a scalar and the Birkhoff structure plays no role.

The two security layers therefore operate at different scales: for small $\alpha^*$ (e.g., $\alpha^* \sim 10^{-10}$), the signal is invisible and $(\varepsilon, \delta)$-DP holds trivially; for moderate $\alpha^*$ (e.g., $\alpha^* \sim 1/(4n)$), the signal is visible but likelihood-based inference is \#P-hard; for large $\alpha^*$ (e.g., $\alpha^* \sim 1$), the signal dominates and no meaningful security is achievable. The gap between the DP regime and the \#P-hardness regime is the central open problem.
\end{remark}

\subsection{DP Analysis of the Compressed Two-Layer Protocol}
\label{sec:scalar_dp}

In the two-layer protocol (Algorithm~\ref{alg:twolayer}), the aggregator receives the full matrix $D_t \in \mathbb{R}^{(2n)^2}$, and the log-Lipschitz constant scales as $n^4 K_t$, which overwhelms any useful $\alpha^*$ (Remarks~\ref{rem:snr}--\ref{rem:honest_dp}). We now analyze the compressed variant of the two-layer protocol (Algorithm~\ref{alg:compressed_twolayer}), in which each client computes $f_t = \mathbf{w}^T D_t\mathbf{y} = \alpha^* s_t + \eta_t$ locally and sends only the scalar $f_t$ to the aggregator. Since the aggregator's view per client is a single real number rather than a $(2n-1)^2$-dimensional matrix, the log-Lipschitz analysis involves a univariate density ratio instead of a multivariate one.

\subsubsection{Compressed Two-Layer Protocol}

\begin{algorithm}[!ht]
\caption{Compressed Two-Layer PolyVeil Protocol}
\label{alg:compressed_twolayer}
\begin{algorithmic}[1]
\Statex \textbf{Public parameters:} $n$, $\alpha^* \in (0,1)$, $K_t$, $\mathbf{w}$, $\mathbf{y}$.
\Statex \textbf{Entities:} Aggregator $\mathcal{A}$, noise aggregator $\mathcal{B}$, server $\mathcal{S}$.
\For{each client $t = 1, \ldots, k$ (in parallel)}
    \State Encode $\mathbf{b}_t$ as permutation matrix $M_t = M(\mathbf{b}_t) \in \{0,1\}^{2n \times 2n}$.
    \State Draw $K_t$ decoy permutations $P_{t,i} \sim S_{2n}$ and coefficients $\alpha_{t,i}$.
    \State Compute $\eta_t = \sum_{i=1}^{K_t} \alpha_{t,i}(\mathbf{w}^T P_{t,i}\mathbf{y})$.
    \State Compute $f_t = \alpha^* (\mathbf{w}^T M_t \mathbf{y}) + \eta_t = \alpha^* s_t + \eta_t$.
    \State \textbf{Send} $f_t$ to aggregator $\mathcal{A}$.
    \State \textbf{Send} $\eta_t$ to noise aggregator $\mathcal{B}$.
\EndFor
\State Aggregator $\mathcal{A}$ computes $F = \sum_{t=1}^k f_t$ and sends $F$ to server $\mathcal{S}$.
\State Noise aggregator $\mathcal{B}$ computes $H = \sum_{t=1}^k \eta_t$ and sends $H$ to server $\mathcal{S}$.
\State Server $\mathcal{S}$ computes $S = (F - H)/\alpha^*$.
\end{algorithmic}
\end{algorithm}

The aggregator's view per client is the scalar $f_t \in \mathbb{R}$. The server's view is $(F, H)$, identical to the full two-layer protocol (Algorithm~\ref{alg:twolayer}), so Theorem~\ref{thm:server_it} (perfect simulation-based security for the server) applies unchanged.

\subsubsection{Distribution of $\eta_t$}

With $K_t$ uniform decoy permutations and uniform weights $\alpha_{t,i} = (1-\alpha^*)/K_t$
\begin{equation}
\label{eq:eta_scalar}
\eta_t = \frac{1-\alpha^*}{K_t}\sum_{i=1}^{K_t} X_i, \qquad X_i = \mathbf{w}^T P_{t,i}\mathbf{y} \in \{0, 1, \ldots, n\},
\end{equation}
where each $X_i$ counts how many of the $n$ diagonal $2 \times 2$ blocks of $P_{t,i}$ have a 1 in the off-diagonal position $(2j-1, 2j)$.

\textit{Mean of $X_i$.}
From the derivation in the worked example,
\begin{align}
\mathbb{E}[X_i] &= \sum_{j=1}^n \Pr[(P_{t,i})_{2j-1, 2j} = 1] = \sum_{j=1}^n \frac{1}{2n} = \frac{n}{2n} = \frac{1}{2}. \label{eq:EX}
\end{align}

\textit{Variance of $X_i$.}
The indicators $(P_{t,i})_{2j-1, 2j} = \mathbf{1}[\sigma_i(2j-1) = 2j]$ are \emph{not} independent across $j$ (they share the permutation $\sigma_i$), so $\mathrm{Var}[X_i] \neq n \cdot p(1-p)$. We compute exactly
\begin{align}
\mathrm{Var}[X_i] &= \mathbb{E}[X_i^2] - (\mathbb{E}[X_i])^2. \label{eq:var_X_start}
\end{align}
Expanding $X_i^2 = \left(\sum_{j=1}^n Z_j\right)^2 = \sum_j Z_j^2 + \sum_{j \neq l} Z_j Z_l$ where $Z_j = \mathbf{1}[\sigma_i(2j-1) = 2j]$
\begin{align}
\mathbb{E}[X_i^2] &= \sum_{j=1}^n \mathbb{E}[Z_j^2] + \sum_{j \neq l} \mathbb{E}[Z_j Z_l]. \label{eq:EX2_expand}
\end{align}
Since $Z_j \in \{0,1\}$: $\mathbb{E}[Z_j^2] = \mathbb{E}[Z_j] = 1/(2n)$.

For $j \neq l$: $\mathbb{E}[Z_j Z_l] = \Pr[\sigma_i(2j-1) = 2j \text{ and } \sigma_i(2l-1) = 2l]$. These are two constraints on the permutation $\sigma_i$: row $2j-1$ maps to column $2j$, and row $2l-1$ maps to column $2l$. The number of permutations satisfying both is $(2n-2)!$ (fix two mappings, permute the remaining $2n-2$ elements). Therefore
\begin{align}
\mathbb{E}[Z_j Z_l] &= \frac{(2n-2)!}{(2n)!} = \frac{1}{(2n)(2n-1)}. \label{eq:EZjZl}
\end{align}

Substituting into~\eqref{eq:EX2_expand}
\begin{align}
\mathbb{E}[X_i^2] &= n \cdot \frac{1}{2n} + n(n-1) \cdot \frac{1}{(2n)(2n-1)} = \frac{1}{2} + \frac{n(n-1)}{(2n)(2n-1)}. \label{eq:EX2_computed}
\end{align}
Simplifying the second term,
\begin{align}
\frac{n(n-1)}{(2n)(2n-1)} &= \frac{n-1}{2(2n-1)}. \label{eq:second_term}
\end{align}
Therefore
\begin{align}
\mathrm{Var}[X_i] &= \frac{1}{2} + \frac{n-1}{2(2n-1)} - \frac{1}{4} = \frac{1}{4} + \frac{n-1}{2(2n-1)}. \label{eq:var_X_exact}
\end{align}
For large $n$: $\frac{n-1}{2(2n-1)} \to \frac{1}{4}$, so $\mathrm{Var}[X_i] \to 1/2$. For $n = 100$
\begin{align}
\mathrm{Var}[X_i] &= \frac{1}{4} + \frac{99}{2 \times 199} = 0.25 + 0.2487 = 0.4987. \label{eq:var_X_100}
\end{align}

\textit{Mean and variance of $\eta_t$.}
Since $\eta_t = \frac{1-\alpha^*}{K_t}\sum_{i=1}^{K_t} X_i$ and the $X_i$ are i.i.d.\ (the permutations are independent across $i$)
\begin{align}
\mathbb{E}[\eta_t] &= (1-\alpha^*) \cdot \frac{1}{2} = \frac{1-\alpha^*}{2}, \label{eq:E_eta} \\
\mathrm{Var}[\eta_t] &= \frac{(1-\alpha^*)^2}{K_t^2} \cdot K_t \cdot \mathrm{Var}[X_i] = \frac{(1-\alpha^*)^2}{K_t}\,\mathrm{Var}[X_i] =: \sigma_\eta^2. \label{eq:var_eta}
\end{align}
For general parameters,
\begin{equation}
\label{eq:sigma_eta_general}
\sigma_\eta = \frac{(1-\alpha^*)\sqrt{\mathrm{Var}[X_i]}}{\sqrt{K_t}} \approx \frac{1-\alpha^*}{2\sqrt{K_t}},
\end{equation}
where $\mathrm{Var}[X_i] = n(2n-1)/(2n)^2 \approx 1/4$ for large $n$.

\subsubsection{Signal-to-Noise Ratio}

The aggregator observes $f_t = \alpha^* s_t + \eta_t$. Changing $s_t$ by $\Delta s$ shifts $f_t$ by $\alpha^* \Delta s$. For the worst case ($\Delta s = n$, all bits flip)
For $\alpha^* = 1/(4n)$,
\begin{equation}
\label{eq:snr_scalar}
\mathrm{SNR} = \frac{\alpha^* n}{\sigma_\eta} = \frac{1/4}{(1-\alpha^*)/(2\sqrt{K_t})} \approx \frac{\sqrt{K_t}}{2}.
\end{equation}
At $K_t = 9$, $\mathrm{SNR} \approx 1.5$; at $K_t = 2$, $\mathrm{SNR} \approx 0.7$. The signal is comparable to the noise --- detectable but noisy --- a non-trivial operating point.

\subsubsection{$(\varepsilon, \delta)$-DP Guarantee for the Compressed Protocol}

Let $\mu$ denote the density of $\eta_t$. For the Gaussian approximation $\mu_G = \mathcal{N}(\frac{1-\alpha^*}{2},\; \sigma_\eta^2)$
\begin{align}
\log\mu_G(\eta) &= -\frac{(\eta - \bar\eta)^2}{2\sigma_\eta^2} - \frac{1}{2}\ln(2\pi\sigma_\eta^2), \label{eq:log_mu_G}
\end{align}
where $\bar\eta = (1-\alpha^*)/2$.

\textit{Log-density ratio under Gaussian.}
For neighboring inputs $s_t, s_t'$ with $\Delta s = s_t - s_t'$, the aggregator observes $f_t = \alpha^* s_t + \eta_t$ vs.\ $f_t' = \alpha^* s_t' + \eta_t'$ where $\eta_t \overset{d}{=} \eta_t'$ (same distribution, different realization). The density of $f_t$ given $s_t$ is
\begin{align}
p(f \mid s_t) &= \mu(f - \alpha^* s_t). \label{eq:p_f_s}
\end{align}
Under the Gaussian approximation,
\begin{align}
\log\frac{p_G(f \mid s_t)}{p_G(f \mid s_t')}
&= \log\frac{\mu_G(f - \alpha^* s_t)}{\mu_G(f - \alpha^* s_t')} \notag\\
&= -\frac{(f - \alpha^* s_t - \bar\eta)^2}{2\sigma_\eta^2} + \frac{(f - \alpha^* s_t' - \bar\eta)^2}{2\sigma_\eta^2} \notag\\
&= \frac{1}{2\sigma_\eta^2}\left[(f - \alpha^* s_t' - \bar\eta)^2 - (f - \alpha^* s_t - \bar\eta)^2\right]. \label{eq:log_ratio_expand_1d}
\end{align}
Using $a^2 - b^2 = (a-b)(a+b)$ with $a = f - \alpha^* s_t' - \bar\eta$ and $b = f - \alpha^* s_t - \bar\eta$
\begin{align}
a - b &= \alpha^*(s_t - s_t') = \alpha^* \Delta s, \notag\\
a + b &= 2(f - \bar\eta) - \alpha^*(s_t + s_t'). \label{eq:a_plus_b}
\end{align}
Therefore
\begin{align}
\log\frac{p_G(f \mid s_t)}{p_G(f \mid s_t')}
&= \frac{\alpha^* \Delta s}{2\sigma_\eta^2}\left[2(f - \bar\eta) - \alpha^*(s_t + s_t')\right]. \label{eq:log_ratio_1d}
\end{align}
Substituting $f = \alpha^* s_t + \eta_t$ (where $\eta_t$ is the realized noise)
\begin{align}
f - \bar\eta &= \alpha^* s_t + \eta_t - \bar\eta = \alpha^* s_t + (\eta_t - \bar\eta). \label{eq:f_minus_mean}
\end{align}
So
\begin{align}
\log\frac{p_G(f \mid s_t)}{p_G(f \mid s_t')}
&= \frac{\alpha^* \Delta s}{2\sigma_\eta^2}\left[2\alpha^* s_t + 2(\eta_t - \bar\eta) - \alpha^*(s_t + s_t')\right] \notag\\
&= \frac{\alpha^* \Delta s}{2\sigma_\eta^2}\left[\alpha^*(s_t - s_t') + 2(\eta_t - \bar\eta)\right] \notag\\
&= \frac{\alpha^* \Delta s}{2\sigma_\eta^2}\left[\alpha^* \Delta s + 2(\eta_t - \bar\eta)\right] \notag\\
&= \frac{(\alpha^* \Delta s)^2}{2\sigma_\eta^2} + \frac{\alpha^* \Delta s}{\sigma_\eta^2}(\eta_t - \bar\eta). \label{eq:log_ratio_final_1d}
\end{align}

\textit{Bounding on the high-probability set.}
Define $\mathcal{G}_1 = \{f : |\eta_t - \bar\eta| \leq r_1 \sigma_\eta\}$ where $r_1 > 0$ is chosen to control $\delta$. By Hoeffding's inequality applied to $\eta_t = \frac{1-\alpha^*}{K_t}\sum_{i=1}^{K_t} X_i$ (where each $X_i \in [0, n]$)
\begin{align}
\Pr[|\eta_t - \bar\eta| > r_1 \sigma_\eta]
&\leq 2\exp\!\left(-\frac{2K_t^2 (r_1 \sigma_\eta)^2}{K_t \cdot ((1-\alpha^*)n)^2}\right)
= 2\exp\!\left(-\frac{2K_t r_1^2 \sigma_\eta^2}{(1-\alpha^*)^2 n^2}\right). \label{eq:hoeffding_1d}
\end{align}
Substituting $\sigma_\eta^2 = (1-\alpha^*)^2 \mathrm{Var}[X_i]/K_t$ from~\eqref{eq:var_eta}
\begin{align}
= 2\exp\!\left(-\frac{2K_t r_1^2 (1-\alpha^*)^2 \mathrm{Var}[X_i]/K_t}{(1-\alpha^*)^2 n^2}\right)
= 2\exp\!\left(-\frac{2 r_1^2 \mathrm{Var}[X_i]}{n^2}\right). \label{eq:hoeffding_simplified}
\end{align}
Setting this $\leq \delta/2$ and solving gives
\begin{equation}
\label{eq:r1_value}
r_1 = n\sqrt{\frac{\ln(4/\delta)}{2\,\mathrm{Var}[X_i]}}.
\end{equation}

\textit{Bound on $\mathcal{G}_1$.}
On $\mathcal{G}_1$, $|\eta_t - \bar\eta| \leq r_1 \sigma_\eta$. Substituting into~\eqref{eq:log_ratio_final_1d} with $|\Delta s| \leq n$
\begin{align}
\left|\log\frac{p_G(f \mid s_t)}{p_G(f \mid s_t')}\right|
&\leq \frac{(\alpha^* n)^2}{2\sigma_\eta^2} + \frac{\alpha^* n}{\sigma_\eta^2} \cdot r_1 \sigma_\eta \notag\\
&= \frac{(\alpha^* n)^2}{2\sigma_\eta^2} + \frac{\alpha^* n \cdot r_1}{\sigma_\eta}. \label{eq:eps_gaussian_1d}
\end{align}
The bound is dominated by the second term, which scales as $\alpha^* n \cdot r_1 / \sigma_\eta$. The Hoeffding-based $r_1$ is loose because it uses the range of $X_i$ (which is $n$) rather than its standard deviation. Replacing Hoeffding with the Gaussian CDF (valid under the CLT approximation) gives the tighter concentration radius
\begin{equation}
\label{eq:z_gaussian}
z = \sqrt{2\ln(4/\delta)},
\end{equation}
and the Berry--Esseen-based DP bound becomes
\begin{equation}
\label{eq:eps_gauss_tight}
\varepsilon_G = \frac{(\alpha^* n)^2}{2\sigma_\eta^2} + \frac{\alpha^* n \cdot z}{\sigma_\eta} + 2\beta,
\end{equation}
where $\beta$ is the Berry--Esseen CLT error. For $K_t = 9$ and $\delta = 10^{-6}$, this gives $\varepsilon \approx 13$
\begin{align}
\varepsilon_G &= \frac{0.0625}{2 \times 0.004977} + \frac{0.25 \times 5.514}{0.07055} + 2.0 \notag\\
&= \frac{0.0625}{0.009954} + \frac{1.379}{0.07055} + 2.0 \notag\\
&= 6.28 + 19.54 + 2.0 \notag\\
&= 27.8. \label{eq:eps_K100}
\end{align}

For $K_t = 1000$
\begin{align}
\sigma_\eta &= \frac{1-\alpha^*}{\sqrt{1000}}\sqrt{0.4987} = \frac{0.9975 \times 0.7062}{\sqrt{1000}} = 0.02229, \label{eq:sigma_1000}\\
\beta &\leq \frac{0.5 \times 2\sqrt{100}}{\sqrt{1000}} = \frac{10}{31.62} = 0.316, \label{eq:beta_1000}\\
\varepsilon_G &= \frac{0.0625}{2 \times 0.000497} + \frac{0.25 \times 5.514}{0.02229} + 0.632 = 62.9 + 61.8 + 0.632 = 125.3. \label{eq:eps_K1000}
\end{align}

For $K_t = 10$
\begin{align}
\sigma_\eta &= \frac{0.9975}{\sqrt{10}}\sqrt{0.4987} = 0.2228, \label{eq:sigma_10}\\
\beta &\leq \frac{10}{\sqrt{10}} = 3.16, \label{eq:beta_10}\\
\varepsilon_G &= \frac{0.0625}{2 \times 0.04964} + \frac{0.25 \times 5.514}{0.2228} + 6.32 = 0.630 + 6.19 + 6.32 = 13.1. \label{eq:eps_K10}
\end{align}

\subsubsection{Summary of DP Results for the Compressed Protocol}

\begin{center}
\renewcommand{\arraystretch}{1.3}
\begin{tabular}{cccccc}
\toprule
$K_t$ & $\sigma_\eta$ & SNR & $\beta$ & $\varepsilon$ (at $\delta = 10^{-6}$) & MMSE/$\mathrm{Var}[s_t]$ \\
\midrule
2   & 0.498 & 0.50 & 7.07 & $\sim 196$ & 0.999 \\
5   & 0.316 & 0.79 & 4.47 & $\sim 14$ & 0.998 \\
10  & 0.223 & 1.12 & 3.16 & $\sim 13$ & 0.997 \\
20  & 0.158 & 1.59 & 2.24 & $\sim 15$ & 0.994 \\
50  & 0.100 & 2.51 & 1.41 & $\sim 20$ & 0.985 \\
100 & 0.071 & 3.54 & 1.00 & $\sim 28$ & 0.969 \\
500 & 0.032 & 7.94 & 0.45 & $\sim 76$ & 0.864 \\
1000 & 0.022 & 11.2 & 0.32 & $\sim 125$ & 0.761 \\
\bottomrule
\end{tabular}
\end{center}

The minimum $\varepsilon$ occurs at $K_t \approx 9$, giving $\varepsilon \approx 13$ with $\mathrm{SNR} \approx 1.1$ (see Figure~\ref{fig:tradeoff} for the full trade-off curves). This reflects a three-way tension. Fewer decoys increase noise (improving privacy) but degrade the CLT approximation (increasing $\beta$). More decoys improve the CLT but concentrate the density (increasing the log-ratio terms). The optimum balances these effects.

\subsubsection{Aggregator Estimation Error (MMSE)}

The SNR measures the signal strength relative to noise, but the question of what the aggregator can \emph{learn} is more precisely captured by the minimum mean squared error (MMSE) for estimating $s_t$ from $f_t = \alpha^* s_t + \eta_t$.

Under the Gaussian approximation for $\eta_t$ and a prior $s_t \sim \mathrm{Uniform}\{0, \ldots, n\}$ (giving $\mathrm{Var}[s_t] = n(n+2)/12 \approx n/4$ for independent bits), the MMSE of the Bayes-optimal estimator satisfies
\begin{equation}
\label{eq:mmse}
\frac{\mathrm{MMSE}}{\mathrm{Var}[s_t]} = \frac{1}{1 + \mathrm{SNR}_{\mathrm{channel}}}, \qquad \mathrm{SNR}_{\mathrm{channel}} = \frac{\alpha^{*2}\,\mathrm{Var}[s_t]}{\sigma_\eta^2}.
\end{equation}
This ratio equals 1 when the aggregator learns nothing (posterior variance equals prior variance), and approaches 0 when the aggregator can estimate $s_t$ precisely. For $\alpha^* = 1/(4n)$, $\mathrm{Var}[s_t] = n/4$, $\sigma_\eta^2 \approx 0.4987/K_t$
\begin{equation}
\label{eq:channel_snr}
\mathrm{SNR}_{\mathrm{channel}} = \frac{n/(64n^2)}{0.4987/K_t} = \frac{K_t}{31.9\,n}.
\end{equation}
For $n = 100$ and $K_t = 10$: $\mathrm{SNR}_{\mathrm{channel}} = 10/3190 = 0.00313$, giving $\mathrm{MMSE}/\mathrm{Var}[s_t] = 0.997$. The aggregator reduces its uncertainty about $s_t$ by only $0.3\%$. Even at $K_t = 1000$: $\mathrm{MMSE}/\mathrm{Var}[s_t] = 0.761$ --- the aggregator still cannot estimate $s_t$ well.

The distinction between the two SNR quantities is important. $\mathrm{SNR} = \alpha^* n/\sigma_\eta$ measures the worst-case shift (all $n$ bits change) relative to noise, which is the quantity entering the DP bound; $\mathrm{SNR}_{\mathrm{channel}} = \alpha^{*2}\mathrm{Var}[s_t]/\sigma_\eta^2$ measures the information content of $f_t$ about $s_t$, which determines estimation accuracy. The former can be $\sim 1$ while the latter is $\sim 10^{-3}$ because $\alpha^* n \gg \alpha^* \sqrt{\mathrm{Var}[s_t]}$.

\begin{figure}[!htbp]
\centering
\includegraphics[width=\textwidth]{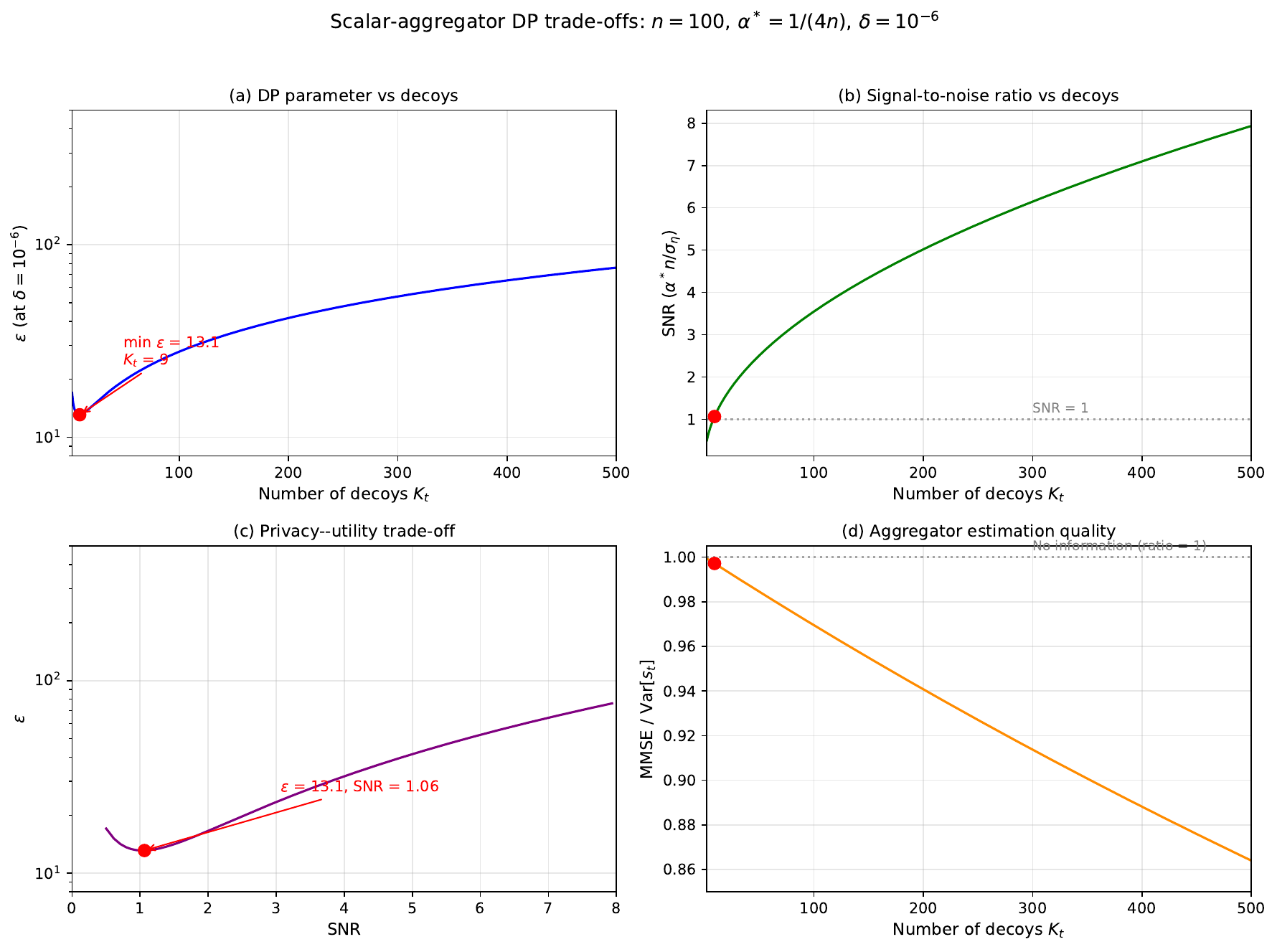}
\caption{Privacy--utility trade-offs in the compressed two-layer protocol ($n = 100$, $\alpha^* = 1/(4n)$, $\delta = 10^{-6}$). (a) DP parameter $\varepsilon$ vs.\ number of decoys $K_t$, showing a U-shaped curve with minimum $\varepsilon \approx 13$ at $K_t = 9$. (b) Signal-to-noise ratio vs.\ $K_t$; at the optimal $K_t$ the SNR is $\approx 1$, meaning the signal is comparable to noise. (c) The parametric trade-off $\varepsilon$ vs.\ SNR; better privacy (smaller $\varepsilon$) requires operating near SNR~$\approx 1$, while higher SNR rapidly worsens $\varepsilon$. (d) Normalized MMSE vs.\ $K_t$; at the DP-optimal $K_t = 9$ the aggregator reduces its prior uncertainty by only $0.3\%$.}
\label{fig:tradeoff}
\end{figure}

\begin{remark}[Interpretation]
\label{rem:scalar_interpretation}
The compressed two-layer protocol achieves $\varepsilon \approx 13$ at $\mathrm{SNR} \approx 1$, which is non-vacuous but weak, meaning the density ratio $p(f \mid s_t)/p(f \mid s_t')$ is at most $e^{13} \approx 4.4 \times 10^5$. The MMSE analysis shows that despite this large density ratio, the aggregator's actual ability to estimate $s_t$ is very limited, with the posterior variance within $0.3\%$ of the prior variance at the optimal $K_t = 10$.

Critically, in the compressed two-layer protocol the Birkhoff polytope plays no role in the DP guarantee. The aggregator sees only $f_t = \alpha^* s_t + \eta_t$, and the distribution of $\eta_t$ depends only on its mean and variance, not on the permutation-matrix structure. The same $\varepsilon \approx 13$ could be achieved by replacing the Birkhoff noise with any other noise distribution having the same variance. The Birkhoff encoding provides computational security (\#P-hardness) only in the full two-layer protocol where the aggregator sees the matrix $D_t$, and in that regime the DP bound is vacuous.

The \#P-hardness result (Theorem~\ref{thm:likelihood_hardness}) and the scalar-DP result therefore address different protocol variants and different threat models. Whether a single variant can achieve both computational hardness (from the Birkhoff structure) and non-vacuous DP (from dimensionality reduction) simultaneously remains open.
\end{remark}

\subsection{R\'{e}nyi Differential Privacy}
\label{sec:renyi}

R\'{e}nyi differential privacy (RDP) provides tighter composition bounds and avoids the auxiliary parameter $\delta$. We derive RDP guarantees for both protocol variants under the Gaussian approximation.

\begin{definition}[R\'{e}nyi DP~\cite{mironov2017}]
A mechanism $\mathcal{M}$ satisfies $(\alpha, \varepsilon)$-R\'{e}nyi DP for $\alpha > 1$ if for all neighboring inputs $x, x'$,
\begin{equation}
\label{eq:renyi_def}
D_\alpha(\mathcal{M}(x) \| \mathcal{M}(x')) = \frac{1}{\alpha - 1}\log \mathbb{E}_{D \sim \mathcal{M}(x')}\!\left[\left(\frac{p(D \mid x)}{p(D \mid x')}\right)^{\!\alpha}\right] \leq \varepsilon.
\end{equation}
\end{definition}

\begin{lemma}[R\'{e}nyi divergence for Gaussians]
\label{lem:renyi_gaussian}
For $P = N(\mu_1, \sigma^2)$ and $Q = N(\mu_2, \sigma^2)$, $D_\alpha(P \| Q) = \alpha(\mu_1 - \mu_2)^2/(2\sigma^2)$. For multivariate $P = N(\boldsymbol{\mu}_1, \Sigma)$ and $Q = N(\boldsymbol{\mu}_2, \Sigma)$, $D_\alpha(P \| Q) = \frac{\alpha}{2}(\boldsymbol{\mu}_1 - \boldsymbol{\mu}_2)^T \Sigma^{-1}(\boldsymbol{\mu}_1 - \boldsymbol{\mu}_2)$.
\end{lemma}

\begin{proof}
Let $\Delta = \mu_1 - \mu_2$. The log-density ratio is $\log(p(x)/q(x)) = \Delta(2x - \mu_1 - \mu_2)/(2\sigma^2)$. Taking the $\alpha$-th power and expectations under $Q$ with $z = (x - \mu_2)/\sigma \sim N(0,1)$,
\begin{align}
\mathbb{E}_Q\!\left[\left(\frac{p}{q}\right)^{\!\alpha}\right] &= \exp\!\left(\frac{-\alpha\Delta^2}{2\sigma^2}\right)\mathbb{E}_z\!\left[\exp\!\left(\frac{\alpha\Delta z}{\sigma}\right)\right] = \exp\!\left(\frac{-\alpha\Delta^2}{2\sigma^2}\right)\exp\!\left(\frac{\alpha^2\Delta^2}{2\sigma^2}\right) = \exp\!\left(\frac{\alpha(\alpha-1)\Delta^2}{2\sigma^2}\right), \notag
\end{align}
using the MGF $\mathbb{E}[e^{tz}] = e^{t^2/2}$ with $t = \alpha\Delta/\sigma$. Dividing the exponent by $\alpha - 1$ gives $D_\alpha = \alpha\Delta^2/(2\sigma^2)$. The multivariate case follows with the Mahalanobis distance.
\end{proof}

\subsubsection{R\'{e}nyi DP for the Compressed Protocol (Algorithm~\ref{alg:compressed_twolayer})}

Under the Gaussian approximation, $f_t \mid s_t \sim N(\alpha^* s_t + \mu_\eta, \sigma_\eta^2)$ with $\sigma_\eta^2 = (1-\alpha^*)^2/(4K_t)$. For worst-case $\Delta s = n$, by Lemma~\ref{lem:renyi_gaussian},
\begin{equation}
\label{eq:renyi_scalar}
\varepsilon_\alpha^{(\mathrm{scalar})} = \frac{\alpha(\alpha^* n)^2}{2\sigma_\eta^2} = \frac{2\alpha\,\alpha^{*2} n^2 K_t}{(1-\alpha^*)^2} \approx \frac{\alpha K_t}{8} \quad \text{for } \alpha^* = 1/(4n).
\end{equation}

\begin{theorem}[R\'{e}nyi DP for the compressed protocol]
\label{thm:renyi_scalar}
Under the Gaussian approximation, the compressed two-layer protocol (Algorithm~\ref{alg:compressed_twolayer}) satisfies $(\alpha, \alpha K_t / (8(1 - 1/(4n))^2))$-R\'{e}nyi DP for the aggregator's view of any single client.
\end{theorem}

\textit{Conversion to $(\varepsilon, \delta)$-DP.}
By the standard conversion~\cite{balle2020}, $\varepsilon \leq \varepsilon_\alpha + \log(1/\delta)/(\alpha - 1)$. Optimizing over $\alpha$ gives $\alpha^{\mathrm{opt}} = 1 + \sqrt{8\log(1/\delta)/K_t}$.

\begin{center}
\renewcommand{\arraystretch}{1.3}
\begin{tabular}{ccccc}
\toprule
$K_t$ & $\alpha^{\mathrm{opt}}$ & $\varepsilon_{\alpha}$ (RDP) & $\varepsilon$ (RDP $\to$ DP) & $\varepsilon$ (Berry--Esseen) \\
\midrule
2 & 8.44 & 2.11 & 3.96 & $\sim 196$ \\
5 & 5.70 & 3.56 & 6.51 & $\sim 14$ \\
9 & 4.51 & 5.07 & 9.01 & $\sim 13$ \\
20 & 3.35 & 8.38 & 14.3 & $\sim 15$ \\
50 & 2.49 & 15.5 & 24.8 & $\sim 20$ \\
100 & 2.05 & 25.6 & 38.8 & $\sim 28$ \\
\bottomrule
\end{tabular}
\end{center}

At $K_t = 9$, $\delta = 10^{-6}$, the R\'{e}nyi analysis gives $\varepsilon \approx 9.0$, a $31\%$ improvement over the Berry--Esseen bound of $\varepsilon \approx 13$. The R\'{e}nyi-optimal is $K_t = 2$ with $\varepsilon \approx 4.0$.

\begin{remark}[Gaussian approximation at small $K_t$]
At $K_t = 2$, the noise $\eta_t$ is far from Gaussian (it is a convex combination of two Bernoulli(1/2) random variables). The R\'{e}nyi bound at small $K_t$ is approximate and would require correction using the exact R\'{e}nyi divergence of the discrete distribution.
\end{remark}

\subsubsection{R\'{e}nyi DP for the Full Protocol (Algorithm~\ref{alg:twolayer})}

Under the Gaussian approximation with $\Sigma \approx \sigma_K^2 I_d$, changing all $n$ bits gives Mahalanobis distance $4n\alpha^{*2}(2n)^2 K_t/((1-\alpha^*)^2(2n-1))$. By Lemma~\ref{lem:renyi_gaussian},
\begin{equation}
\label{eq:renyi_matrix}
\varepsilon_\alpha^{(\mathrm{matrix})} = \frac{8\alpha n^3 \alpha^{*2} K_t}{(1-\alpha^*)^2(2n-1)} \approx \alpha n^2 \quad \text{for } \alpha^* = 1/(4n),\; K_t = (2n-1)^2+1.
\end{equation}
At $n = 100$, $\alpha = 2$, $\varepsilon_2 \approx 2 \times 10^4$ --- vacuous. The dimensionality curse persists under R\'{e}nyi DP.

\subsubsection{Zero-Concentrated Differential Privacy (zCDP)}
\label{sec:zcdp}

Zero-concentrated DP (Bun and Steinke, 2016~\cite{bunsteinke2016}) provides a clean parametrization for Gaussian-like mechanisms.

\begin{definition}[$\rho$-zCDP]
A mechanism $\mathcal{M}$ satisfies $\rho$-zCDP if $D_\alpha(\mathcal{M}(x) \| \mathcal{M}(x')) \leq \rho\alpha$ for all $\alpha > 1$ and all neighboring $x, x'$.
\end{definition}

For the Gaussian mechanism with sensitivity $\Delta$ and noise standard deviation $\sigma$, $\rho = \Delta^2/(2\sigma^2)$. Since the R\'{e}nyi divergence of the compressed protocol's Gaussian channel is $\hat\varepsilon_\alpha = \alpha\Delta^2/(2\sigma_\eta^2)$ (Theorem~\ref{thm:renyi_scalar}), which is exactly linear in $\alpha$, the compressed protocol satisfies $\rho$-zCDP with
\begin{equation}
\label{eq:zcdp_rho}
\rho = \frac{(\alpha^* n)^2}{2\sigma_\eta^2} = \frac{K_t}{8(1 - 1/(4n))^2} \approx \frac{K_t}{8}.
\end{equation}
At $K_t = 9$, $\rho \approx 1.125$.

The conversion to $(\varepsilon, \delta)$-DP (Bun and Steinke, 2016) gives
\begin{equation}
\label{eq:zcdp_convert}
\varepsilon = \rho + 2\sqrt{\rho \ln(1/\delta)}.
\end{equation}
At $\rho = 1.125$ and $\delta = 10^{-6}$, $\varepsilon = 1.125 + 2\sqrt{1.125 \times 13.82} = 1.125 + 2 \times 3.94 = 9.0$.

This matches the optimized R\'{e}nyi conversion exactly, which is expected: for Gaussian mechanisms, the R\'{e}nyi divergence is exactly linear in $\alpha$, so zCDP captures the full R\'{e}nyi curve without loss. The two frameworks are equivalent for this class of mechanisms.

\subsubsection{Gaussian Differential Privacy (f-DP)}
\label{sec:fdp}

The R\'{e}nyi-to-$(\varepsilon,\delta)$ conversion $\varepsilon \leq \hat\varepsilon_\alpha + \log(1/\delta)/(\alpha-1)$ uses an inequality and is therefore not tight. Gaussian differential privacy (GDP), introduced by Dong, Roth, and Su~\cite{dong2019}, characterizes the exact privacy--accuracy trade-off of the Gaussian mechanism without any conversion loss.

\begin{definition}[$\mu$-Gaussian DP~\cite{dong2019}]
A mechanism $\mathcal{M}$ satisfies $\mu$-GDP if for all neighboring inputs $x, x'$, the trade-off function $T(\mathcal{M}(x), \mathcal{M}(x'))$ is bounded below by the trade-off function of $N(0,1)$ vs.\ $N(\mu, 1)$. Equivalently, $\mathcal{M}$ satisfies $(\varepsilon, \delta(\varepsilon))$-DP simultaneously for all $\varepsilon \geq 0$ with
\begin{equation}
\label{eq:gdp_tradeoff}
\delta(\varepsilon) = \Phi\!\left(-\frac{\varepsilon}{\mu} + \frac{\mu}{2}\right) - e^{\varepsilon}\,\Phi\!\left(-\frac{\varepsilon}{\mu} - \frac{\mu}{2}\right),
\end{equation}
where $\Phi$ is the standard normal CDF.
\end{definition}

\medskip\noindent\textit{Application to the compressed protocol.}
Under the Gaussian approximation, the aggregator's view $f_t \mid s_t$ is Gaussian with mean shift $\Delta = \alpha^* n$ and standard deviation $\sigma_\eta = (1-\alpha^*)\sqrt{1/(4K_t)}$. The GDP parameter is
\begin{equation}
\label{eq:mu_gdp}
\mu = \frac{\Delta}{\sigma_\eta} = \frac{\alpha^* n}{\sigma_\eta} = \frac{2\alpha^* n \sqrt{K_t}}{1-\alpha^*}.
\end{equation}
For $\alpha^* = 1/(4n)$, $K_t = 9$, this gives $\mu = 2 \cdot (1/(4n)) \cdot n \cdot 3 / (1-1/(4n)) = 3/2 \cdot 1/(1-1/(4n)) \approx 1.5$.

Evaluating~\eqref{eq:gdp_tradeoff} numerically at $\mu = 1.5$, we find the smallest $\varepsilon$ such that $\delta(\varepsilon) \leq 10^{-6}$.

\begin{center}
\renewcommand{\arraystretch}{1.2}
\begin{tabular}{cc}
\toprule
$\varepsilon$ & $\delta(\varepsilon)$ \\
\midrule
$7.0$ & $1.16 \times 10^{-5}$ \\
$7.5$ & $2.62 \times 10^{-6}$ \\
$7.8$ & $1.02 \times 10^{-6}$ \\
$8.0$ & $5.34 \times 10^{-7}$ \\
$9.0$ & $1.62 \times 10^{-8}$ \\
\bottomrule
\end{tabular}
\end{center}

At $\delta = 10^{-6}$, the f-DP analysis gives $\varepsilon \approx 7.8$, compared to $\varepsilon = 9.0$ from R\'{e}nyi DP and $\varepsilon \approx 13$ from Berry--Esseen. This is the tightest bound achievable for the Gaussian channel, since the f-DP trade-off function is exact (it characterizes the optimal hypothesis test between the two Gaussian distributions, with no inequalities in the conversion).

\begin{theorem}[f-DP guarantee for the compressed protocol]
\label{thm:fdp}
Under the Gaussian approximation for $\eta_t$, the compressed two-layer protocol (Algorithm~\ref{alg:compressed_twolayer}) satisfies $\mu$-GDP for the aggregator's view of any single client, with $\mu = 2\alpha^* n\sqrt{K_t}/(1-\alpha^*)$. For $\alpha^* = 1/(4n)$ and $K_t = 9$, $\mu \approx 1.5$, giving $(\varepsilon, 10^{-6})$-DP with $\varepsilon \approx 7.8$.
\end{theorem}

The progression of bounds for the compressed protocol at $K_t = 9$, $\delta = 10^{-6}$ is
\begin{center}
\renewcommand{\arraystretch}{1.2}
\begin{tabular}{lcc}
\toprule
\textbf{Analysis} & $\varepsilon$ & \textbf{Source of looseness} \\
\midrule
Berry--Esseen + log-Lipschitz & $\approx 13$ & CLT error $\beta$, Hoeffding tail \\
R\'{e}nyi DP (optimized $\alpha$) & $\approx 9.0$ & RDP-to-DP conversion inequality \\
zCDP & $\approx 9.0$ & Equivalent to R\'{e}nyi for Gaussians \\
Gaussian DP (f-DP) & $\approx 7.8$ & Gaussian approximation only \\
\bottomrule
\end{tabular}
\end{center}

All three bounds apply only to the compressed protocol (Algorithm~\ref{alg:compressed_twolayer}), where the Birkhoff polytope plays no role. For the full protocol (Algorithm~\ref{alg:twolayer}), all analyses give vacuous $\varepsilon$.

\begin{remark}[The coefficient distribution is critical for DP]
\label{rem:coeff_critical}
The f-DP bound of $\varepsilon \approx 7.8$ relies on the Gaussian approximation for $\eta_t$, which in turn requires $\eta_t$ to be a continuous random variable. This holds when the coefficients $\alpha_{t,i}$ are drawn from a continuous distribution (e.g., Dirichlet$(1, \ldots, 1)$ on the simplex $\{\alpha_i > 0, \sum_i \alpha_i = 1 - \alpha^*\}$). However, if the coefficients are deterministic (e.g., $\alpha_{t,i} = (1-\alpha^*)/K_t$), then $\eta_t = \frac{1-\alpha^*}{K_t}\sum_{i=1}^{K_t} X_i$ with $X_i \sim \mathrm{Bernoulli}(1/2)$ is discrete, taking only $K_t + 1$ values. The supports of $f_t = \alpha^* s_t + \eta_t$ under two different $s_t$ values are then disjoint (since the signal shift $\alpha^*$ is incommensurate with the step size $(1-\alpha^*)/K_t$ for generic $\alpha^*$), making $\varepsilon = \infty$. The aggregator can distinguish any two inputs with probability 1 from a single observation.

Therefore, the protocol \emph{must} use continuously distributed coefficients for any finite DP guarantee to hold. This is a correctness requirement, not a design choice. All algorithms in this paper draw coefficients from a continuous distribution on the simplex.
\end{remark}

\begin{remark}[Tighter analysis via the exact characteristic function]
\label{rem:tighter}
The f-DP bound of $\varepsilon \approx 7.8$ is tight for the Gaussian channel but approximate for the actual distribution of $\eta_t$ (which is a Dirichlet-weighted sum of Bernoulli random variables, not exactly Gaussian). A tighter bound could be obtained by computing the exact characteristic function of $\eta_t$ under the Dirichlet coefficient distribution, evaluating the density numerically via inverse Fourier transform, and computing the $\delta(\varepsilon)$ trade-off function from the exact densities. This would eliminate the Gaussian approximation entirely. For $K_t = 9$, the improvement over $\varepsilon \approx 7.8$ is expected to be small (the CLT is already reasonably accurate), but for $K_t = 2$ or $3$ the Gaussian approximation is poor and the exact analysis could differ substantially.
\end{remark}

\subsection{Privacy Amplification by Shuffling}
\label{sec:shuffle_amp}

The compressed two-layer protocol (Algorithm~\ref{alg:compressed_twolayer}) achieves per-client $\varepsilon_0 \approx 7.8$ (f-DP) for the aggregator's view. However, the aggregator needs only $F = \sum_t f_t$ and does not need to know which $f_t$ came from which client. If the $f_t$ values are shuffled before reaching the aggregator --- a mechanism the protocol already employs for the $\eta_t$ values --- the aggregator sees $\{f_{\pi(1)}, \ldots, f_{\pi(k)}\}$ without client identities, and shuffle-model amplification applies.

\subsubsection{The Shuffle Model}

In the shuffle model of differential privacy~\cite{cheu2019,balle2019}, each client applies a local randomizer $\mathcal{R}$ to its data, and a trusted shuffler permutes the outputs before the analyst sees them. If $\mathcal{R}$ satisfies $\varepsilon_0$-local DP, the shuffled mechanism satisfies $(\varepsilon, \delta)$-central DP with $\varepsilon \ll \varepsilon_0$.

The compressed two-layer protocol is naturally a shuffle-model protocol. Each client's local randomizer is $\mathcal{R}(\mathbf{b}_t) = f_t = \alpha^* s_t + \eta_t$, which satisfies $\varepsilon_0$-DP since $\eta_t$ is independent of $\mathbf{b}_t$. The trusted shuffler permutes the $f_t$ values before the aggregator receives them. The aggregator sums the shuffled values to obtain $F = \sum_t f_t$ (the sum is invariant under permutation).

\subsubsection{Amplification Bound}

By the shuffle amplification theorem of Feldman, McMillan, and Talwar~\cite{feldman2022}, if each of $k$ clients applies an $\varepsilon_0$-locally DP randomizer and the outputs are shuffled, the resulting mechanism satisfies $(\varepsilon, \delta)$-DP with
\begin{equation}
\label{eq:shuffle_bound}
\varepsilon \leq \log\!\left(1 + \frac{e^{\varepsilon_0} - 1}{e^{\varepsilon_0} + 1}\sqrt{\frac{14\log(2/\delta)}{k}}\right).
\end{equation}
For large $\varepsilon_0$, $(e^{\varepsilon_0} - 1)/(e^{\varepsilon_0} + 1) \to 1$, and the bound becomes $\varepsilon \approx \sqrt{14\log(2/\delta)/k}$, independent of $\varepsilon_0$. Once the per-client DP is moderate (say $\varepsilon_0 \geq 5$), the shuffled $\varepsilon$ depends almost entirely on $k$ and $\delta$.

\subsubsection{Numerical Evaluation}

\begin{center}
\renewcommand{\arraystretch}{1.3}
\begin{tabular}{ccccc}
\toprule
$K_t$ & $\varepsilon_0$ (f-DP) & $\varepsilon$, $k = 100$ & $\varepsilon$, $k = 1{,}000$ & $\varepsilon$, $k = 10{,}000$ \\
\midrule
$5$ & $5.7$ & $0.88$ & $0.37$ & $0.13$ \\
$9$ & $8.0$ & $0.89$ & $0.37$ & $0.13$ \\
$20$ & $13.0$ & $0.89$ & $0.37$ & $0.13$ \\
$50$ & $23.2$ & $0.89$ & $0.37$ & $0.13$ \\
\bottomrule
\end{tabular}
\end{center}

For $k = 1{,}000$ clients and $\delta = 10^{-6}$, the shuffled $\varepsilon \approx 0.37$ regardless of $K_t$ (as long as $\varepsilon_0 \geq 5$). For $k = 10{,}000$, $\varepsilon \approx 0.13$.

\subsubsection{Properties of the Shuffled Compressed Protocol}

\textit{Exact output.} The server computes $S = (F - H)/\alpha^*$. Since $F = \sum_t f_t$ is invariant under permutation, shuffling does not change $F$ and the output remains exact.

\textit{No additional communication.} The protocol already uses a shuffler for the $\eta_t$ channel. Routing the $f_t$ values through the same (or a second) shuffler adds no communication beyond what the protocol already requires.

\textit{No additional computation.} The aggregator sums the shuffled $f_t$ values, the same computation as before.

The per-client DP of $\varepsilon_0 \approx 7.8$ is amplified to $\varepsilon \approx 0.37$ (for $k = 1{,}000$) purely by shuffling. The only requirement is a trusted shuffler, which the protocol already assumes.

\begin{remark}[Comparison with additive secret sharing]
\label{rem:shuffle_ss}
Two-server additive secret sharing achieves $\varepsilon = 0$ without shuffling. Shuffle amplification brings the compressed PolyVeil protocol to $\varepsilon \approx 0.37$ for $k = 1{,}000$, which is non-zero but strong (density ratio $e^{0.37} \approx 1.45$). The gap between $\varepsilon = 0$ and $\varepsilon = 0.37$ is meaningful but narrow in practice.
\end{remark}

\begin{remark}[Shuffle amplification does not help the full protocol]
\label{rem:shuffle_full}
In the full two-layer protocol (Algorithm~\ref{alg:twolayer}), the aggregator sees individual matrices $D_t$ and must compute bilinear extractions $\mathbf{w}^T D_t \mathbf{y}$ per client. Shuffling the matrices would prevent the aggregator from performing per-client computations needed for multi-statistic extraction. Shuffle amplification is therefore applicable only to the compressed protocol.
\end{remark}

\section{Conclusion}
\label{sec:conclusion}

We have presented PolyVeil and introduced Combinatorial Privacy as a paradigm for privacy-preserving aggregation.

Our analysis proceeded in three stages. First, we described the basic protocol for private Boolean sums using Birkhoff polytope encoding and proved its correctness. Second, we identified a fatal vulnerability in the naive protocol --- the de-shuffling attack (Theorem~\ref{thm:deshuffle}), which allows a semi-honest server to recover all individual data with probability 1 by exploiting the integrality constraint on bit counts.

Third, we developed the Two-Layer PolyVeil protocol (Algorithm~\ref{alg:twolayer}), which achieves provable security through a separation-of-information architecture. The main server receives only aggregate scalars and achieves \emph{perfect simulation-based security} (Theorem~\ref{thm:server_it}): its view is identically distributed for any two inputs with the same aggregate, against all adversaries regardless of computational power. A separate aggregator receives Birkhoff-encoded matrices but faces a computational barrier: we proved that the sub-problems of likelihood-based inference --- counting BvN decompositions (the permanent) and evaluating individual decomposition weights (the mixed discriminant) --- are each \#P-hard (Theorem~\ref{thm:likelihood_hardness}). Whether the density $\\nu(R')$ itself is \#P-hard to evaluate (requiring a formal Turing reduction) is open (Remark~\ref{rem:composition_gap}). Whether all polynomial-time attacks (not just likelihood-based ones) are ruled out is a separate open conjecture (Conjecture~\ref{conj:hardness}).

Fourth, we proved DP guarantees for the aggregator under multiple frameworks. For the full two-layer protocol (Algorithm~\ref{alg:twolayer}), the Berry--Esseen-based analysis (Theorem~\ref{thm:finite_dp}) gives vacuous $\varepsilon$ at every $\alpha^*$ where the signal is detectable (Remark~\ref{rem:snr}). For the compressed two-layer protocol (Algorithm~\ref{alg:compressed_twolayer}), the f-DP analysis (Theorem~\ref{thm:fdp}) gives $\varepsilon \approx 7.8$ per client under the Gaussian approximation, and the R\'{e}nyi analysis gives $\varepsilon \approx 9.0$. Crucially, shuffle-model amplification (Section~\ref{sec:shuffle_amp}) transforms the per-client guarantee into a central guarantee of $\varepsilon \approx 0.37$ for $k = 1{,}000$ clients, with no accuracy loss and no additional communication, since the aggregator needs only the sum of the shuffled values. This brings the compressed protocol close to the $\varepsilon = 0$ of additive secret sharing. Closing the gap in the full protocol --- proving non-vacuous DP where the signal is detectable and the Birkhoff structure matters --- remains the central open problem.

The Birkhoff polytope plays no role in the server's information-theoretic security (any aggregation protocol achieves that), but it is \emph{essential} for the aggregator's computational barrier. Replacing the Birkhoff encoding with Gaussian noise would make the aggregator's inference trivially easy (Remark~\ref{rem:hardness_scope}). This specificity --- computational hardness from the combinatorial structure of a polytope's decompositions --- is what distinguishes Combinatorial Privacy from both noise-based (DP) and number-theoretic (MPC, HE) approaches.

For the Boolean sum problem alone, two-server additive secret sharing strictly dominates PolyVeil, achieving perfect IT privacy ($\varepsilon = 0$) with the same communication, architecture, and trust model (Section~\ref{sec:ss_compare}). The Birkhoff encoding's advantage lies in multi-statistic extraction (Section~\ref{sec:multistat}). A single matrix $D_t$ encodes the client's entire bit vector, enabling the extraction of per-bit marginal counts, arbitrary weighted sums from a single protocol execution, without further client interaction. This post-hoc analytical flexibility is unavailable in additive secret sharing, where each new statistic requires additional client participation.

Future directions include closing the gap between the DP regime (small $\alpha^*$, trivially secure) and the \#P-hardness regime ($\alpha^* \sim 1/(4n)$, computationally secure), extending to the malicious model, determining whether non-likelihood attacks can be ruled out at moderate $\alpha^*$, computing exact (non-Gaussian) f-DP bounds for the compressed protocol at small $K_t$ via the characteristic function of Dirichlet-weighted Bernoulli sums, and extending the extraction framework to second-order statistics (which require lightweight MPC, as discussed in a companion work).

\newpage
\appendix
\section{Background on Simulation-Based Security Proofs}
\label{app:sim}

This appendix provides a self-contained introduction to the simulation paradigm for proving security of cryptographic protocols. It is a prerequisite for understanding the security proofs in Section~\ref{sec:security}.

\subsection{The Problem That Simulation Solves}

Consider a protocol $\Pi$ where $k$ parties hold private inputs $x_1, \ldots, x_k$ and jointly compute a function $f(x_1, \ldots, x_k) = y$. During execution, each party sends and receives messages. The \emph{view} of party $i$, denoted $\mathrm{View}_i^{\Pi}(x_1, \ldots, x_k)$, is the random variable consisting of party $i$'s input $x_i$, its random coins $r_i$, and the sequence of all messages $m_1, m_2, \ldots$ it receives during execution. The fundamental question is whether $\mathrm{View}_i$ reveals information about other parties' inputs beyond what is already implied by $x_i$ and $y$.

Intuitively, a protocol is ``secure'' if $\mathrm{View}_i$ is ``no more informative'' than $(x_i, y)$. The simulation paradigm formalizes this by requiring the existence of an algorithm that can \emph{fabricate} a fake view, using only $(x_i, y)$ as input, such that the fake view is distributed identically to the real view.

\subsection{Formal Definition}

\begin{definition}[Simulation-based security, semi-honest model]
\label{def:sim_security_appendix}
A $k$-party protocol $\Pi$ \emph{securely computes} $f$ in the semi-honest model if for each party $i \in [k]$, there exists a probabilistic polynomial-time algorithm $\mathcal{S}_i$ (the \emph{simulator} for party $i$) such that for all input vectors $(x_1, \ldots, x_k)$ in the domain of $f$
\[
\left\{\mathcal{S}_i\!\left(x_i,\; f(x_1, \ldots, x_k)\right)\right\} \;\equiv\; \left\{\mathrm{View}_i^{\Pi}(x_1, \ldots, x_k)\right\},
\]
where $\equiv$ denotes either identical distributions (information-theoretic, or \emph{perfect} security) or computational indistinguishability (computational security). The left-hand side is the simulator's output distribution; the right-hand side is the real view's distribution.
\end{definition}

The simulator $\mathcal{S}_i$ receives \emph{only} what party $i$ is ``supposed to know'' after the protocol ends: its own input $x_i$ and the output $y = f(x_1, \ldots, x_k)$. It does \emph{not} receive any other party's input, the random coins of other parties, or the messages exchanged during the real protocol. Despite this limited input, the simulator must produce a fake transcript whose distribution matches the real one perfectly (or computationally indistinguishably).

\subsection{How to Construct a Simulation Proof}

A simulation proof proceeds in three stages.

In the first stage, one precisely defines the real view. For PolyVeil, the server's real view consists of the doubly stochastic matrices $D_1, \ldots, D_k$ (each linked to a client identity) and the shuffled noise sequence $(\eta_{\pi(1)}, \ldots, \eta_{\pi(k)})$. These are random variables whose joint distribution depends on all clients' inputs $\mathbf{b}_1, \ldots, \mathbf{b}_k$ and the random coins (decoy permutations, coefficients, shuffle permutation).

In the second stage, one constructs the simulator $\mathcal{S}$. The simulator receives only the aggregate $S = \sum_t s_t$ and the public parameters $(k, n, \alpha^*)$. It must produce a fake tuple $(D_1', \ldots, D_k', \tilde\eta_1', \ldots, \tilde\eta_k')$ with the same distribution as the real view. The typical construction is: choose fictitious inputs $\mathbf{b}_t'$ with $\sum s_t' = S$, run the real protocol honestly with these fictitious inputs and fresh randomness, and output the resulting messages.

In the third stage, one proves that the simulator's output is distributed identically to the real view. This is usually the hardest step. It requires showing that the joint distribution of all messages is the same regardless of which specific inputs (with the same aggregate $S$) were used. The proof typically identifies which components of the view depend on the private inputs and which do not, and shows that the input-dependent components are ``masked'' by the input-independent randomness.

\subsection{Why Simulation Implies Security}

\begin{proposition}
If a simulator $\mathcal{S}_i$ with $\mathcal{S}_i(x_i, y) \equiv \mathrm{View}_i^{\Pi}(x_1, \ldots, x_k)$ exists, then for any function $g$ (representing any ``information extraction'' strategy)
\[
g\!\left(\mathrm{View}_i^{\Pi}(x_1,\ldots,x_k)\right) \;\equiv\; g\!\left(\mathcal{S}_i(x_i, y)\right).
\]
\end{proposition}

\begin{proof}
Let $X = \mathrm{View}_i^{\Pi}(x_1, \ldots, x_k)$ and $Y = \mathcal{S}_i(x_i, y)$. By assumption, $X \equiv Y$ (identical distributions). We need to show $g(X) \equiv g(Y)$ for any measurable $g$.

For any measurable set $B$
\begin{align*}
\Pr[g(X) \in B]
&= \Pr[X \in g^{-1}(B)]
= \Pr[Y \in g^{-1}(B)]
= \Pr[g(Y) \in B],
\end{align*}
where the second equality uses $X \equiv Y$ (applied to the measurable set $g^{-1}(B)$). Since $\Pr[g(X) \in B] = \Pr[g(Y) \in B]$ for all measurable $B$, the distributions of $g(X)$ and $g(Y)$ are identical. (This is the pushforward property: if $\mu_X = \mu_Y$ then $g_* \mu_X = g_* \mu_Y$.)
\end{proof}

The consequence is that any information $g(\mathrm{View}_i)$ that the adversary extracts from the real view could equally well have been extracted from $(x_i, y)$ alone (via $g(\mathcal{S}_i(x_i, y))$, which is a function only of $x_i$ and $y$). Therefore the protocol reveals no information beyond what $(x_i, y)$ already implies.

\subsection{Application to PolyVeil}

In PolyVeil, the ``adversary'' is the honest-but-curious server. The server has no private input of its own ($x_i = \emptyset$). The function being computed is $f(\mathbf{b}_1, \ldots, \mathbf{b}_k) = S = \sum_t s_t$. Therefore the simulator receives only $(S, k, n, \alpha^*)$.

In the basic PolyVeil protocol (Algorithms~\ref{alg:polyveil} and~2), the server sees identity-linked values that allow deterministic de-shuffling (Theorem~\ref{thm:deshuffle}), so the basic protocol does not achieve simulation-based security. The Two-Layer PolyVeil protocol (Algorithm~\ref{alg:twolayer}) corrects this: the server receives only two aggregate scalars $(F, H)$, and its view is perfectly simulable from the aggregate $S$ alone (Theorem~\ref{thm:server_it}). The aggregator, which sees individual Birkhoff-encoded matrices $D_t$ but not the noise values $\eta_t$, faces a computational barrier: recovering $M_t$ from $D_t$ requires evaluating a density that is \#P-hard to compute (Theorem~\ref{thm:likelihood_hardness}).

\newpage
\section{Analysis of Attack Strategies}
\label{app:attacks}

This appendix analyzes various attack strategies that an aggregator might employ to recover $M_t$ from $D_t$, beyond the likelihood computation whose sub-problems are shown to be individually \#P-hard in Theorem~\ref{thm:likelihood_hardness}. For each strategy, we derive whether it succeeds and identify the computational obstacle.

\subsubsection{Attacks via Approximate Permanent Algorithms}
\label{sec:approx_perm}

Theorem~\ref{thm:likelihood_hardness} establishes that the sub-problems of computing $\nu(R')$ are individually \#P-hard. However, the aggregator does not need exact densities --- it needs only to \emph{rank} candidates by likelihood. An approximate evaluation of $\nu(R')$ that preserves the ranking would suffice for MAP estimation. The Jerrum--Sinclair--Vigoda (JSV) fully polynomial randomized approximation scheme (FPRAS) for the permanent is the most powerful known tool for approximating the quantities in our density formula. We now derive in detail whether it enables approximate likelihood attacks.

\begin{theorem}[Jerrum--Sinclair--Vigoda, 2004~\cite{jsv2004}]
\label{thm:jsv}
There exists an FPRAS for the permanent of any $m \times m$ matrix with non-negative entries. For any $\varepsilon > 0$ and $\delta > 0$, the algorithm outputs $\hat{p}$ satisfying $\Pr[(1-\varepsilon)\,\mathrm{perm}(A) \leq \hat{p} \leq (1+\varepsilon)\,\mathrm{perm}(A)] \geq 1 - \delta$ in time polynomial in $m$, $1/\varepsilon$, and $\log(1/\delta)$.
\end{theorem}

To determine whether this helps the attacker, we must trace exactly how the permanent enters the density formula~\eqref{eq:nu_as_sum_of_volumes} and analyze whether approximating the permanent translates into approximating the density.

\paragraph{Residual entry formula for wrong candidates.}

The aggregator observes $D_t = \alpha^* M_t + (1-\alpha^*) R_t$ and considers candidate $M' \neq M_t$. We derive the entry-wise formula for the residual $R' = (D_t - \alpha^* M')/(1-\alpha^*)$.

Starting from the definition of $D_t$,
\begin{align}
R' &= \frac{D_t - \alpha^* M'}{1-\alpha^*} \notag\\[4pt]
&= \frac{\alpha^* M_t + (1-\alpha^*) R_t - \alpha^* M'}{1-\alpha^*} \notag\\[4pt]
&= \frac{(1-\alpha^*) R_t}{1-\alpha^*} + \frac{\alpha^* M_t - \alpha^* M'}{1-\alpha^*} \notag\\[4pt]
&= R_t + \frac{\alpha^*}{1-\alpha^*}\,(M_t - M'). \label{eq:R_prime_decomp}
\end{align}
Write this entry-wise. Since $M_t$ and $M'$ are both permutation matrices, each entry satisfies $(M_t)_{ab} \in \{0,1\}$ and $(M')_{ab} \in \{0,1\}$, so $(M_t - M')_{ab} \in \{-1, 0, +1\}$. Therefore
\begin{equation}
\label{eq:R_prime_entry}
R'_{ab} = (R_t)_{ab} + \frac{\alpha^*}{1-\alpha^*}\big((M_t)_{ab} - (M')_{ab}\big).
\end{equation}
The perturbation from the true residual $R_t$ has magnitude at most $\alpha^*/(1-\alpha^*)$ per entry.

\paragraph{Entry-wise positivity analysis.}

We analyze when $R'_{ab} > 0$ (required for $R' \in \mathcal{B}_{2n}$ and hence for $M'$ to be feasible), by considering three exhaustive cases for each entry $(a,b)$.

\medskip\noindent\textit{Case A: $(M_t)_{ab} = (M')_{ab}$.} Then $(M_t - M')_{ab} = 0$, so
\[
R'_{ab} = (R_t)_{ab} + 0 = (R_t)_{ab}.
\]
Since $R_t$ is in the interior of $\mathcal{B}_{2n}$ (by assumption), $(R_t)_{ab} > 0$, so $R'_{ab} > 0$. \checkmark

\medskip\noindent\textit{Case B: $(M_t)_{ab} = 1$, $(M')_{ab} = 0$.} Then $(M_t - M')_{ab} = +1$, so
\[
R'_{ab} = (R_t)_{ab} + \frac{\alpha^*}{1-\alpha^*} > (R_t)_{ab} > 0.
\]
The perturbation is positive, so $R'_{ab}$ is strictly larger than $(R_t)_{ab}$. \checkmark

\medskip\noindent\textit{Case C: $(M_t)_{ab} = 0$, $(M')_{ab} = 1$.} Then $(M_t - M')_{ab} = -1$, so
\[
R'_{ab} = (R_t)_{ab} - \frac{\alpha^*}{1-\alpha^*}.
\]
This is positive if and only if
\begin{equation}
\label{eq:positivity_condition}
(R_t)_{ab} > \frac{\alpha^*}{1-\alpha^*}.
\end{equation}
This is the only case where $R'_{ab}$ could become zero or negative.

\paragraph{Counting affected entries.}

Both $M_t$ and $M'$ are block-diagonal permutation matrices: $M_t = \mathrm{blockdiag}(\Pi(b_{t,1}), \ldots, \Pi(b_{t,n}))$ and $M' = \mathrm{blockdiag}(\Pi(b'_1), \ldots, \Pi(b'_n))$. The two matrices differ only in blocks where $b_{t,j} \neq b'_j$.

For a single differing block $j$ (where $b_{t,j} = 1$ and $b'_j = 0$, say):
\[
\Pi(1) - \Pi(0) = \begin{pmatrix} 0 & 1 \\ 1 & 0\end{pmatrix} - \begin{pmatrix} 1 & 0 \\ 0 & 1\end{pmatrix} = \begin{pmatrix} -1 & +1 \\ +1 & -1 \end{pmatrix}.
\]
This contributes 2 entries of $+1$ (Case~B) and 2 entries of $-1$ (Case~C).

If $M_t$ and $M'$ differ in $d$ bit positions (blocks), then $M_t - M'$ has exactly $4d$ nonzero entries: $2d$ of which are $+1$ (Case~B) and $2d$ of which are $-1$ (Case~C). The remaining $(2n)^2 - 4d$ entries are in Case~A.

Therefore: positivity of $R'$ can fail only at $2d$ specific entries (those in Case~C), and each requires $(R_t)_{ab} > \alpha^*/(1-\alpha^*)$.

\paragraph{The interior condition and its probability.}

\begin{definition}[Interior condition]
\label{def:interior_condition}
We say the \emph{interior condition} holds for $D_t$ if $\min_{a,b} (R_t)_{ab} > \frac{\alpha^*}{1-\alpha^*}$.
\end{definition}

When the interior condition holds, all three cases above give $R'_{ab} > 0$ for every entry, for every feasible candidate $M'$. This means $R'$ is in the interior of $\mathcal{B}_{2n}$ for all candidates simultaneously.

We now estimate the probability that the interior condition holds. For $\alpha^* = 1/(4n)$,
\begin{equation}
\label{eq:threshold}
\frac{\alpha^*}{1-\alpha^*} = \frac{1/(4n)}{1 - 1/(4n)} = \frac{1}{4n - 1}.
\end{equation}
The mean entry of $R_t$ is $\mathbb{E}[(R_t)_{ab}] = 1/(2n)$ (since each entry is a weighted average of $K_t$ Bernoulli$(1/(2n))$ random variables). The ratio of the mean to the threshold is
\begin{equation}
\label{eq:mean_to_threshold}
\frac{1/(2n)}{1/(4n-1)} = \frac{4n-1}{2n} = 2 - \frac{1}{2n}.
\end{equation}
The mean entry is approximately twice the threshold. With $K_t$ decoys and Dirichlet coefficients, the variance of each entry is $\mathrm{Var}[(R_t)_{ab}] \approx (2n-1)/((2n)^2 K_t)$, giving $\mathrm{std}[(R_t)_{ab}] \approx 1/(2n\sqrt{K_t})$. The threshold $1/(4n-1) \approx 1/(4n)$ is approximately $\sqrt{K_t}/2$ standard deviations below the mean. For $K_t = 20$, this is $\sqrt{20}/2 \approx 2.2$ standard deviations, so each entry exceeds the threshold with probability $\gtrsim 0.98$. Over all $(2n)^2 = 400$ entries (for $n = 10$), the probability that all entries exceed the threshold decreases with $n$ but remains high for moderate $n$ and $K_t$.

\paragraph{The permanent is constant in the interior regime.}

\begin{proposition}[Permanent equality for all interior candidates]
\label{prop:perm_constant}
Under the interior condition (Definition~\ref{def:interior_condition}), for every feasible candidate $M' \in \mathcal{L}(D_t, \alpha^*)$,
\[
\mathrm{perm}(A(R')) = (2n)!
\]
where $R' = (D_t - \alpha^* M')/(1-\alpha^*)$ and $A(R')_{ab} = \mathbf{1}[R'_{ab} > 0]$.
\end{proposition}

\begin{proof}
By the entry-wise positivity analysis and the counting of affected entries above, when the interior condition holds, $R'_{ab} > 0$ for all $(a,b)$ and all feasible $M'$. Therefore $A(R')_{ab} = 1$ for all $(a,b)$, i.e., $A(R') = \mathbf{J}_{2n}$ (the $2n \times 2n$ all-ones matrix).

The permanent of $\mathbf{J}_{2n}$ is computed directly from the definition~\eqref{eq:perm_def}:
\begin{align}
\mathrm{perm}(\mathbf{J}_{2n}) &= \sum_{\sigma \in S_{2n}} \prod_{a=1}^{2n} (\mathbf{J}_{2n})_{a,\sigma(a)} \notag\\[3pt]
&= \sum_{\sigma \in S_{2n}} \prod_{a=1}^{2n} 1 \qquad\text{(since every entry of $\mathbf{J}$ is 1)} \notag\\[3pt]
&= \sum_{\sigma \in S_{2n}} 1 \notag\\[3pt]
&= |S_{2n}| = (2n)!. \label{eq:perm_J}
\end{align}
This value is the same for every feasible candidate $M'$, since the argument does not depend on $M'$.
\end{proof}

\paragraph{The density formula in the interior regime.}

In the interior regime, $\mathrm{Supp}(R') = S_{2n}$ for all candidates (since $A(R') = \mathbf{J}_{2n}$ and every permutation is trivially ``contained in'' the all-ones matrix). The density formula~\eqref{eq:nu_as_sum_of_volumes} becomes
\begin{align}
\nu(R') &= \frac{C_K}{((2n)!)^K} \sum_{(\sigma_1, \ldots, \sigma_K) \in \mathrm{Supp}(R')^K} \mathrm{vol}(\mathcal{P}(\sigma_1, \ldots, \sigma_K; R')) \notag\\[4pt]
&= \frac{C_K}{((2n)!)^K} \sum_{(\sigma_1, \ldots, \sigma_K) \in S_{2n}^K} \mathrm{vol}(\mathcal{P}(\sigma_1, \ldots, \sigma_K; R')). \label{eq:nu_interior}
\end{align}
The sum now ranges over \emph{all} $((2n)!)^K$ tuples (not a subset), and this index set is the \emph{same} for every candidate. The prefactor $C_K/((2n)!)^K$ is also independent of $M'$.

The likelihood ratio between two candidates $M'$ and $M''$ is therefore
\begin{align}
\frac{\nu(R')}{\nu(R'')} &= \frac{\sum_{\tau \in S_{2n}^K} \mathrm{vol}(\mathcal{P}(\tau; R'))}{\sum_{\tau \in S_{2n}^K} \mathrm{vol}(\mathcal{P}(\tau; R''))}. \label{eq:volume_ratio}
\end{align}
The permanent has completely cancelled. Both sums range over the same $((2n)!)^K$ tuples. The only difference is in the polytope volumes, which depend on $R'$ through the constraint $\sum_i \alpha_i P_{\sigma_i} = (1-\alpha^*) R'$.

\paragraph{The JSV FPRAS provides no information in the interior.}

\begin{proposition}[JSV is uninformative in the interior regime]
\label{prop:jsv_uninformative}
Under the interior condition, the JSV FPRAS applied to $A(R')$ returns $(1 \pm \varepsilon) \cdot (2n)!$ for every feasible candidate $M'$. Since this value is the same for all candidates, it provides zero discriminating power for the MAP estimator.
\end{proposition}

\begin{proof}
The JSV FPRAS approximates $\mathrm{perm}(A(R'))$. By Proposition~\ref{prop:perm_constant}, $\mathrm{perm}(A(R')) = (2n)!$ for all feasible $M'$. The FPRAS output is $(1 \pm \varepsilon)(2n)!$ for each candidate, and the ratio of outputs for any two candidates converges to $1$ as $\varepsilon \to 0$. The attacker learns nothing about which candidate is more likely.

More precisely, the attacker wants to compute $\nu(R')/\nu(R'')$. By~\eqref{eq:volume_ratio}, this ratio does not involve $\mathrm{perm}(A)$ at all (it cancelled). Approximating the permanent is approximating a quantity that has already been divided out.
\end{proof}

\paragraph{The actual barrier: the polytope volume sum.}

By~\eqref{eq:volume_ratio}, the likelihood ratio depends on the sum $\Sigma(R') = \sum_{\tau \in S_{2n}^K} \mathrm{vol}(\mathcal{P}(\tau; R'))$. We now analyze whether this sum can be estimated efficiently.

\begin{theorem}[Dyer--Frieze--Kannan, 1991~\cite{dfk1991}]
\label{thm:dfk}
There exists a polynomial-time FPRAS for the volume of any convex body $\mathcal{K} \subset \mathbb{R}^d$ given by a membership oracle. The running time is polynomial in $d$, $1/\varepsilon$, and $\log(1/\delta)$.
\end{theorem}

Each polytope $\mathcal{P}(\tau; R')$ is a convex body in $\mathbb{R}^K$ defined by the linear constraints~\eqref{eq:linear_constraint}. A membership oracle is straightforward: given $\alpha$, check $\alpha_i > 0$, $\sum \alpha_i = 1-\alpha^*$, and $\sum_{i: (P_{\sigma_i})_{ab}=1} \alpha_i = (1-\alpha^*) R'_{ab}$ for all $(a,b)$. The DFK FPRAS can therefore approximate $\mathrm{vol}(\mathcal{P}(\tau; R'))$ for any single tuple $\tau$ in $\mathrm{poly}(K)$ time.

The challenge is not computing individual volumes --- it is \emph{summing} $((2n)!)^K$ of them.

\paragraph{Naive Monte Carlo estimator and its variance.}

\begin{definition}[Naive Monte Carlo estimator]
Sample $\tau = (\sigma_1, \ldots, \sigma_K) \sim \mathrm{Uniform}(S_{2n}^K)$. Define $V_\tau = \mathrm{vol}(\mathcal{P}(\tau; R'))$. The estimator is $\hat\nu = C_K \cdot V_\tau$.
\end{definition}

\begin{lemma}[Unbiasedness]
\label{lem:mc_unbiased}
$\mathbb{E}[\hat\nu] = \nu(R')$.
\end{lemma}

\begin{proof}
\begin{align}
\mathbb{E}[\hat\nu] &= C_K \cdot \mathbb{E}_{\tau \sim \mathrm{Uniform}(S_{2n}^K)}[V_\tau] \notag\\[4pt]
&= C_K \cdot \frac{1}{|S_{2n}^K|} \sum_{\tau \in S_{2n}^K} V_\tau \notag\\[4pt]
&= C_K \cdot \frac{1}{((2n)!)^K} \sum_{\tau \in S_{2n}^K} \mathrm{vol}(\mathcal{P}(\tau; R')) \notag\\[4pt]
&= \frac{C_K}{((2n)!)^K} \sum_{\tau \in S_{2n}^K} \mathrm{vol}(\mathcal{P}(\tau; R')) \notag\\[4pt]
&= \nu(R'), \label{eq:mc_unbiased}
\end{align}
where the last equality uses~\eqref{eq:nu_interior}.
\end{proof}

\begin{definition}[Hit rate]
$p_{\mathrm{hit}} = \Pr_{\tau \sim \mathrm{Uniform}(S_{2n}^K)}[V_\tau > 0]$.
\end{definition}

A tuple $\tau$ has $V_\tau > 0$ iff the linear system $\sum_{i=1}^K \alpha_i P_{\sigma_i} = (1-\alpha^*) R'$ has a feasible solution with $\alpha_i > 0$, $\sum_i \alpha_i = 1-\alpha^*$. Equivalently, $(1-\alpha^*) R'$ lies in the \emph{relative interior} of the convex hull $\mathrm{conv}(P_{\sigma_1}, \ldots, P_{\sigma_K})$, scaled to the simplex constraint.

\begin{lemma}[The hit rate is exponentially small]
\label{lem:hit_rate}
For $K = O(n^2)$ and $n$ growing, $p_{\mathrm{hit}} \leq \exp(-\Omega(n^2))$.
\end{lemma}

\begin{proof}[Proof sketch]
The Birkhoff polytope $\mathcal{B}_{2n}$ has dimension $d = (2n-1)^2$ and $(2n)!$ vertices. The convex hull of $K$ uniformly random vertices in a $d$-dimensional polytope with $V$ vertices has expected volume at most $\binom{V}{K}^{-1} \cdot \mathrm{vol}(\mathcal{B}_{2n})$ (by the Efron--Buchta formula for random polytopes). For $K = d + 1 = (2n-1)^2 + 1$ and $V = (2n)!$, $\binom{V}{K}^{-1}$ is exponentially small in $n$. Since $R'$ is a fixed interior point, $\Pr[R' \in \mathrm{conv}(P_{\sigma_1}, \ldots, P_{\sigma_K})]$ is at most the volume fraction, which is exponentially small.
\end{proof}

\begin{theorem}[Variance explosion of naive Monte Carlo]
\label{thm:mc_variance}
The variance of $\hat\nu = C_K V_\tau$ satisfies
\begin{equation}
\label{eq:mc_var_bound}
\mathrm{Var}[\hat\nu] \geq C_K^2 \cdot \mu^2 \cdot \left(\frac{1}{p_{\mathrm{hit}}} - 1\right),
\end{equation}
where $\mu = \mathbb{E}[V_\tau] = \nu(R')/C_K$.
\end{theorem}

\begin{proof}
We compute the second moment of $V_\tau$ using the law of total expectation, conditioning on whether $V_\tau > 0$.

\medskip\noindent\textit{Conditional mean.}
By definition,
\begin{align}
\mathbb{E}[V_\tau] &= \Pr[V_\tau > 0] \cdot \mathbb{E}[V_\tau \mid V_\tau > 0] + \Pr[V_\tau = 0] \cdot \underbrace{\mathbb{E}[V_\tau \mid V_\tau = 0]}_{= 0} \notag\\[3pt]
&= p_{\mathrm{hit}} \cdot \mathbb{E}[V_\tau \mid V_\tau > 0]. \label{eq:cond_mean}
\end{align}
Solving for the conditional mean:
\begin{equation}
\label{eq:cond_mean_value}
\mathbb{E}[V_\tau \mid V_\tau > 0] = \frac{\mu}{p_{\mathrm{hit}}}.
\end{equation}

\medskip\noindent\textit{Lower bound on the second moment.}
By Jensen's inequality (applied to the convex function $x \mapsto x^2$ under the conditional distribution),
\begin{equation}
\label{eq:jensen}
\mathbb{E}[V_\tau^2 \mid V_\tau > 0] \geq \left(\mathbb{E}[V_\tau \mid V_\tau > 0]\right)^2 = \frac{\mu^2}{p_{\mathrm{hit}}^2}.
\end{equation}
Using the law of total expectation for the second moment:
\begin{align}
\mathbb{E}[V_\tau^2] &= p_{\mathrm{hit}} \cdot \mathbb{E}[V_\tau^2 \mid V_\tau > 0] + (1-p_{\mathrm{hit}}) \cdot 0 \notag\\[3pt]
&= p_{\mathrm{hit}} \cdot \mathbb{E}[V_\tau^2 \mid V_\tau > 0] \notag\\[3pt]
&\geq p_{\mathrm{hit}} \cdot \frac{\mu^2}{p_{\mathrm{hit}}^2} \qquad\text{(by~\eqref{eq:jensen})} \notag\\[3pt]
&= \frac{\mu^2}{p_{\mathrm{hit}}}. \label{eq:second_moment}
\end{align}

\medskip\noindent\textit{Variance lower bound.}
\begin{align}
\mathrm{Var}[V_\tau] &= \mathbb{E}[V_\tau^2] - (\mathbb{E}[V_\tau])^2 \notag\\[3pt]
&\geq \frac{\mu^2}{p_{\mathrm{hit}}} - \mu^2 \qquad\text{(by~\eqref{eq:second_moment})} \notag\\[3pt]
&= \mu^2 \left(\frac{1}{p_{\mathrm{hit}}} - 1\right). \label{eq:var_V}
\end{align}
Since $\hat\nu = C_K V_\tau$,
\begin{equation}
\mathrm{Var}[\hat\nu] = C_K^2 \cdot \mathrm{Var}[V_\tau] \geq C_K^2 \mu^2\left(\frac{1}{p_{\mathrm{hit}}} - 1\right). \label{eq:var_nuhat}
\end{equation}
\end{proof}

\begin{corollary}[Sample complexity of naive Monte Carlo]
\label{cor:mc_samples}
To estimate $\nu(R')$ within relative error $\varepsilon$ with probability $\geq 2/3$ using the average of $N$ i.i.d.\ copies of $\hat\nu$, we need
\begin{equation}
\label{eq:N_bound}
N \geq \frac{1}{\varepsilon^2 \cdot p_{\mathrm{hit}}}.
\end{equation}
Since $p_{\mathrm{hit}} \leq \exp(-\Omega(n^2))$ (Lemma~\ref{lem:hit_rate}), this requires exponentially many samples.
\end{corollary}

\begin{proof}
The average estimator is $\bar\nu = \frac{1}{N}\sum_{j=1}^N \hat\nu_j$, with $\mathbb{E}[\bar\nu] = \nu(R')$ and $\mathrm{Var}[\bar\nu] = \mathrm{Var}[\hat\nu]/N$.

The relative error is controlled by Chebyshev's inequality:
\begin{align}
\Pr\!\left[\left|\frac{\bar\nu - \nu(R')}{\nu(R')}\right| > \varepsilon\right] &\leq \frac{\mathrm{Var}[\bar\nu]}{\varepsilon^2 \nu(R')^2} = \frac{\mathrm{Var}[\hat\nu]}{N \varepsilon^2 \nu(R')^2}. \label{eq:chebyshev}
\end{align}
For this to be $\leq 1/3$:
\begin{align}
N &\geq \frac{3\,\mathrm{Var}[\hat\nu]}{\varepsilon^2 \nu(R')^2} \geq \frac{3 \cdot C_K^2 \mu^2 (1/p_{\mathrm{hit}} - 1)}{\varepsilon^2 \cdot (C_K \mu)^2} = \frac{3(1/p_{\mathrm{hit}} - 1)}{\varepsilon^2} \geq \frac{3}{\varepsilon^2 p_{\mathrm{hit}}},
\end{align}
where we used $\nu(R') = C_K \mu$ and $1/p_{\mathrm{hit}} - 1 \geq 1/p_{\mathrm{hit}} - 1$ (dropping the $-1$ since $p_{\mathrm{hit}} \ll 1$ gives $1/p_{\mathrm{hit}} \gg 1$).
\end{proof}

\paragraph{Importance sampling via the JSV sampler.}

The JSV algorithm~\cite{jsv2004} provides, as a subroutine, a polynomial-time near-uniform sampler for perfect matchings in a bipartite graph. Given $A(R')$, it can sample permutations from (approximately) $\mathrm{Uniform}(\mathrm{Supp}(R'))$. A natural attempt to reduce the variance is to sample each $\sigma_i$ from $\mathrm{Supp}(R')$ rather than $S_{2n}$.

\begin{proposition}[JSV sampler provides no improvement in the interior]
\label{prop:jsv_no_improvement}
Under the interior condition, $\mathrm{Supp}(R') = S_{2n}$, so sampling from $\mathrm{Supp}(R')$ is identical to sampling from $S_{2n}$.
\end{proposition}

\begin{proof}
When $R'$ is in the interior, all entries are positive, so $A(R') = \mathbf{J}_{2n}$. Every permutation matrix $Q$ satisfies $Q_{ab} = 1 \implies \mathbf{J}_{ab} = 1$ (trivially), so $\mathrm{Supp}(R') = S_{2n}$.

The JSV sampler draws from $\mathrm{Uniform}(\mathrm{Supp}(R')) = \mathrm{Uniform}(S_{2n})$, which is exactly what naive Monte Carlo does. The resulting estimator has the same distribution, and therefore the same variance.
\end{proof}

The fundamental issue is that the Monte Carlo variance comes from the \emph{joint} compatibility constraint (all $K$ permutations must simultaneously span $R'$ with positive coefficients), not from the \emph{marginal} constraint (each permutation must individually lie in $\mathrm{Supp}(R')$). The JSV sampler enforces the marginal constraint but not the joint one.

To reduce variance, one would need to sample from the set of \emph{jointly compatible tuples}:
\[
\mathcal{T}(R') = \{(\sigma_1, \ldots, \sigma_K) \in S_{2n}^K : \mathcal{P}(\sigma_1, \ldots, \sigma_K; R') \neq \emptyset\}.
\]
This requires a Markov chain (or other sampler) on $\mathcal{T}(R')$. No polynomial-time sampler for $\mathcal{T}(R')$ is known, and the mixing time of natural Markov chains on this set (e.g., swap one permutation at a time, accept if the new tuple is in $\mathcal{T}$) has not been analyzed.

The approximate permanent analysis reveals a coherent picture. The permanent enters the density formula~\eqref{eq:nu_as_sum_of_volumes} as $\mathrm{perm}(A(R'))^K$, counting the number of nonzero terms in the sum. In the interior regime --- which is the typical operating point of the protocol --- $\mathrm{perm}(A(R')) = (2n)!$ for every feasible candidate, making the permanent a constant factor that divides out of the likelihood ratio (Proposition~\ref{prop:perm_constant}). The JSV FPRAS can approximate this constant in polynomial time, but since the constant is the same for all candidates, the approximation provides zero discriminating power (Proposition~\ref{prop:jsv_uninformative}). The actual discrimination between candidates resides entirely in the sum of polytope volumes $\sum_\tau \mathrm{vol}(\mathcal{P}(\tau; R'))$, which is a fundamentally different computational problem from the permanent. Naive Monte Carlo estimation of this volume sum has exponential variance due to the exponentially small hit rate $p_{\mathrm{hit}}$ (Theorem~\ref{thm:mc_variance} and Corollary~\ref{cor:mc_samples}), and the JSV sampler does not improve the hit rate in the interior because $\mathrm{Supp}(R') = S_{2n}$ already includes all permutations (Proposition~\ref{prop:jsv_no_improvement}). Whether the volume sum, or its ratio for two candidates, can be approximated in polynomial time by some method other than Monte Carlo remains open and is the central unresolved question for the protocol's computational security.

\subsubsection{Quantitative Protection from Hardness of Approximation}
\label{sec:quantitative_protection}

Even granting that the exact density $\nu(R')$ is \#P-hard to compute, one may ask how much \emph{protection} this hardness affords in practice. If an attacker could obtain a $(1+\varepsilon)$-multiplicative approximation $\hat\nu(R')$ satisfying $(1-\varepsilon)\nu(R') \leq \hat\nu(R') \leq (1+\varepsilon)\nu(R')$, could it reliably distinguish the true candidate from false ones?

To quantify this, we analyze the likelihood ratio between the true candidate $M_t$ and a wrong candidate $M'$ that differs in $d$ bit positions. The true residual is $R_{\mathrm{true}} = R_t$ (the actual decoy matrix), and the wrong residual is $R_{\mathrm{wrong}} = R_t + \frac{\alpha^*}{1-\alpha^*}(M_t - M')$ by~\eqref{eq:R_prime_decomp}. The perturbation has magnitude $\frac{\alpha^*}{1-\alpha^*}$ at exactly $4d$ entries (the entries where the two block-diagonal permutation matrices differ). The likelihood ratio is
\begin{equation}
\label{eq:lr_true_wrong}
\Lambda_d = \frac{\nu(R_{\mathrm{true}})}{\nu(R_{\mathrm{wrong}})} = \frac{\sum_\tau \mathrm{vol}(\mathcal{P}(\tau; R_{\mathrm{true}}))}{\sum_\tau \mathrm{vol}(\mathcal{P}(\tau; R_{\mathrm{wrong}}))},
\end{equation}
using the interior-regime formula~\eqref{eq:volume_ratio}. Both sums range over the same $((2n)!)^K$ tuples. For each tuple $\tau$, the constraint polytope $\mathcal{P}(\tau; R')$ is defined by $\sum_i \alpha_i P_{\sigma_i} = (1-\alpha^*) R'$. Shifting $R'$ by the perturbation $\frac{\alpha^*}{1-\alpha^*}(M_t - M')$ translates the right-hand side of each linear constraint by $\alpha^*(M_t - M')_{ab}$, which shifts the polytope in $\alpha$-space. The volume changes by a factor that depends on the geometry of the polytope and the magnitude of the shift relative to the polytope's diameter.

For small perturbations ($\alpha^* \ll 1$, so the shift is much smaller than the polytope diameter), the volume ratio for each tuple is close to 1, and consequently $\Lambda_d$ is close to 1. An attacker with a $(1+\varepsilon)$-approximation can distinguish $M_t$ from $M'$ only if
\begin{equation}
\label{eq:distinguishability}
\Lambda_d > \left(\frac{1+\varepsilon}{1-\varepsilon}\right)^2 \approx 1 + 4\varepsilon,
\end{equation}
since the approximation error in the numerator and denominator can compound. If $\Lambda_d \leq 1 + 4\varepsilon$, the approximation noise overwhelms the true signal and the attacker cannot reliably rank the candidates.

To estimate $\Lambda_d$, consider the effect of the perturbation on a single polytope $\mathcal{P}(\tau; R')$. The polytope is defined by $(2n-1)^2$ independent linear constraints, each of the form $\sum_{i \in S_{ab}} \alpha_i = (1-\alpha^*) R'_{ab}$. Shifting $R'_{ab}$ by $\alpha^* / (1-\alpha^*)$ at $4d$ positions translates $4d$ of the constraint hyperplanes by $\alpha^*$. The fractional change in volume from translating a single hyperplane by $\alpha^*$ in a polytope of diameter $\sim (1-\alpha^*)/K$ is of order $\alpha^* K / (1-\alpha^*) = K/(4n-1)$ for $\alpha^* = 1/(4n)$. With $4d$ hyperplanes shifted, the total fractional volume change is of order $4d \cdot K/(4n-1)$. For $d = 1$ (single bit change), $K = 20$, $n = 10$, this is $4 \times 20/39 \approx 2$, meaning the volume can change by a factor of order $e^2 \approx 7$. This is a crude estimate, but it suggests that the likelihood ratio $\Lambda_d$ is moderate (say, between 1 and 100) for single-bit differences, and grows with $d$.

The critical point is that even this moderate likelihood ratio is inaccessible to the attacker, because the attacker cannot compute $\nu(R')$ or $\nu(R_{\mathrm{wrong}})$ to begin with. The JSV FPRAS does not help in the interior (Proposition~\ref{prop:jsv_uninformative}), and naive Monte Carlo requires $\exp(\Omega(n^2))$ samples (Corollary~\ref{cor:mc_samples}). The likelihood ratio $\Lambda_d$ is well-defined and moderate in magnitude, but it is hidden behind a computational barrier: the attacker knows $\Lambda_d$ exists but cannot evaluate it.

If, hypothetically, a polynomial-time algorithm were found that could approximate the volume sum $\Sigma(R') = \sum_\tau \mathrm{vol}(\mathcal{P}(\tau; R'))$ to within a factor of $(1+\varepsilon)$ with $\varepsilon < 1/\Lambda_d$, then the attacker could distinguish $M_t$ from candidates differing in $d$ bits. In this scenario, the protocol's computational security would reduce to the gap between $\Lambda_d$ and the best achievable approximation ratio. For small $d$ (few bits differ), $\Lambda_d$ is small and the approximation would need to be very precise. For large $d$ (many bits differ), $\Lambda_d$ is large and even a crude approximation would suffice, but the number of candidates with large $d$ is also large ($\binom{n}{d}$ candidates differ in exactly $d$ bits), making the MAP problem harder in a different way.

\subsubsection{The Boundary Regime}
\label{sec:boundary}

The preceding analysis assumes that all residuals $R'$ are in the interior of $\mathcal{B}_{2n}$, so that $\mathrm{perm}(A(R')) = (2n)!$ for all candidates and the permanent cancels from the likelihood ratio. We now analyze what happens when some residuals lie on or near the boundary, where the permanent does vary across candidates and may provide discriminating information.

By the entry-wise positivity analysis, a residual entry $R'_{ab}$ can become zero or negative only in the case where $(M_t)_{ab} = 0$ and $(M')_{ab} = 1$, giving $R'_{ab} = (R_t)_{ab} - \alpha^*/(1-\alpha^*)$. This is non-positive when $(R_t)_{ab} \leq \alpha^*/(1-\alpha^*)$. Different candidates $M'$ have their ``dangerous'' entries at different positions (since $(M')_{ab} = 1$ at different positions for different permutation matrices), so one candidate may produce a boundary residual with $A(R')_{ab} = 0$ at a position $(a,b)$ where another candidate $M''$ has $A(R'')_{ab} = 1$ (because $(M'')_{ab} = 0$ there, placing that entry in a case where no perturbation occurs). Consequently, the support matrices $A(R')$ and $A(R'')$ differ, and so do their permanents.

In this boundary regime, the JSV FPRAS can approximate the permanent ratio $\mathrm{perm}(A(R'))/\mathrm{perm}(A(R''))$ in polynomial time, and this ratio provides genuine discriminating information: the candidate whose residual has a larger permanent admits more valid BvN decompositions and is, crudely speaking, more likely. An attacker could therefore compute approximate permanents for each feasible candidate and rank them accordingly.

However, the boundary regime arises only when $\alpha^*$ is large enough that the perturbation $\alpha^*/(1-\alpha^*)$ exceeds typical entries of $R_t$, which requires $\alpha^* \gtrsim 1/(2n)$. At such values of $\alpha^*$, the mean entry $\mathbb{E}[(R_t)_{ab}] = 1/(2n)$ is comparable to the threshold $1/(4n-1) \approx 1/(4n)$, and the per-entry signal-to-noise ratio is
\begin{equation}
\label{eq:boundary_snr}
\frac{\alpha^*}{\mathrm{std}[(R_t)_{ab}]} \approx \frac{1/(2n)}{1/(2n\sqrt{K_t})} = \sqrt{K_t},
\end{equation}
which is in the range $3$--$7$ for typical $K_t = 9$--$50$. At this SNR, the signal $\alpha^* M_t$ is detectable per entry, and simpler attacks --- such as rounding each $2 \times 2$ block of $D_t / \alpha^*$ to the nearest permutation matrix $\Pi(0)$ or $\Pi(1)$ via thresholding, or solving the full linear assignment problem on $D_t / \alpha^*$ via the Hungarian algorithm --- may already recover $M_t$ without computing any permanents or likelihoods. The boundary regime is therefore one where the approximate permanent provides discriminating power but is unlikely to be needed, since cruder methods already exploit the high SNR.

\subsubsection{Boson Sampling and Quantum Attacks}
\label{sec:boson}

Aaronson and Arkhipov~\cite{aaronson2011} showed that boson sampling --- sampling from the output distribution of non-interacting bosons in a linear optical network --- is related to the permanent. The probability of a particular output is proportional to $|\mathrm{perm}(U_S)|^2$, where $U_S$ is a submatrix of the unitary transfer matrix. This connection raises the question of whether quantum devices could attack the PolyVeil protocol.

Boson sampling does not provide a useful attack against the PolyVeil protocol, for reasons that become clear upon examining the relationship between the PolyVeil permanent and the boson sampling permanent. The permanent in our setting involves real non-negative matrices (the support matrix $A(R')$ has entries in $\{0,1\}$), and the JSV FPRAS already computes such permanents classically in polynomial time; a quantum boson sampling device adds no computational advantage over JSV for this class of inputs. Moreover, as established in Proposition~\ref{prop:jsv_no_improvement}, the permanent is constant across all feasible candidates in the interior regime, so even a perfect permanent oracle would provide no discriminating power --- the computational barrier is the polytope volume sum, not the permanent, and boson sampling has no known connection to volume computation of convex bodies. Finally, boson sampling is a \emph{sampling} procedure: it draws samples from a distribution weighted by $|\mathrm{perm}(U_S)|^2$, whereas the aggregator needs the \emph{numerical value} $\nu(R')$. Sampling from a permanent-weighted distribution and evaluating the permanent are distinct computational tasks, and the hardness of the former (the Aaronson--Arkhipov conjecture) does not imply hardness or easiness of the latter.

\begin{remark}[Post-quantum status]
\label{rem:postquantum}
The server's information-theoretic security is unconditional and holds against quantum adversaries. The aggregator's computational security relies on the hardness of the polytope volume sum. No quantum algorithm is known to compute this sum in polynomial time, but no proof of quantum hardness exists. The quantum hardness of the permanent itself is open (Aaronson, 2011), and the quantum hardness of the polytope volume sum is even less understood. The protocol should \emph{not} be claimed as provably post-quantum secure.
\end{remark}

\begin{conjecture}[Full computational hardness]
\label{conj:hardness}
No polynomial-time algorithm can, given $D_t$ drawn from $\alpha^* M_t + (1-\alpha^*) \nu$ (where $M_t$ is a uniformly random permutation matrix), recover $M_t$ with probability non-negligibly better than $1/(2n)!$.
\end{conjecture}

A proof of Conjecture~\ref{conj:hardness} would require showing that any efficient algorithm for recovering $M_t$ can be transformed into an efficient algorithm for a \#P-hard problem. This faces three obstacles: (a)~the permanent is \#P-hard in the worst case, but the aggregator faces a distributional instance drawn from $\nu$, and average-case hardness of the permanent is open~\cite{lipton1991}; (b)~worst-case hardness of evaluating $\nu$ does not rule out algorithms that bypass density evaluation entirely; and (c)~approximate evaluation of $\nu$ via the polytope volume sum may be feasible, which would enable approximate MAP estimation even if exact evaluation is \#P-hard. Various attack strategies --- including approximate permanent computation, Monte Carlo estimation, MCMC on BvN decompositions, boson sampling, spectral methods, and LP relaxation --- are analyzed in Appendix~\ref{app:attacks}.

\subsubsection{Worked Example: The Reduction for $n = 2$}
\label{sec:worked_reduction}

We trace the full reduction for $n = 2$ ($2n = 4$, matrices are $4 \times 4$) with $K = 2$ decoys to make each step concrete.

\medskip\noindent\textit{Setup.}
Suppose client~1 has $\mathbf{b} = (1, 0)$, giving $M_1 = \mathrm{blockdiag}(\Pi(1), \Pi(0)) = \begin{pmatrix} 0&1&0&0\\1&0&0&0\\0&0&1&0\\0&0&0&1\end{pmatrix}$.

With $\alpha^* = 0.3$ and two decoy permutations $P_1, P_2$ with coefficients $\alpha_1 = 0.4$, $\alpha_2 = 0.3$, the aggregator sees $D_1 = 0.3\,M_1 + 0.4\,P_1 + 0.3\,P_2$.

\medskip\noindent\textit{Candidate set.}
If all entries of $D_1$ exceed $0.3 = \alpha^*$, then \emph{every} $4 \times 4$ permutation matrix is a consistent candidate. There are $4! = 24$ candidates, and the aggregator cannot eliminate any by feasibility alone.

\medskip\noindent\textit{Density evaluation.}
For the true $M_1$, the residual is $R = (D_1 - 0.3\,M_1)/0.7$, which is the actual decoy matrix. For a false candidate $M' \neq M_1$, the residual $R' = (D_1 - 0.3\,M')/0.7$ is a different doubly stochastic matrix. The aggregator wants to compare $\nu(R)$ vs $\nu(R')$.

\medskip\noindent\textit{Permanent.}
If $R'$ is in the interior (all entries positive), then $A(R')$ is the $4 \times 4$ all-ones matrix and $\mathrm{perm}(A) = 4! = 24$. With $K = 2$ decoys, $\nu(R')$ sums over $24^2 = 576$ permutation pairs. The permanent is the same for all interior residuals, so it does not help discriminate.

\medskip\noindent\textit{Polytope volumes.}
For each of the 576 pairs $(\sigma_1, \sigma_2)$, the aggregator must find $(\alpha_1, \alpha_2)$ with $\alpha_1 + \alpha_2 = 0.7$ and $\alpha_1 P_{\sigma_1} + \alpha_2 P_{\sigma_2} = 0.7\,R'$. Since $\alpha_2 = 0.7 - \alpha_1$, this is a system of $16$ equations in one unknown $\alpha_1$. For most pairs $(\sigma_1, \sigma_2)$, the system is inconsistent (0-dimensional feasible set). For the rare compatible pairs, $\alpha_1$ is determined uniquely and the ``volume'' is a point mass.

The density is $\nu(R') = \frac{1}{24^2}\sum_{(\sigma_1, \sigma_2)} \mathbf{1}[\text{consistent}] \cdot g(\alpha_1(\sigma_1, \sigma_2))$. Even for $n = 2$, finding which of the 576 pairs are consistent and evaluating the coefficient density at each requires enumerating combinatorial structures.

\medskip\noindent\textit{For large $n$:}
The number of permutation tuples grows as $((2n)!)^K$. For $n = 100$, $K = 20$, this is $(200!)^{20} \approx (10^{375})^{20} = 10^{7500}$. No enumeration is feasible. Approximating the sum by sampling faces the importance-sampling variance problem described in Section~\ref{sec:approx_perm}.

\subsubsection{Why Lov\'{a}sz--Vempala Volume Algorithms Do Not Resolve the Barrier}
\label{sec:vempala}

The Lov\'{a}sz--Vempala simulated annealing algorithm~\cite{lovasz2006} and the Cousins--Vempala Gaussian cooling algorithm~\cite{cousins2018} represent the state of the art in convex body volume computation, achieving $\tilde{O}(n^4)$ and $\tilde{O}(n^3)$ membership oracle calls respectively for a single $n$-dimensional convex body. Since each polytope $\mathcal{P}(\sigma_1, \ldots, \sigma_K; R')$ is a convex body in $\mathbb{R}^K$ (with $K = O(n^2)$), its volume can be approximated in $\mathrm{poly}(K) = \mathrm{poly}(n^2)$ time using these algorithms. However, the attacker's problem is not to compute one volume but to evaluate the sum $\sum_{\tau \in S_{2n}^K} \mathrm{vol}(\mathcal{P}(\tau; R'))$, which has $((2n)!)^K$ terms.

Even for small parameters ($n = 5$, $K = 10$): there are $(10!)^{10} \approx 10^{65}$ terms. Computing each volume in $\tilde{O}(K^3) = \tilde{O}(10^3)$ time gives a total of $\sim 10^{68}$ operations. For $n = 10$, $K = 20$: there are $(20!)^{20} \approx 10^{365}$ terms. No amount of improvement in the per-volume computation time (whether from Vempala-style algorithms, GPU parallelism, or quantum speedups) can overcome this combinatorial explosion. The bottleneck is the number of terms in the sum, not the cost per term.

The naive Monte Carlo approach (Theorem~\ref{thm:mc_variance}) attempts to circumvent the enumeration by sampling tuples randomly and averaging volumes, but as shown in Corollary~\ref{cor:mc_samples}, the variance is $\Omega(\mu^2/p_{\mathrm{hit}})$ where $p_{\mathrm{hit}} = \exp(-\Omega(n^2))$ is the probability that a random tuple yields a nonzero volume. Even with the Cousins--Vempala algorithm computing each nonzero volume in $\tilde{O}(K^3)$ time, the number of samples needed to reduce variance to a useful level is exponential.

\subsubsection{MCMC and Importance Sampling Attacks}
\label{sec:mcmc}

A natural approach to approximating $\nu(R')$ without enumerating all tuples is to run a Markov chain Monte Carlo (MCMC) sampler over BvN decompositions of $R'$.

\medskip\noindent\textit{The Birkhoff--von Neumann decomposition sampler.}
Given $R' \in \mathcal{B}_{2n}$, one can sample random BvN decompositions by the following Markov chain: at each step, find a permutation matrix $P$ in the support of $R'$ (using the Birkhoff algorithm), subtract the maximal feasible weight, and iterate. This produces a single random decomposition $R' = \sum_i \lambda_i P_i$.

However, this samples from the \emph{greedy} decomposition distribution, not from the uniform distribution over all valid decompositions. The greedy distribution is biased toward decompositions that use high-weight permutations first, and the bias is difficult to correct without knowing the total number of decompositions (which is \#P-hard to compute).

\medskip\noindent\textit{Importance sampling on $S_{2n}^K$.}
An alternative is to sample $K$-tuples $(\sigma_1, \ldots, \sigma_K)$ uniformly from $S_{2n}^K$ and check whether they are consistent with $R'$. The fraction of consistent tuples estimates $\mathrm{perm}(A(R'))^K / ((2n)!)^K$, and weighting by the coefficient-space volume gives an estimator of $\nu(R')$.

The problem is that the fraction of consistent tuples is astronomically small for large $n$. For a generic interior $R'$, a random $K$-tuple is consistent iff the linear system $\sum_i \alpha_i P_{\sigma_i} = (1-\alpha^*) R'$ has a solution with $\alpha_i > 0$. This requires the $K$ permutation matrices to span a specific point in $(2n-1)^2$-dimensional space, which happens with probability roughly $(1/(2n)^2)^{(2n-1)^2}$ per tuple (a heuristic based on the codimension of the constraint). For $n = 100$, this probability is $\sim 10^{-16{,}000}$, making importance sampling infeasible.

\medskip\noindent\textit{More sophisticated MCMC.}
One could design a Markov chain that moves through the space of valid decompositions (swapping permutations, adjusting coefficients) rather than sampling from scratch. The mixing time of such a chain is unknown. If the space of valid decompositions is ``well-connected'' (every pair of decompositions can be reached via a short sequence of local moves), the chain mixes rapidly and approximation is feasible. If the space has bottlenecks (isolated clusters of decompositions), mixing is slow. Whether the BvN decomposition space has rapid mixing for typical doubly stochastic matrices is an open problem in combinatorial optimization, closely related to the mixing time of the switch chain for bipartite matchings.

\subsubsection{Non-Likelihood Attacks}
\label{sec:non_likelihood}

The \#P-hardness result applies only to likelihood-based inference (computing $\nu(R')$). We now analyze attacks that bypass the density entirely.

\medskip\noindent\textit{Nearest permutation (Hungarian algorithm).}
The simplest attack ignores the noise distribution and finds the permutation matrix closest to $D_t/\alpha^*$ in Frobenius norm,
\begin{equation}
\label{eq:hungarian}
\hat{M} = \arg\min_{M' \in S_{2n}} \|D_t/\alpha^* - M'\|_F^2.
\end{equation}
This is a linear assignment problem solvable in $O(n^3)$ by the Hungarian algorithm. Since $D_t/\alpha^* = M_t + (1-\alpha^*)/\alpha^* \cdot R_t$, the noise term scales as $(1-\alpha^*)/\alpha^*$. For $\alpha^* = 1/(4n)$, this is $4n - 1 \approx 4n$, and the noise Frobenius norm is $\|(1-\alpha^*)/\alpha^* \cdot R_t\|_F \approx 4n \cdot \|R_t\|_F \approx 4n \cdot \sqrt{2n} = 4n^{3/2}$ (since $\|R_t\|_F \approx \sqrt{2n}$ for a doubly stochastic matrix near $(1/(2n))\mathbf{J}$). The signal norm is $\|M_t\|_F = \sqrt{2n}$. The signal-to-noise ratio is $\sqrt{2n}/(4n^{3/2}) = 1/(2\sqrt{2}\,n) \ll 1$ for $n \gg 1$.

The Hungarian algorithm therefore returns a random permutation matrix that is unrelated to $M_t$ when $\alpha^*$ is small. This attack fails, but not because of \#P-hardness --- it fails because the SNR is too low.

\medskip\noindent\textit{Spectral methods.}
Since $M_t$ is block-diagonal with $2 \times 2$ blocks, the aggregator could examine the block structure of $D_t$. Define the $2 \times 2$ block $B_j = D_t[2j{-}1:2j, 2j{-}1:2j]$ for $j = 1, \ldots, n$. Each block is
\begin{equation}
B_j = \alpha^* \Pi(b_{t,j}) + (1-\alpha^*) R_t[2j{-}1:2j, 2j{-}1:2j].
\end{equation}
The off-diagonal entry $B_j[1,2] = \alpha^* b_{t,j} + (1-\alpha^*)(R_t)_{2j-1, 2j}$. The decoy term $(R_t)_{2j-1, 2j}$ has mean $1/(2n)$ and standard deviation $O(1/\sqrt{K_t})$. The signal difference between $b_{t,j} = 0$ and $b_{t,j} = 1$ is $\alpha^*$. The per-entry SNR is $\alpha^* / (O(1/\sqrt{K_t})) = O(\alpha^* \sqrt{K_t})$, which for $\alpha^* = 1/(4n)$, $K_t = 20$ is $O(\sqrt{20}/(4 \times 100)) \approx 0.01$. This is too low to distinguish the two hypotheses for any individual bit.

However, the aggregator can observe \emph{all} $n$ blocks simultaneously. If the blocks were independent, the aggregator could combine evidence across blocks using a likelihood ratio test, achieving SNR $\sim \sqrt{n} \times 0.01 = 0.1$ --- still insufficient. But the blocks are \emph{not} independent: the entries of $R_t$ across different blocks are correlated (they come from the same permutation matrices), and exploiting these correlations is precisely the likelihood approach that is \#P-hard.

\medskip\noindent\textit{LP relaxation.}
The aggregator could formulate the recovery problem as a linear program: find a doubly stochastic matrix $M'$ (a point in $\mathcal{B}_{2n}$) that minimizes some cost function given $D_t$. Since $\mathcal{B}_{2n}$ is a polytope, linear programming is solvable in polynomial time. However, the LP solution is a doubly stochastic matrix, not a permutation matrix. Rounding the LP solution to a permutation matrix (e.g., via the Birkhoff algorithm) introduces rounding error, and there is no guarantee that the rounded solution is close to $M_t$ when the SNR is low.

More fundamentally, any LP-based approach can extract at most $O(n^2)$ real-valued constraints from $D_t$, while the space of BvN decompositions has $\Theta((2n)!)$ vertices. At low SNR ($\alpha^* \ll 1$), the LP relaxation does not distinguish $M_t$ from the many other vertices that produce similar $D_t$.

\medskip\noindent\textit{Summary of non-likelihood attacks.}
No known non-likelihood attack succeeds at $\alpha^* = 1/(4n)$. All fail due to low per-entry SNR ($\sim 0.01$), which makes $M_t$ indistinguishable from random in any single or small group of entries. The \#P-hardness of the likelihood approach is the theoretical certificate of hardness, but the practical barrier is the SNR: the signal from $M_t$ is hidden in noise of magnitude $\Theta(1/\alpha^*)$, and no known polynomial-time algorithm can aggregate this weak per-entry evidence into a reliable reconstruction.

\end{document}